\documentclass[preprint, 3p,12pt]{elsarticle}
\usepackage{amssymb}
\usepackage{amsmath}
\usepackage{yhmath}
\usepackage{subcaption} 

\usepackage{amsmath}
\usepackage{graphicx}

\usepackage{bbm}
\usepackage{graphicx}
\usepackage{hyperref}
\usepackage{cancel}
\hypersetup{
    colorlinks=true,
    linkcolor=blue,
    filecolor=magenta,      
    urlcolor=cyan,
    pdftitle={Overleaf Example},
    pdfpagemode=FullScreen,
    }
\usepackage{enumitem}
\usepackage[normalem]{ulem}

\makeatletter
\def\namedlabel#1#2{\begingroup
    #2%
    \def\@currentlabel{#2}%
    \phantomsection\label{#1}\endgroup
}
\makeatother

\usepackage{float} 
\usepackage[ruled]{algorithm2e}
\usepackage[noend]{algorithmic}
\usepackage{yhmath}

\def\cF{{\mathcal F}}

\def\sF{{\mathsf{F}}}

\def \st{\mathsf t}

\def \cN {\mathcal N}

\def\sfi{\mathsf{i}}

\def \x {\mathbf x}

\def \y {\mathbf y}

\def \I{\mathbf I}
\def \0{\mathbf 0}

\def \cF{\mathcal F}

\def \st{\mathsf{t}}

\def \bP{\mathbb P}

\def \Q{\mathrm Q}
\def \bN{\mathbb N}
\def\md{{\mathrm d}}

\def\b0{\mathbf 0}
\def \Rb {\mathbb R}

\def \bR {\mathbb R}
\def \D{D}
\def \S{\mathrm{S}}
\def \U{\mathrm U}
\def \L{\text L}

\def \cP{\mathcal P}
\def\bsDelta{\triangle}

\def \bb1{\mathbbm{1}}
\def \abar{\overline{a}}
\def \cbar{\overline{c}}
\def \Abar{\overline{A}}
\def\sfsm{\mathsf{sm}}

\def\sinc{\mathrm{sinc}}

\def \bT{\mathbb T}

\def \cG{\mathcal G}

\def \y{\mathbf{y}}

\def\Unif{\mathsf{Unif}}
\def\Id{\mathsf{Id}}
\def \Pb{\mathbb P}

\def \I {\mathrm I}

\def \sDir{\mathsf{Dir}}
\def \sCat{\mathsf{Cat}}
\def \sfD{\mathsf{D}}
\def \sfd{\mathsf{d}}
\def \sfn{\mathsf{n}}
\def \per{\mathrm{per}}

\def\se{\mathrm{se}}

\usepackage{centernot}

\def \argmin{\mathrm{arg\,min}}

\def\sJ{{\mathsf{J}}}

\def\cF{{\mathcal F}}

\def\sF{{\mathsf{F}}}

\def \bE{\mathbb E}

\def \sW{\mathsf{W}}
\def \sSW{\mathsf{SW}}
\def\sfDir{\mathsf{Dir}}
\def\sF{\mathsf{F}}

\def\Jvor{{J_\mathsf{vor}}}

\def\sfGamma{\mathsf{Gamma}}

\def \sfw{\mathsf{w}}
\def\sfM{\mathsf{M}}

\def \b{\mathrm{b}}

\def\th{\text{th}}
\def\sfy{\mathsf{y}}
\usepackage{xcolor}
\usepackage{bbm}
\usepackage{amsmath,scalerel}
\DeclareMathOperator*{\Bigcdot}{\scalerel*{\cdot}{\bigodot}}

\def\Javg{J_{\mathsf{avg}}}
\def\sfH{\mathsf{H}}

\def\finite{\mathsf{finite}}
\def\sfC{\mathsf{C}}
\def \sfCbar{\overline{\mathsf{C}}}

\definecolor{darkred}{rgb}{.7,0,0}
\definecolor{darkgreen}{rgb}{0,0.6,0}
\definecolor{darkblue}{rgb}{0,0,0.7}
\definecolor{darkmagenta}{rgb}{0.55,0,0.55}
\definecolor{darkteal}{rgb}{0,0.45,0.55}
\definecolor{darkorange}{rgb}{1,0.40,0}

\newcommand{\dad}{\mathsf{datadistribution}}
\newcommand{\gem}{\mathsf{gendistribution}}
\def\genbox#1#2#3#4#5#6{% #1=0/1, #2=color, #3=shape, #4=raise, #5=width, #6=width/2
    \leavevmode\raise#4bp\hbox to#5bp{\vrule height#5bp depth0bp width0bp
    \pdfliteral{q .5 w \csname #2COLOR\endcsname\space RG
                       \csname #3PDF\endcsname{#5}{#6} S Q
             \ifx1#1 q \csname #2COLOR\endcsname\space rg 
                       \csname #3PDF\endcsname{#5}{#6} f Q\fi}\hss}}
\usepackage[noend]{algorithmic}

\usepackage{amsthm}
\newtheorem{theorem}{Theorem}[section]
\newtheorem{lemma}[theorem]{Lemma}

\newtheorem{remark}[theorem]{Remark}

\usepackage{parskip}
\makeatletter
\def\ps@pprintTitle{%
  \let\@oddhead\@empty
  \let\@evenhead\@empty
  \def\@oddfoot{\hfil\textit{\today}}%
  \let\@evenfoot\@oddfoot
}
\makeatother
\usepackage{tikz}
\begin{document}

% \newgeometry{top=1.5cm,bottom=2cm}

\usetikzlibrary{shapes.geometric, arrows.meta, decorations.pathreplacing, positioning, calc, matrix}
\vspace*{-2.5cm}
\begin{frontmatter}

\title{\LARGE Distributional Inverse Homogenization}

% \author[1]{Author\corref{cor1}\fnref{fn1}}
% \cortext[cor1]{Corresponding author}
% \fntext[fn1]{-- was supported through --.}
% \affiliation[1]{organization={Department and Organization},%Department and Organization
%             addressline={--}, 
%             city={--},
%             postcode={--}, 
%             % state={},
%             country={--}}

% % % Order and Inclusion TBD %%
\author[1]{Arnaud Vadeboncoeur \corref{cor1}\fnref{fn1}} %% Author name
% \ead{av537@cam.ac.uk}
\author[1,2]{Mark Girolami \fnref{fn2}} %% Author name
% \ead{mag92@cam.ac.uk}
\author[3]{Kaushik Bhattacharya \fnref{fn3}}
% \ead{bhatta@caltech.edu}
\author[4]{\\Andrew M. Stuart \fnref{fn3}} %% Author name
% \ead{astuart@caltech.edu}

\cortext[cor1]{Corresponding author}
\fntext[fn1]{AV (av537@cam.ac.uk): supported through the EPSRC ROSEHIPS grant EP/W005816/1.}
\fntext[fn2]{MG (mag92@cam.ac.uk): supported by a Royal Academy of Engineering Research Chair and UKRI grants EP/X037770/1, EP/Y028805/1,
EP/W005816/1, EP/V056522/1, EP/V056441/1. }
\fntext[fn3]{KB (bhatta@caltech.edu ): supported by a US Office of Naval Research DDCR MURI (award N00014-23-1-2654).}
\fntext[fn4]{AMS (astuart@caltech.edu): supported by a Department of Defense (DoD) Vannevar Bush Faculty Fellowship (award N00014-22-1-2790) and US Office of Naval Research DDCR MURI (award N00014-23-1-2654).}
% AMS was supported by a Department of Defense Vannevar Bush
% Faculty Fellowship and by the SciAI Center, funded by the Office of Naval Research (ONR), under grant N00014-23-1-2729.}

% %% Author affiliation
\affiliation[1]{organization={Department of Engineering, University of Cambridge},%Department and Organization
            addressline={7a JJ Thomson Ave}, 
            city={Cambridge},
            postcode={CB3 0FA}, 
            % state={},
            country={UK.}}
\affiliation[2]{organization={The Alan Turing Institute},%Department and Organization
            addressline={96 Euston Rd.}, 
            city={London},
            postcode={NW1 2DB}, 
            % state={},
            country={UK.}}
\affiliation[3]{organization={Mechanical and Civil Engineering,  California Institute of Technology},%Department and Organization
            addressline={1200 E California Blvd}, 
            city={Pasadena},
            postcode={CA 91125}, 
            % state={CA},
            country={US.}}
\affiliation[4]{organization={Computing and Mathematical Sciences,  California Institute of Technology},%Department and Organization
            addressline={1200 E California Blvd}, 
            city={Pasadena},
            postcode={CA 91125}, 
            % state={CA},
            country={US.}}

% %% Abstract
\vspace{-1.5em}
\begin{abstract}
    For many materials, macroscopic mechanical behavior is determined by an intricate microstructure.
    Understanding the relation between these two scales helps scientists and engineers design better materials. 
    The relation which maps microstructure to bulk material properties can be understood via the well-established theory of homogenization.
    However inverting the homogenization process, to recover microstructural information from measured macroscopic properties, is fraught with difficulties because of the averaging processes that
    underlie homogenization.
    Therefore, scientists and engineers usually need recourse to more invasive, often highly localized, investigations to estimate the microstructure.
    In this work, we develop a noninvasive methodology by which one can leverage large collections of measured bulk material properties to infer information about the \textit{statistics} of microstructure at a global level.
    We call this, \textit{distributional inverse homogenization}.
    We study this problem in one and two dimensions, considering both 
    periodic and stochastic homogenization. We demonstrate the methodology
    in the context of 2D Voronoi constructions and underpin the observed
    empirical success with theory in 1D. We also show how the natural spatial 
    variability of microstructure can be exploited to gather data that enables distributional inversion.
    And we concurrently learn a surrogate model, approximating the homogenization map, that
    accelerates the resulting computations in this setting.
    The work formulates a new class of inverse problems, bridging ideas from
    probability and homogenization to facilitate the inference
    of microstructural material variability from macroscopic  measurements.
\end{abstract}


\end{frontmatter}

% \clearpage
\tableofcontents

\clearpage

\section{Introduction}

The theory of homogenization can be used to define a map from microscopic material properties to the resulting macroscopic material properties. Inverting this map is difficult, because it involves a complex
averaging procedure, meaning that multiple microstructures map onto the same macroscopic response. However, inferring information about the microstructure from macroscopic properties, obviating the need for intrusive methods such as etching and microscopy, is potentially of great interest. Whilst inverting the homogenization map is severely ill-posed, we show in this paper that inverting for the \emph{statistics} of the microstructure is a tractable problem. We refer to this as \textit{distributional inverse homogenization}.

Statistical characterization of microstructure has many applications. In alloy development, one may wish to understand the impact of certain processes on the distribution of grain size, shape, orientation~\cite{rollett2007three}. In metamaterials, regularity of periodic microstructure is paramount for quality control~\cite{catalanotti2016generation}. Natural and engineered cellular materials such as wood and foam are highly heterogeneous; characterizing this inherent variability can help gauge suitability of production batches for industrial use~\cite{Gibson_2003, fratzl2007nature}. In plastics, the length of polymer chains and crystallinity determine many mechanical properties~\cite{gent1970hypothetical}. In concrete, void sizes and geometric disposition affect both mechanical properties and electrolyte solution transport behavior, affecting in turn the strength and longevity of structures~\cite{martin2000study, martin2001numerical, martin1999thesis}.
In such applications statistical information is highly relevant to understanding how manufacturing controls material and how operational conditions may affect materials at a microstructural level.

In this work, we frame the task of inferring microstructure from a generative modelling perspective. That is, given expert knowledge, one poses a statistical generative model for a posited {microstructural class}. Such a model class can be hypothesized on the basis of a combination of experience, first principles and collection of small numbers of microscale images.
It may help the reader to consider the following example: posit a model class for a material as that of a polycrystaline microstructure, of $M$ crystal types, generated via a Voronoi diagram with a given number of uniformly distributed nucleation sites. The (random) map taking the properties of the $M$ crystal types to the random Voronoi diagram of the microstructure is a physically interpretable generative model.
The hypothesized microstructural generative model, for given sets of physically interpretable parameters, generates a distribution on microstructure. One then relates the sampled microstructures from the hypothesized model to macroscopic material properties using the theory of homogenization. We are then left with the task of calibrating the physically interpretable model parameters for our hypothesized model class to a dataset of macroscopic material properties from the material specimen(s) of interest. This model calibration/data-fitting stage is performed through matching
the distribution of the homogenized coefficients sampled from the generative model with the distribution of observed macroscopic responses. This is the procedure that we refer to as distributional inverse homogenization.

We develop our methodology in the setting of both periodic and stochastic homogenization. In the former, one a assumes a regular repeating microstructure extending over the entire domain to be homogenized; in the later, one assumes instead that the distribution over the microstructure is given by a random process whose $n$-point statistics are invariant to translation and that any single random realization of the field is statistically representative of any other random realization. In both cases, homogenization boils down to mapping the material configuration in a representative computational \textit{cell} to a single summary coefficient. This coefficient can be anisotropic, even if the microstructural material constituent properties are isotropic, when the underlying geometry of the material layout is anisotropic. Central to our investigation is the relation between 
microstructural constituent
material properties, the distribution of material volume fractions, their microscale geometric organization, and the resulting distribution on bulk (homogenized) properties. As will be seen, the distribution on macroscopic material properties for randomization of the microstructure is extremely information rich. Distributional inverse homogenization exploits this information.

In this paper we focus on homogenizing operators for scalar solution fields in one and two spatial dimensions. These are representative of physical applications such as steady-state heat flow and electrostatics.
We confine our attention to scalar problems in one and two spatial dimensions because of analytical tractability (one dimension) and straightforward but informative computations (Voronoi microstructures in two dimensions); this leads to a straightforward exposition of the ideas. However, we expect the work herein to generalize to distributional inverse homogenization for elasticity in two and three spatial dimensions, where a 4-tensor field characterizes the microstructure. Indeed investigation of this inverse problem in
the context of elasticity constitutes an interesting avenue for further work.

\begin{figure}[t]
    \centering
    \resizebox{0.8\textwidth}{!}{
\begin{tikzpicture}[font=\sffamily]

\node (left_grid) {
    \begin{tabular}{cc}
        \includegraphics[width=3.2cm]{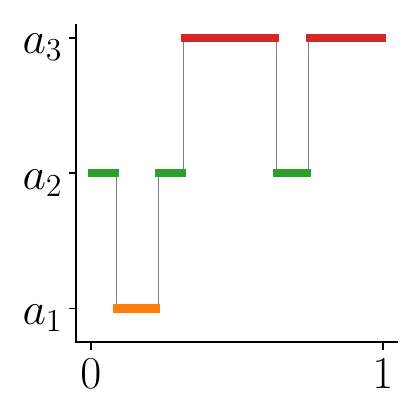} &
        \includegraphics[width=3.2cm]{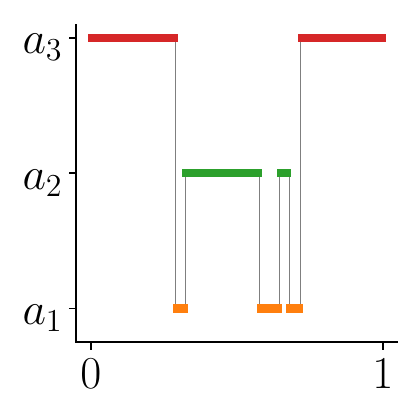} \\
        \includegraphics[width=3.2cm]{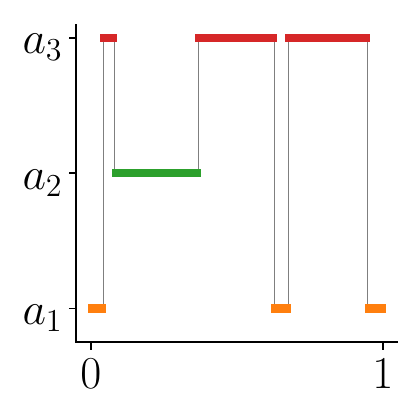} &
        \includegraphics[width=3.2cm]{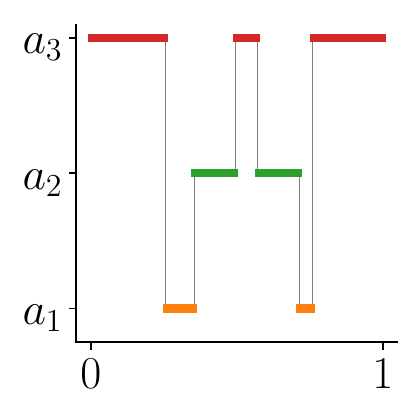}
    \end{tabular}
};

\node[below=0.2cm of left_grid] (dots1) {\Huge $\dots$};
\node[above=0.2cm of left_grid] (same_vf) {Same Vol. Frac.};

\draw [decorate, decoration={brace, amplitude=10pt, mirror, raise=4pt}, thick]
    ($(dots1.south west) + (-1,0)$) -- ($(dots1.south east) + (1,0)$);

\node[below=1.cm of dots1] (compressed1) {
    \includegraphics[width=4.cm]{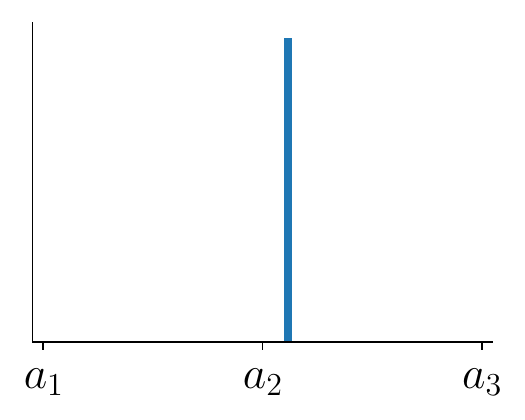}
};

\begin{scope}[xshift=8cm]
    \node (right_grid) {
        \begin{tabular}{cc}
            \includegraphics[width=3.2cm]{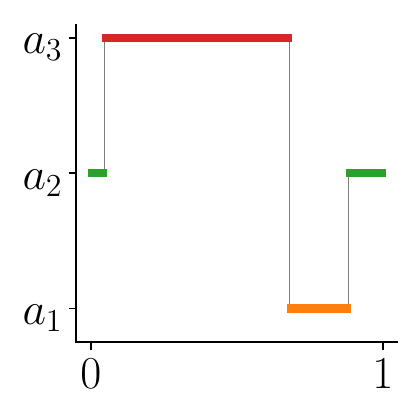} & \includegraphics[width=3.2cm]{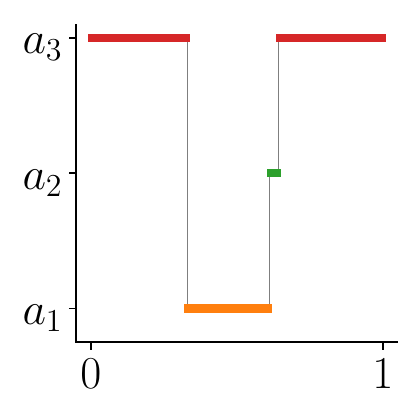} \\
            \includegraphics[width=3.2cm]{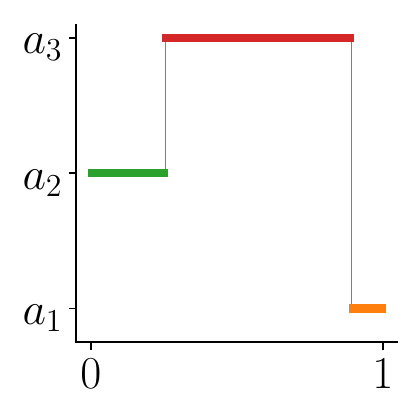} & \includegraphics[width=3.2cm]{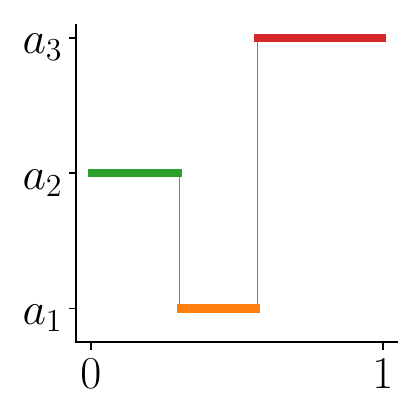}
        \end{tabular}
    };

\node[below=0.2cm of right_grid] (dots2) {\Huge $\dots$};
\node[above=0.2cm of right_grid] (same_vf) {Random Vol. Frac.};

    \draw [decorate, decoration={brace, amplitude=10pt, mirror, raise=4pt}, thick]
        ($(dots2.south west) + (-1,0)$) -- ($(dots2.south east) + (1,0)$);

    \node[below=1.cm of dots2] (compressed2) {
        \includegraphics[width=4cm]{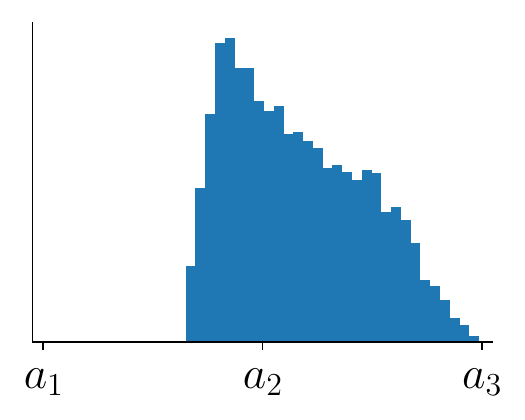}
    };
\end{scope}
\end{tikzpicture}
}
\caption{The left half of this figure illustrates why inverting the homogenization map for microstructure
is hard. The right half of the figure shows why inverting the homogenization map for statistics of
the microstructure is tractable. {Both left and right show materials comprised of the same
three constituent microstructures with $(a_1,a_2,a_3)=(1,2,3)$.} 
On the left each constituent occupies the same volume fraction in
each realization of the material, and hence all map to
a single bulk property (displayed, lower left). On the right the volume fractions are chosen at random from
a statistical distribution; this distribution may be estimated from the (displayed, lower right) distribution on
bulk properties. The particular distribution of volume fractions in the right half of this figure is described and used in~\ref{sec:proof_claim_ident}.}  
\label{fig:1D_explanation_varVolFrac}
\end{figure}

\subsection{Contributions and Outline}\label{sec:contrib_outline}
We now summarize the main contributions of this work.
\begin{enumerate}
    \item[\namedlabel{con:C1}{(C1)}] Inferring a unique microstructure from measurement of a bulk material property is not well posed. We demonstrate however, both with theory and numerical experiment, that inferring the statistics of microstructure, from measurements of a bulk material properties, is well-posed.
    \item[\namedlabel{con:C2}{(C2)}] We demonstrate the effectiveness and practical relevance of inverting for the statistics of microstructure in a suite of numerical experiments focused on Voronoi microstructures in the periodic and stochastic homogenization setting in two dimensions. {The periodic i.i.d. case is extended to the setting of elastic homogenization.}
    \item[\namedlabel{con:C3}{(C3)}] We develop a novel surrogate learning strategy to overcome computational challenges arising in distributional inversion for stochastic homogenization, making the methodology practical.
    \item[\namedlabel{con:C4}{(C4)}] We develop statistical models for locally periodic and locally stationary-ergodic Voronoi fields to approximate microstructural variations in large material specimens. Using this construction, we show how distributional inversion can exploit the natural spatial variability of materials to statistically characterize microstructure.
\end{enumerate}

Figure~\ref{fig:1D_explanation_varVolFrac} illustrates \ref{con:C1}.
The basic approach that we adopt in this paper can be described at a high level as follows.
Let $\dad$ and $\gem$ denote two probability distributions. The first is defined by the collection of
measurements at the bulk (homogenized) scale. The second is defined by a generative model that generates
microstructures randomly and maps them to bulk properties through forward homogenization (periodic or
stochastic). The generative process is parameterized by $\theta$, a parameter which controls the generation of
microstructures. Let $\sfD$ denote some form of distance measure between probability distributions.
Then the methodology underlying the computations in \ref{con:C2}--\ref{con:C4} is to solve an
optimization problem for parameter $\theta^*$ that minimizes the dissimilarity between these distributions:
\begin{equation}
    \label{eq:optp}
    \theta^\star=\argmin_{\theta}\, \sfD\bigl(\dad,\gem(\theta)\bigr).
\end{equation}
Explaining the details of the two distributions $\dad$ and $\gem$ and the choice of $\sfD$ defines
the basic methodology defined and deployed in this paper. The remainder of this section is organized as follows: Subsection~\ref{sec:litrev} reviews related work on the topics of generative modelling, distributional inference and homogenization; Subsection~\ref{sec:notation} highlights important notation.
In Section~\ref{sec:background} we provide background on homogenization and distributional inversion; Section~\ref{sec:methodology} then presents the proposed methodological contribution of this work, addressing \ref{con:C1} and explaining \eqref{eq:optp} in detail; Section~\ref{sec:1D_Per_Hom} is devoted to the 1D component of \ref{con:C1}. Section~\ref{sec:2D_per_hom} extends numerical investigation to the 2D periodic case, Section~\ref{sec:stoch_hom} in turn focuses on 2D stochastic homogenization, and Section~\ref{sec:elast} extends some of results to the case of elasticity; together these sections address the 2D component of \ref{con:C1}, together with \ref{con:C2} and \ref{con:C3}.
Contribution \ref{con:C3} amounts to learning a judicious approximation of the forward homogenization map that makes evaluation of $\gem(\theta)$ orders of magnitude faster.

\subsection{Related Works}\label{sec:litrev}

In this section we review relevant literature. We first discuss recent developments in generative modelling which provides the tools necessary for our work. This is followed by discussion of distributional inversion. Finally we discuss select works in microstructure design, making links between material science and various  objective functionals used for distributional inversion. 

\subsubsection{Generative Modelling}
In recent years, statistical inference and machine learning has focused considerable attention on tasks related to modelling datasets at the distributional level. Many works in this direction focus on generating new samples from a distribution implied by a dataset~\cite{bond2021deep, songscore, goodfellow2014generative}. A subset of these methods use divergences and metrics on the space of probability distributions which are well behaved when comparing empirical distributions: mixtures of Dirac measures. Such metrics/divergences include Wasserstein distances~\cite{panaretos2019statistical}, Energy Distance~\cite{sejdinovic2013equivalence, szekely2013energy}, Maximum Mean Discrepancy (MMD)~\cite{gretton2012kernel} and Sliced-Wasserstein~\cite{bonneel2015sliced} ($\mathsf{SW}$). The $\sSW$ distance, which plays an important role in this paper, has many variants~\cite{flamary2024pot, flamary2021pot, kolouri2019generalized, chenaugmented, deshpande2019max, nguyen2023energy, nguyendistributional}. Some works make use of generative modelling in the context of physical simulation~\cite{gao2024generative} or for solving inverse problems~\cite{goh2019solving, glyn2025primer}. However, these works do not target distributional inversion.

\subsubsection{Distributional Inversion}
The concept of distributional inversion as applied here was introduced and developed,
primarily for applications in the physical systems in~\cite{akyildiz2025efficient, vadeboncoeur2025efficient}.
In these works, problems typically approached via hierarchical Bayes~\cite{gelman1995bayesian} and empirical Bayes~\cite{robbins1992empirical} are re-interpreted as distributional inversion tasks. 
The resulting task requires a very large number of forward PDE solves over a relatively small
set of parameter values. The methodology introduced in~\cite{akyildiz2025efficient, vadeboncoeur2025efficient} establishes an approach to concurrently learn surrogate models of the forward
solve, needed at each of the optimization algorithm defining distributional inversion. Such concurrently learned surrogates are also explored in~\cite{zhang2026bilo}.
Related to the presented work are~\cite{li2025inverse, li2024stochastic, li2025least, bernton2019parameter}. Relevant application of distributionally inferred priors can be found in~\cite{vadeboncoeur2025geometric}.

\subsubsection{Microstructure Inverse Design and Indirect Investigation}
We now highlight works in the context of homogenization where objectives related to ours are explored. Of particular relevance to our work are~\cite{cherkaev2001inverse, cherkaev1998inverse}, in which the authors develop methodology for recovering spectral representations of structural information of microstructure for two-phase materials using electromagnetic measurement.
In~\cite{quey2018optimal}, design and microstructural description questions are explored by optimizing polycrystalline structure to match desired distributional properties, such as for grain size. More broadly the tools in~\cite{quey2022neper, quey2011large, sundararaghavan2006design, sternfels2011stochastic} may be of great use in answering such questions. We also highlight~\cite{al2007statistical}, which develops related methodology to understand microstructure and the associated local stress fields. In~\cite{sigmund1994materials} the author presents methodology for obtaining desired bulk material properties by designing microstructure. Questions related to forward UQ and random microstructure are explored in~\cite{yi2025cooperative}. Other noteworthy works focus on microstructural investigation and design via machine learning~\cite{rixner2022self, senthilnathan2024comparison, liu2015predictive, chatzopoulos2024physics, sundararaghavan2005classification}, some also use generative modelling methodology~\cite{zang2026design}.

\subsection{Notation}\label{sec:notation}
This work is naturally expressed in terms of the notation of probability theory, nevertheless we prioritize wherever possible notation from mechanics. We now provide an introduction to relevant notation used throughout the paper.

We denote by $v_i$ the $i^\th$ element, $i=1, \hdots, n,$ of an ordered $n-$tuple $v$.
We define functions either by specifying domain and co-domain as $f:X\rightarrow Z$, or, sometimes more conveniently, by specifying $f:x\mapsto z$ where it should be understood from context what is the implied domain and co-domain. On occasion we also overload notation for functions: we include the (co-)domain in the definition of a function so that a function denoted with the same symbol but with different (co-)domain should be understood as a different (likely very related) function.

We define the simplex 
\begin{align}
    \bsDelta_{M-1} = \{\rho\in \bR^{M}:\sum_{m=1}^M \rho_m=1, \rho_m\geq0 \text{ for } m=1,\hdots,M\}.
\end{align}
One important application of the simplex is that elements in it define a discrete probability measure on $M$ discrete elements: each $\rho_i$ is the
probability of element $i$. The simplex is also useful for describing volume fractions which,
like the probabilities, must add to $1$ and be non-negative.

For exposition, let $(X, \cF)$ be a measurable space and $\bP, \bP'$ be probability measures.
We denote superscripts in parenthesis when the values these objects take are random draws:
$v^{(i)}\sim\bP$ is the $i^\th$ random draw from $\bP$.
We use ``$\sim$'' both to mean ``drawn from'', and ``distributed according to'' depending on there being a superscript in parenthesis on the variable.
Given function $f$ and probability measure $\bP$, the pushforward $f_\#\bP$ is defined by the identity
 $f_\#\bP(B)=\bP(f^{-1}(B))$, holding for any $B\in\cF$ and measurable $f$;
a random $z^{(\Bigcdot)}\sim f_\#\bP$ is drawn by first sampling $x^{(\Bigcdot)}\sim\bP$, then setting $z^{(\Bigcdot)}=f(x^{(\Bigcdot)})$.

To draw a sample from the convolution of probability measures, denoted $z\sim(\bP\ast\bP')$, one samples $x^{(\Bigcdot)}\sim\bP$ and $x'^{(\Bigcdot)}\sim\bP'$ and sets $z^{(\Bigcdot)}=x^{(\Bigcdot)}+x'^{(\Bigcdot)}$.
By $\bP_n\xrightarrow[]{}\bP$ we mean converges in distribution.
We denote {empirical expectations} with $\widehat{\bE}$.
From context it should be clear if (sub/super)scripts denotes labels or other operations like exponentiation. 
The symbol $e^m$ is a vector of all $0$'s with a 1 in position $m$. We use $\langle\Bigcdot, \Bigcdot\rangle$ to denote the Euclidean inner product.

For a map $f:X\rightarrow Z$ to be \textit{injective} means that $f(x_1)=f(x_2)$ automatically implies that $x_1=x_2$, i.e. the map is invertible for elements in its image.
We say that identifying $x$ from $f(x)$ is identifiable if $f$ is injective.

In this work we make use of two statistical approaches to modelling random microstructure: (i) \textit{independent identically distributed} (i.i.d.) models, whereby a random realization of a microstructure contains no information pertaining to any other realization from the model;  and (ii) \textit{statistical copulas}, for
which it is possible specify the one- and two-point statistics of random spatial functions, and measurements of material properties at different locations are
statistically identical and correlated, but decorrelate on lengthscales defined
by the copula construction. In this work, the generative models used for the inference task will always be of type (i). However we apply our methodology to synthetic data of both types (i) and (ii). Data of type (i) enables us to conduct proof-of-concept
studies; data of type (ii) enables us to mimic realistic data acquisition scenarios. The construction of these two statistical models are detailed in Sections~\ref{sec:2D_per_hom}-\ref{sec:stoch_hom}.

\section{Background}\label{sec:background}
In this section we give a brief overview of homogenization for scalar elliptic operators, both for periodic and stationary coefficient fields. We then introduce distributional inversion, and finalize this section with a note on surrogate-model-enabled acceleration to amortize the compute cost of homogenization in the context of distributional inversion.

\subsection{Homogenization of Scalar Elliptic Operators}
The exposition herein is largely based on~\cite{bensoussan2011asymptotic,blanc2016some, pavliotis2008multiscale}.
We are interested in problems of the form
\begin{subequations}\label{eq:darcy_mscale}
\begin{alignat}{2}
    -\nabla\cdot(A^\varepsilon\nabla u^\varepsilon) &= f, \; &&\mathrm{for}\;x \in \D\subset\bR^d,\\
    u^\varepsilon&=0,\; &&\mathrm{for}\;x\in \partial \D,
\end{alignat}
\end{subequations}
where $A^\varepsilon(x)=A(x/\varepsilon)$, for small $\varepsilon$, and $A:\Q\rightarrow \sfM_d$ where $\sfM_d$ is the space of real, symmetric, and uniformly elliptic $d\times d$ matrices; $\Q$ is the ``material cell'' whose definition depends on the type of homogenization procedure used to model the dependence between the material micro-structure and corresponding bulk properties. We assume the forcing $f$ to be bounded. As $A^\varepsilon$ varies on a much faster scale than does the solution $u^\varepsilon$, one can 
employ homogenization to identify a constant coefficient $\Abar$, on $D$, for which the solution to~\eqref{eq:darcy_mscale} is well approximated by solution $u$ of the PDE
\begin{subequations}\label{eq:darcy_mscale2}
\begin{alignat}{2}
    -\nabla\cdot(\Abar\nabla u) &= f, \; &&\mathrm{for}\;x \in \D\subset\bR^d,\\
    u&=0,\; &&\mathrm{for}\;x\in \partial \D.
\end{alignat}
\end{subequations}
The type of homogenization used will depend what structure of $A^\varepsilon$ can be exploited for averaging. In this work we study periodic and stochastic homogenization.

\subsubsection{Periodic Homogenization}\label{sec:per_hom}
In periodic homogenization, we make the identification $\Q:=\bT^d_{[0,1)}$, so $A:\Q\rightarrow \sfM_d$ is 1-periodic. 
The homogenized coefficient $\Abar$ corresponding the cell micro-structure $A$ is
\begin{align}
    \overline{A} = \int_{\Q}\left(A(y)+A(y)\nabla\chi(y)^\top\md y\right),
\end{align}
where the corrector function $\chi:\Q\rightarrow\bR^d$ is in $H_{\mathrm{per}}^1$ and weakly satisfies the cell problem
\begin{align}
    -\nabla_y\cdot(\nabla_y\chi A^\top)=\nabla_y\cdot A^\top, \quad\quad  0 = \int_{\Q}\chi(y)\, \md y.
\end{align}
Note that the map $A \mapsto \chi$ may be viewed as a complex nonlocal averaging procedure; the homogenization map is then defined by using $\chi$ and $A$ in a simple integration over the material cell.  It is useful to define the homogenization map 
\begin{align}\label{eq:sfFdagger}
    \sF^\dagger:A\mapsto \chi\mapsto\Abar.
\end{align} 
This map may be readily approximated, for example using finite elements to solve for $\chi$, and quadrature to form $\Abar.$  The latter map induces a complicated nonlinear
averaging of $A$. If $A$ is chosen at random from a probability distribution, parameterized by $\theta$, then (approximate) evaluation of $\sF^\dagger$ defines $\gem(\theta)$, appearing in the optimization
problem \eqref{eq:optp}.

\subsubsection{Stochastic Homogenization}\label{sec:background:stoch_hom}
Let $(\Omega,\mathcal{F}, \Pb_\omega)$ be a probability triple. 
In stochastic homogenization the cell microstructure $A$ is now
defined on the whole of $\bR^d$, and depends on a randomly chosen parameter in $\Omega.$ Specifically we take $A:\bR^d\times\Omega\rightarrow \sfM_d$. Furthermore we assume that the field is stationary-ergodic:
 $\forall k\in\mathbb{Z^d}$, $A(x+k, \omega)=A(x, h_k\omega)$, where $h_k$ is an ergodic group action~\cite{blanc2016some}. Then the homogenized coefficient $\Abar$ is well-defined through an almost
 sure limiting procedure.
To define $\Abar$, we approximate the unbounded domain $\bR^{d}$ by a family of finite domains, and limits taken in the size of this domain. The resulting methodology is also the basis of computational methods for stochastic homogenization.  

We introduce the truncated domain $\Q_\sfn=(0,\sfn)^d$, with $\sfn$ large and $A_\sfn=A|_{\Q_\sfn}$. For given $\omega \in \Omega$ we may then define
\begin{align}\label{eq:stoch_hom_avg}
    \overline{A}_\sfn(\omega) = \int_{\Q_\sfn}\left(A_\sfn(y, \omega)+A_\sfn(y, \omega)\nabla\chi_\sfn(y,\omega)^\top \right)\md  y,
\end{align}
where $\chi_\sfn:\bT_{[0,\sfn)}^d\rightarrow\bR^d$ is periodic and viewed as an element in $H_{\sfn\mathrm{-per}}^1$ which weakly satisfies the $\sfn$-cell problem
\begin{align}\label{eq:stoch_hom_corrector}
    -\nabla_y\cdot(\nabla_y\chi_\sfn A^\top_\sfn)=\nabla_y\cdot A^\top_\sfn, \quad\quad  0 = \int_{\Q_\sfn}\nabla\chi_\sfn(y, \omega)\, \md y.
\end{align}
Here, as in the periodic case, one can use, for example, finite elements to obtain the solution of the corrector problem and then quadrature to define $\overline{A}_\sfn(\omega).$  We denote by $\sF_\sfn^\dagger$ the map from $A_\sfn$ to $\Abar_\sfn$, analogously to~\eqref{eq:sfFdagger}.

The homogenized coefficient $\Abar$ is obtained as the almost sure limit $\sfn\rightarrow\infty$ of $\Abar_\sfn(\omega)$. We can define the stochastic homogenization  map $\cG^\dagger:\bP_\omega\mapsto\Abar$ because $\bP_\omega$ contains all information of $A$ which is not lost in the almost sure limiting process. In practice, however, we work with finite $\sfn$,
assuming that $\sfn$ is large enough so that $\Q_\sfn$  itself is large with respect to typical
lengthscales of $A_\sfn$.
Using the fact that $A_\sfn$ is approximately stationary-ergodic in $\Q_\sfn$,   
we see that $\sF_\sfn^\dagger$ thus implements an effective averaging procedure over both variation in
$\bR^d$ and in $\Omega$. 
To reduce the variance in estimating $\Abar$ from finite $\sfn$ it is common to average over many random realizations of $\Abar_\sfn(\omega)$ via 
\begin{align}\label{eq:stoch_hom_Monte_Carlo}
    \Abar_\finite=\frac{1}{\Javg}\sum_{j=1}^{\Javg} \Abar_{\sfn}^{(j)}, \quad \Abar_\sfn^{(j)}\sim \sF_\sfn^\dagger\bigl(A_\sfn(\Q_\sfn,\Bigcdot)\bigr)_\#\bP_\omega.
\end{align}
The notation $\sF_\sfn^\dagger\bigl(A_\sfn(\Q_\sfn,\Bigcdot)\bigr)$ 
expresses the fact that $\sF_\sfn^\dagger$ acts on
$A_\sfn(y,\omega)$ as a function of $y \in \Q_\sfn$, for fixed $\omega$, to produce a constant (with respect to $y$) object; and that the resulting function of $\omega$ alone can be used to pushforward the underlying distribution on $\omega.$
The precise definition of $\Abar$ is recovered by taking the large $(\sfn,\Javg)$ limit of $\Abar_\finite$ as defined
in the preceding identity. 
In the remainder of this paper, whenever $\cG^\dagger$ is invoked, one should assume the computation is performed via~\eqref{eq:stoch_hom_Monte_Carlo} to obtain $\Abar_\finite\approx\Abar$.
There is a tradeoff in approximation efficiency between the size of $\sfn$ and $\Javg$~\cite{gloria2012numerical, gloria2012optimal, gloria2011optimal}. We have worked with choices
of the pair which work well empirically.

\subsection{Distributional Inversion}\label{sec:background:dist_inv}

In distributional inversion -- equipped with some forward model $G^\dagger$ --  we assume to have access to data of the form
\begin{subequations}
\begin{align}
    \sfy^{(n)}&=G^\dagger(z^{(n)})+\xi^{(n)},\\
    \xi^{(n)}&\sim\bP_\xi:=\cN(0, \sigma^2 \I),\\
    z^{(n)}&\sim\bP_z^\dagger,
\end{align}
\end{subequations}
with $\sfy^{(n)}\in\bR^{d_\sfy}$. (Other noise models for $\xi^{(n)}$ could be used, but we choose this
specific one for expository simplicity.)
Probability measure $\bP_z^\dagger$ is the object of interest, and we would like to infer it from the data $\sfy^{(n)}.$ 
We collect data $\{\sfy^{(n)}\}_{n=1}^N$ into an empirical distribution 
\begin{align} \label{eq:emp}
    \bP^N_\sfy=\frac{1}{N}\sum_{n=1}^N\delta_{\sfy^{(n)}};\quad \delta_{\sfy^{(n)}}(B)=\begin{cases}
    1,\; \sfy^{(n)}\in B,\\
    0,\; \sfy^{(n)} \notin B.
\end{cases}
\end{align}
This distribution is simply a collection of points; to sample from $\bP_\sfy^N$ one simply selects points at random.

\subsubsection{Optimization Problem}
 We assume to have a class of parametrized probability measures $\cP_{\Theta}$ and a map $\bP_z:\Theta\rightarrow\cP_{\Theta}$. 
The inference task is defined by the optimization problem 
\begin{align}\label{eq:measure^match_obj}
    \theta^\star\in\underset{\theta\in\Theta}{\argmin}\; \sfD\bigl(\bP^N_y, \bP_\xi * (G^\dagger)_\#\bP_z(\theta)\bigr).
\end{align}
Here $\sfD:\cP\times\cP\rightarrow\bR_{\geq0}$ is a divergence between probability measures which is well defined under empiricalization. Recall from Section~\ref{sec:notation} that ``$\ast$'' denotes the convolution operator between distributions (we add noise to the output of the generative model according to the observational noise distribution, $\bP_\xi$). 
Also recall from Section~\ref{sec:notation} that ``$\#$'' denotes the pushforward of the distribution on the right, under the map on the left. Computing $\sfD\bigl(\bP^N_y, \bP_\xi * (G^\dagger)_\#\bP_z(\theta)\bigr)$ involves three steps: (i) sampling a collection of i.i.d. realization of the form $z^{(\Bigcdot)}\sim\bP_z(\theta)$, applying $G^\dagger(z^{(\Bigcdot)})$ to them (the pushforward), and adding noise $\xi^{(\Bigcdot)}\sim\bP_\xi$ (the convolution); (ii) selecting a collection of data from $\bP^N_y$; (iii) evaluating $\sfD$ for the two collections of particles.

If the true $\bP_z^\dagger$ is obtained from the map $\bP_z$ by choosing $\theta^\dagger \in \Theta$ and if $N$ is large then it is natural to expect $\theta^\star \approx \theta^\dagger$. This is because, if 
$G^\dagger$ does not compress variations in the input, then
$\bP_\sfy^N$ will contains information about the choice of $\theta.$

In this work we employ, for $\sfD$, the Sliced-Wasserstein metric, defined by
\begin{align} \label{eq:sliced}
    \sSW_2^2(\bP, \bP') = \int_{\mathbb{S}^{d-1}}\sW_2^2(P_{\varrho\,\#}\bP, P_{\varrho\,\#}\bP')\md \varrho,
\end{align}
where $\mathbb{S}^{d-1}$ is the sphere embedded in $\bR^d$, and $P_\varrho(\,\cdot\,)=\langle \varrho, \,\cdot\,\rangle$. Numerically, the integral is approximated with $N_\varrho$ Monte-Carlo samples
\begin{subequations}\label{eq:emSW}
\begin{align}
    \widehat{\sSW_2^2}(\bP, \bP')&= \frac{1}{N_\varrho}\sum_{i=1}^{N_\varrho} \sW_2^2(P_{\varrho^{(i)}\,\#}\bP, P_{\varrho^{(i)}\,\#}\bP');\\ \varrho^{(i)}&=\tilde{\varrho}^{(i)}/\|\tilde{\varrho}^{(i)}\|,\quad \tilde{\varrho}^{(i)}\sim\cN(0, \I_{d_y}).
\end{align}
\end{subequations}
This metric is widely used because the Monte Carlo approximation requires evaluation only of one dimensional Wasserstein distances, which are analytically tractable. In contrast the Wasserstein distance is hard to compute in higher dimensional spaces or when there is large amounts of data. 
Thus, assuming $G^\dagger$ and $\bP_z$ are differentiable, we can approximate~\eqref{eq:measure^match_obj} via standard gradient descent through stochastic optimization methods such as Adam~\cite{kingma2014adam}.

\subsubsection{Surrogate Acceleration}\label{sec:background:surrogate}

One might be interested in replacing the computation of $G^\dagger$ using a surrogate model. This serve two purposes: i) it amortizes the cost of the forward map; and ii) it enables the differentiation of possibly ``black-box''\slash discontinuous forward maps $G^\dagger$.  Ideally, we would like our surrogate model, $G^\phi$ with parameters $\phi$, to be accurate for $z\sim\bP_z^\dagger$. As $\bP_z^\dagger$ is not known, we instead choose $\phi$ so that $G^{\phi}(z)\approx G^\dagger(z)$ for $z\sim\bP_z(\theta')$ with $\theta'$ an estimate of $\theta^\star$.
We formalize this as a bilevel optimization problem seeking
\begin{subequations}
\begin{align}\label{eq:surrogate_active}
    \theta^\star&\in\underset{\theta\in\Theta}{\argmin}\; \sfD(\bP^N_y, ( G^{\phi^\star{(\theta)}})_\#\bP_z(\theta)),\\ 
    \text{such that}\;\phi^\star{(\theta)} &\in \underset{\phi}{\argmin}\;\bE_{z\sim\bP_z(\theta)}\left[\|G^\phi(z) - G^\dagger(z)\|^2_2\right].\label{eq:surrogate_active_constraint}
\end{align}
\end{subequations}
This methodology is developed in depth in~\cite{akyildiz2025efficient, vadeboncoeur2025efficient, Zhangetal2026, zhang2026bilo} with various implementation strategies and approximations involving physics residuals or non-differentiable ``black-box'' models.
If one were interested in using surrogate acceleration in the periodic case with neural operators, it is straightforward to combine~\cite{Bhattacharyaetal2024} and~\cite{vadeboncoeur2025efficient, akyildiz2025efficient}. This work will explore a different surrogate modelling strategy in the stochastic homogenization case, using simple feedforward multilayer perceptron neural networks.

\section{Methodology}\label{sec:methodology}
\begin{figure}[t]
\begin{center}
\resizebox{0.8\textwidth}{!}{
\begin{tikzpicture}[
    >=Stealth, 
    thick,
    large arrow/.style={->, >={Stealth[scale=1.5]}} 
]

    \draw (0,0) rectangle (4,4);
    \node at (2, 4.5) {\Large \textbf{Specimen}};
    \node at (11.8, 4.8) {\textbf{Data Distribution}};

    \newcounter{sqcount}
    \setcounter{sqcount}{0}
    \foreach \y in {3.5, 2, 0.5} {
        \foreach \x in {0.5, 1.5, 3.5} {
            \stepcounter{sqcount}
            \node[draw, minimum size=0.4cm] (sq-\thesqcount) at (\x, \y) {};
        }
    }

    \node at (2.5, 3.5) {$\dots$}; \node at (2.5, 0.5) {$\dots$};
    \node at (0.5, 2.75) {$\vdots$}; \node at (1.5, 2.75) {$\vdots$}; \node at (3.5, 2.75) {$\vdots$};
    \node at (0.5, 1.25) {$\vdots$}; \node at (1.5, 1.25) {$\vdots$}; \node at (3.5, 1.25) {$\vdots$};

    \node[draw=none, inner sep=0pt, minimum size=2cm] (vor1) at (0.5,-2.5) {
        \includegraphics[width=2.8cm, height=2cm]{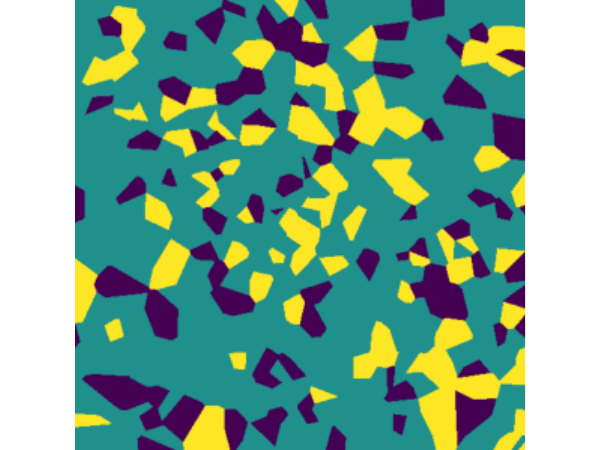}
    };
    \node[draw=none, inner sep=0pt, minimum size=2cm] (vor2) at (3.0,-2.5) {
        \includegraphics[width=2.8cm, height=2cm]{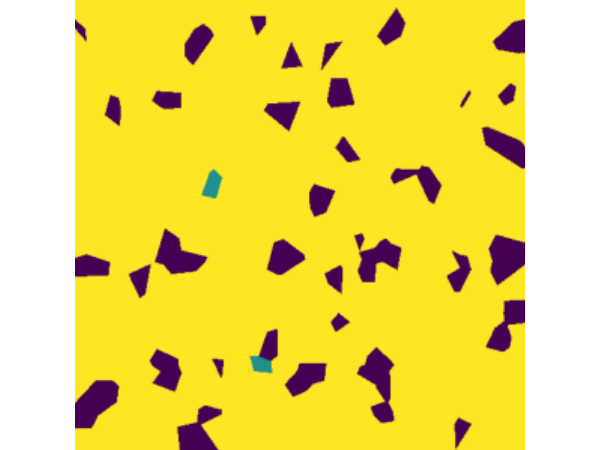}
    };
    \node[draw=none, inner sep=0pt, minimum size=2cm] (vor3) at (5.5,-2.5) {
        \includegraphics[width=2.8cm, height=2cm]{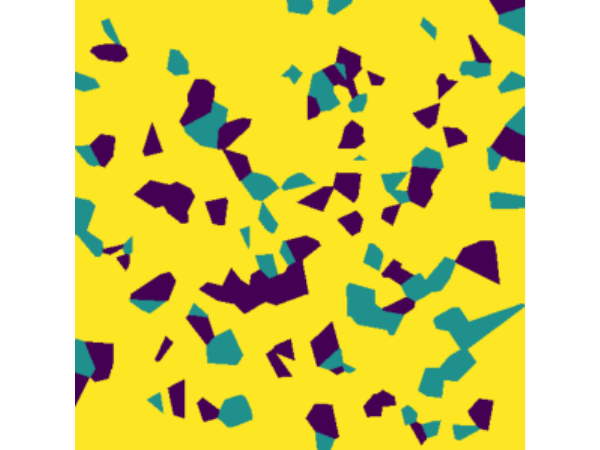}
    };

    \draw[gray] (sq-7.south west) -- ([xshift=.4cm]vor1.north west); 
    \draw[gray] (sq-7.south east) -- ([xshift=-.4cm]vor1.north east);
    
    \draw[gray] (sq-8.south west) -- ([xshift=.4cm]vor2.north west); 
    \draw[gray] (sq-8.south east) -- ([xshift=-.4cm]vor2.north east);

    \draw[gray] (sq-9.south west) -- ([xshift=.4cm]vor3.north west); 
    \draw[gray] (sq-9.south east) -- ([xshift=-.4cm]vor3.north east);

    \node[draw=none, inner sep=0pt, minimum size=5cm] (A22) at (11.5, 1.7) {
        \includegraphics[width=5.5cm, height=5cm]{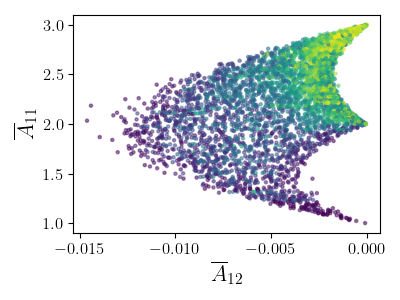}
    };

    \draw[large arrow, ultra thick, red] (sq-2.east) to [out=20, in=160] ([yshift=1.6cm]A22.west);
    \draw[large arrow,ultra thick, red] (sq-3.east) to [out=10, in=170] ([yshift=0.9cm]A22.west);
    \draw[large arrow, ultra thick, red] (sq-5.east) to [out=-10, in=190] ([yshift=-0.5cm]A22.west);
    \draw[large arrow, ultra thick,red] (sq-6.east) to [out=-20, in=200] ([yshift=-1.3cm]A22.west);

    \node[font=\large] at (6.4, 2.25) {\Large \textbf{Forward Hom.}};

    \draw[large arrow, blue, ultra thick] ([yshift=-0.5cm]A22.south) to [out=-90, in=0] 
        node[below right, pos=0.4, black] {\Large \textbf{Dist. Inv. Hom.}} (7, -5.5);

    \node[anchor=north, align=center] (matbox) at (3.0, -3.6) {
        \begin{tikzpicture}
            \draw[->, red, ultra thick, >=Stealth] (0,0) -- (-1.5,1);
        \end{tikzpicture} 
        \begin{tikzpicture}
            \draw[->, red, ultra thick, >=Stealth] (0,0) -- (0,1);
        \end{tikzpicture} 
        \begin{tikzpicture}
            \draw[->, red, ultra thick, >=Stealth] (0,0) -- (1.5,1);
        \end{tikzpicture} \\
        \textbf{Volume Fractions:} $\boldsymbol{\rho(x) \sim \Lambda(\alpha;\ell, \nu,\S)}$ \\
        \textbf{Materials:} $\boldsymbol{(a_1, \dots, a_m)}$
    };
\end{tikzpicture}
}
\end{center}
\caption{Distributional inverse homogenization schematic workflow in the locally stationary-ergodic microstructure setting described in Section~\ref{sec:loc_se_stoch}. Here $(\alpha,a)$ are the sought parameters, where $(\alpha, \ell, \nu)$ parametrize the copula description of volume fractions, $\rho$ is a spatially varying volume fraction, and $a$ are the material properties.}
\label{fig:graph_abs_stoch_hom}
\end{figure}
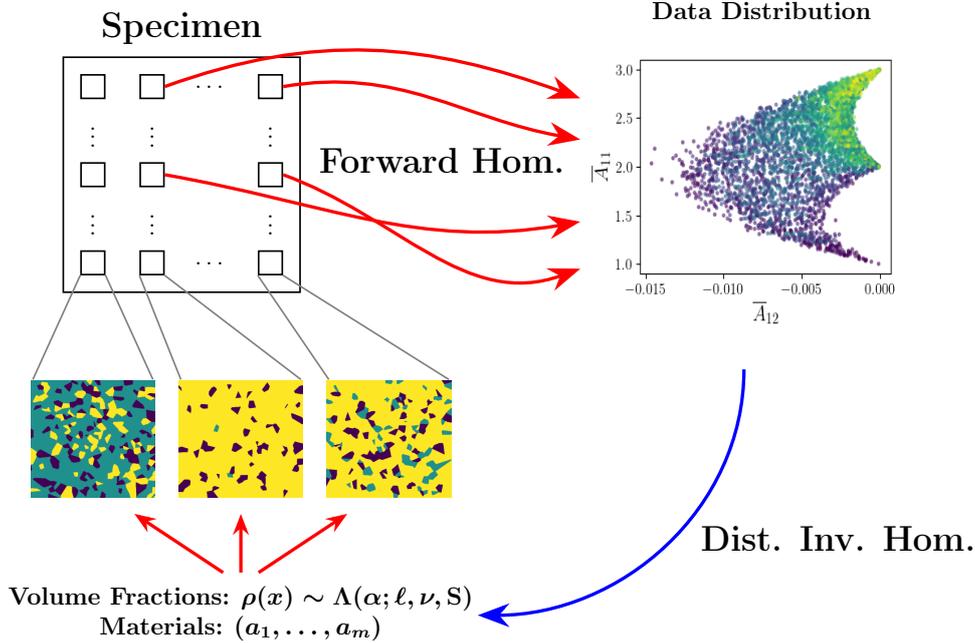
%%%%%%%%%%%%%%%%%%%%%%%%%%%%%%%%%%%%%%%%%%%%%%%%%%%%%
%%%%%%%%%%%%%%%%%%%%%%%%%%%%%%%%%%%%%%%%%%%%%%%%%%%%%
%%%%%%%%%%%%%%%%%%%%%%%%%%%%%%%%%%%%%%%%%%%%%%%%%%%%%

This section brings together the ideas from the preceding section to define
the proposed methodology for distributional inverse homogenization. We make the
idea encapsulated in \eqref{eq:optp} concrete.
The implications of the proposed approach, these being mathematical, numerical, and practical, will be explored in later sections. 

We start by assuming that there exists a data-generating distribution on material microstructure $\bP_A^\dagger.$ We observe, by way of macroscopic measurement, a collection of bulk material properties from a generating model
\begin{subequations}
\begin{align}
    \Abar^{(n)} &= \sfH^\dagger(A^{(n)})+\xi^{(n)},\\
    A^{(n)}&\sim\bP^\dagger_A,\\
    \xi^{(n)}&\sim\bP_\xi.
\end{align}
\end{subequations}
Here $\sfH^\dagger$ is the homogenization map taking a physical microstructure of a specimen to a macroscopically measurable bulk material property. 
The probability measure $\bP_\xi$  describes noise entering the measurements.
As described in abstract in \eqref{eq:emp}, we collect data 
$n=1, \hdots, N$ into an empirical distribution
\begin{align}
    \bP^{^N}_{\Abar} = \frac{1}{N}\sum_{n=1}^N\delta_{\overline{A}^{\,(n)}}.
\end{align}
We are interested in inferring aspects of the distribution on microstructure $\bP^\dagger_A$ by way of parametric approximation. 

\textbf{Distributional Inverse Periodic Homogenization}
Recall map $\sF^\dagger$ given by \eqref{eq:sfFdagger}. As explained in the general setting in Section~\ref{sec:background:dist_inv}), we directly posit a map $\theta\mapsto\bP_A(\theta)$ such that $A\sim\bP_A(\theta)$, and minimize
\begin{align}\label{eq:objective_per}
\mathsf{J_p}(\theta) = \mathsf{D}\left(\bP^{^N}_{\Abar}\,,\; \bP_\xi * (\sF^\dagger)_\#\bP_A(\theta)\right).\tag{Per. Inv. Hom. Objective}    
\end{align}
For example in our 2D periodic homogenization examples we take $\bP_A(\theta)$ to be a generative model for Voronoi diagrams with random seed locations where each Voronoi seed is deterministically assigned a material property. Furthermore, the Voronoi diagram may have a material dependent additive weighting.
The goal of distributional inversion is to determine the statistics of the Voronoi diagrams.
This idea is explored in Section~\ref{sec:2D_per_hom}. If we now choose $\sfD$ to be the sliced-Wasserstein metric \eqref{eq:sliced}, then we have defined a specific instance of \eqref{eq:optp}.

\textbf{Distributional Inverse Stochastic Homogenization}
Similarly to the periodic case, we posit a parametric statistical model for $A_\sfn^{(n)}$ which is a map $\varpi^{(n)}\mapsto\bP_\omega(\varpi^{(n)})$; the distribution $\bP_\omega(\varpi^{(n)})$ subsumes all information of $A^{(n)}$, which according to exact stochastic homogenization is defined on an unbounded domain. We then define a map $\theta\mapsto\bP_\varpi(\theta)$ such that $\varpi^{(n)}\sim\bP_\varpi(\theta)$. This way, $\{A_\sfn^{(n, i)}\}_{i=1}^\infty$ for $A_\sfn^{(n, i)}\sim\bP_\omega(\varpi^{(n)})$ summarizes all information in $A^{(n)}$, yet $A_\sfn^{(n,i)}$  is defined on a bounded domain. We are now in a position to pose our objective: minimize
\begin{align}\label{eq:objective_stoch}
\mathsf{J_s}(\theta) = \mathsf{D}\left(\bP^{^N}_{\Abar}\,,\; \bP_\xi * (\cG^\dagger)_\#\bP_\varpi(\theta)\right).\tag{Stoch. Inv. Hom. Objective} 
\end{align}
In this paper, when studying 2D stochastic homogenization,
$\bP_\omega(\varpi^{(n)})$ is the generative model from which we can sample random Voronoi diagrams, with random seeding locations; each Voronoi cell is assigned a material property at random. The variables $\varpi^{(n)}$ are themselves drawn from a distribution with the parameters $\theta$ we seek to infer; $\varpi$ will represent a description of volume fraction per material type as well as the associated constituent material property. 
If we again choose $\sfD$ to be the sliced-Wasserstein metric \eqref{eq:sliced}, then we have defined a specific instance of \eqref{eq:optp}.

To implement \eqref{eq:objective_per} and \eqref{eq:objective_stoch}, we use the approximate Sliced-Wasserstein distance~\eqref{eq:emSW} as our divergence $\sfD$, rather than the exact
Sliced-Wasserstein distance~\eqref{eq:sliced}. The methodology presented in this current form is defined for observational data $\smash{\{\Abar^{(n)}\}_{n=1}^N}$ which are i.i.d. Of course the homogenized coefficients are symmetric, so we only compute the sliced-Wasserstein distance with the three unique coefficient elements.
In this work we also explore a novel surrogate modelling direction specifically for stochastic 
homogenization, by concurrently learning an approximation to $\cG^\dagger$ given by \eqref{eq:stoch_hom_Monte_Carlo}, which we will call $\cG^\phi$. More details on this may be found in Section~\ref{sec:stoch_hom}. In later sections, we further extend our stochastic homogenization methodology to consider physically realistic scenarios where the data originates from a single large specimen queried at various locations, and is only approximately i.i.d. Here, homogenization approximately holds due to local conditions. Figure~\ref{fig:graph_abs_stoch_hom} is graphical summary of that particular workflow in the stochastic setting. More on this in later sections.

\section{1D Homogenization}\label{sec:1D_Per_Hom}

We now explore certain mathematical properties, in one spatial dimension, of  the problem posed in Section~\ref{sec:methodology}. We take a perspective that subsumes both periodic and stochastic homogenization into a single framework, in this 1D setting.
Subsection~\ref{sec:first_dist_simplex} constructs an explicit example and a nontrivial inverse homogenization problem concerning distributions on volume fractions of $M$ materials, which we prove to be identifiable.
Subsection~\ref{sec:1Ddirichlet} explores a more practically relevant class of distributional inverse homogenization problems using Dirichlet distributed volume fractions. We prove identifiability
in a specific limiting regime, of both material properties and parameters of the distribution and, in Subsection~\ref{ssec:labelNE}, provide numerical evidence demonstrating that this
theorem extends beyond the limiting case that our theory covers.
This addresses Contribution~\ref{con:C1}.

In one spatial dimension, periodic homogenization of a periodic  field, $a_\per(\Bigcdot)$, has closed form~\cite{pavliotis2008multiscale}
\begin{align}\label{eq:1D_per_homog}
    \bar{a}_\per=\left(\int_\bT a_\per(y)^{-1}\md y\right)^{-1}.
\end{align}
Stochastic homogenization of a stationary-ergodic field, $a_\se(\Bigcdot)$, also has closed form~\cite{kozlov1980averaging} given by
\begin{align}
    \bar{a}_\se=\lim_{L\rightarrow\infty}\left(\frac{1}{2L}\int_{-L}^L a_\se(y, \omega)^{-1}\md y\right)^{-1} \text{for } \bP_\omega-a.e.\; \omega.
\end{align}
This limiting procedure in $L$ is what we approximate with the large $(\sfn,\Javg)$ limit in
\eqref{eq:stoch_hom_Monte_Carlo}.

In the periodic homogenization case, assume the piecewise constant structure
\begin{align}\label{eq:gen_1D_per}
    a_\per(y) &= \sum_{m=1}^M a_j\mathbbm{1}_{\Q_m}(y),
\end{align}
where $\{\Q_m\}_{m=1}^M$ is a partition of $\bT$ with $a_1>\hdots>a_M>0$.
In the stochastic homogenization case, assume
\begin{align}\label{eq:gen_1D_stoch}
        a_\se(y, \omega) &= \sum_{m=1}^M a_j\mathbbm{1}_{\mathrm{U}_m(\omega)}(y),
\end{align}
where $\{\mathrm{U}_m(\omega)\}_{m=1}^M$ is a partition of $\bR$  with $a_1>\hdots>a_M>0$.
Now, in both cases assume
$\rho=(\rho_1, \hdots, \rho_M)\in\bsDelta_{M-1}$ is the volume fraction (of $\Q_m$ or $\U_m$) associated to the material coefficient $a_m$.
We define $a=(a_1, \hdots,a_M)$ and note that we may then write the following simple formula
for the homogenized coefficient $\overline{a}$, valid in both periodic and stochastic homogenization:
\begin{align} \label{eq:to}
    \overline{a}=\langle a^{-1}, \rho\rangle^{-1}.
\end{align}
That is to say, in both periodic and stochastic 1D homogenization the homogenized coefficient is the volume fraction weighted harmonic mean.
In the rest of this section we use the 1D simple-coefficient homogenization map
\begin{align} \label{eq:stnal}
        \sF(\rho;a)&= \left(\sum_{m=1}^Ma_m^{-1}\rho_m\right)^{-1}.
\end{align}

\subsection{Volume Fractions and a First Distribution on the Simplex}\label{sec:first_dist_simplex}

Recall Figure~\ref{fig:1D_explanation_varVolFrac} and in particular recall that the left-hand side
illustrates the fact, apparent from \eqref{eq:to}, that in the 1D setting multiple materials have the same
homogenized response -- if the volume fractions are fixed then the harmonic average is the same regardless of how the materials are arranged, as formula \eqref{eq:to} shows.
We develop here our first construction that illustrates the right-hand side of
Figure~\ref{fig:1D_explanation_varVolFrac}. We employ a distribution on the simplex of volume fractions
and show that it is possible to recover parameters of this distribution.
\vspace{1em}
\begin{theorem}\label{claim:exist_identif} Let $M=2$, fix $a \in (0,\infty) \times (0,\infty)$ 
and assume that $\rho_1 \sim \Unif(b-l, b+l)$ with $b \ge l>0$ and $b+l \le 1$; necessarily
$\rho_2 =1-\rho_1$ to ensure the volume fractions add to unity.
Then $b,l$ are determined uniquely by the distribution of 
homogenized material property $\sF(\rho;a).$
\end{theorem}
Thus we can recover the distribution of the volume fractions of the microstructure from the distribution
of the bulk properties. Theorem \ref{claim:exist_identif} can be extended to $l=0$ where it is simply the statement that in one dimension the volume fractions are identifiable from a single measurement of bulk properties if the constituent materials are known. Theorem \ref{claim:exist_identif} is a special case of a more general result, applying for any integer $M \ge 2$, and given in  \ref{sec:proof_claim_ident}.

\subsection{The Dirichlet Distribution}\label{sec:1Ddirichlet}
\begin{figure}[t] 
    \centering
    \begin{subfigure}[b]{0.48\textwidth}
        \centering
        \includegraphics[width=\linewidth]{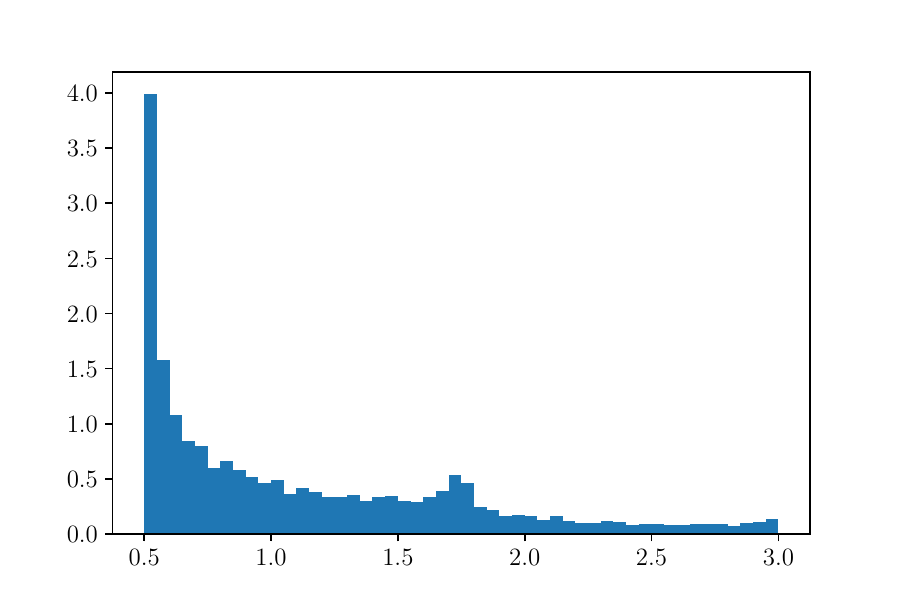}
        \caption{Hist $\overline{a}^{\,(n)}$}
        \label{fig:Diri1D_demo:hist}
    \end{subfigure}
    \hfill
    \begin{subfigure}[b]{0.48\textwidth}
        \centering
        \includegraphics[width=\linewidth]{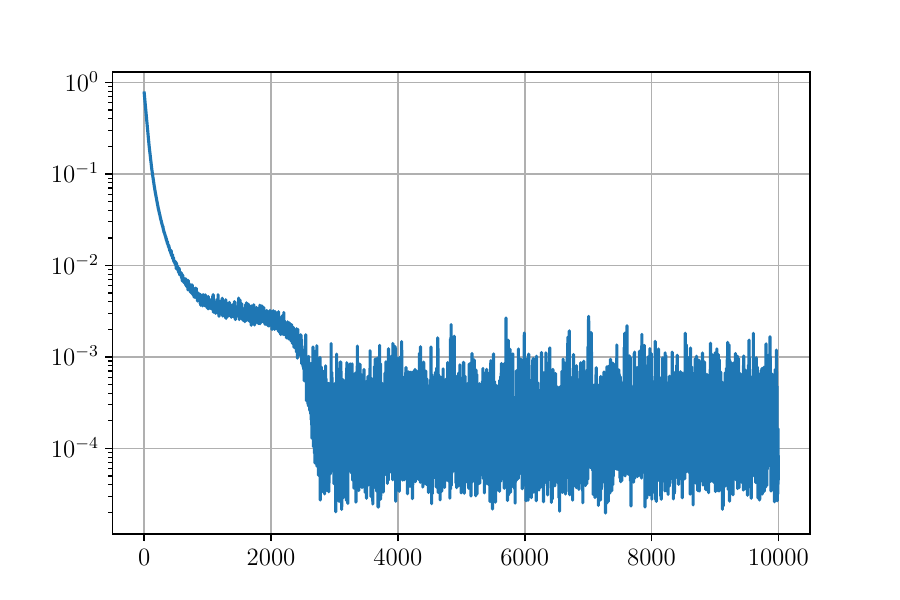}
        \caption{Objective}
        \label{fig:loss}
    \end{subfigure}
    \vspace{1em}
    \begin{subfigure}[b]{0.48\textwidth}
        \centering
        \includegraphics[width=\linewidth]{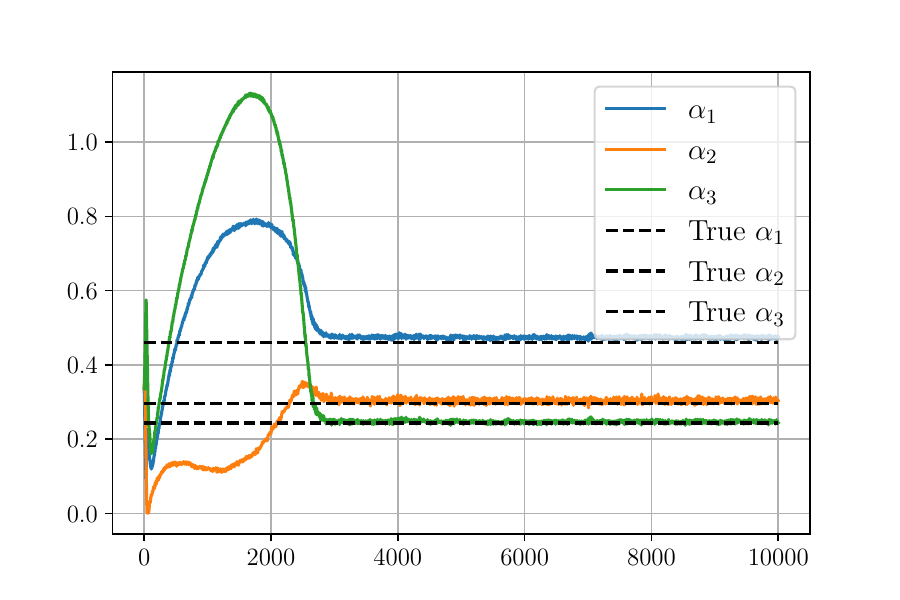}
        \caption{$\alpha$}
        \label{fig:alpha}
    \end{subfigure}
    \hfill
    \begin{subfigure}[b]{0.48\textwidth}
        \centering
        \includegraphics[width=\linewidth]{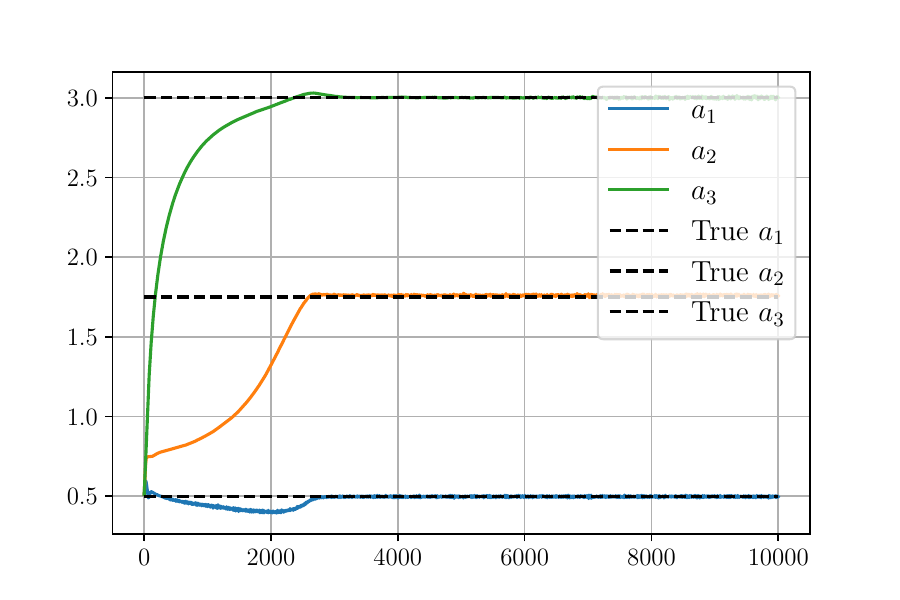}
        \caption{$a$}
        \label{fig:cs}
    \end{subfigure}
    \caption{(a) shows the histogram of the observed homogenized coefficient dataset, (b) the decreasing Objective, (c) the estimated volume fraction parameters of the microstructural generative model during minimization Objective, (d) the estimated constituent material parameters of the microstructural generative model.}
    \label{fig:Diri1D_demo}
\end{figure}
\begin{figure}[h!] 
    \centering
    \hspace{-1em}
    \begin{subfigure}[b]{0.42\linewidth}
        \centering
        \includegraphics[width=\linewidth]{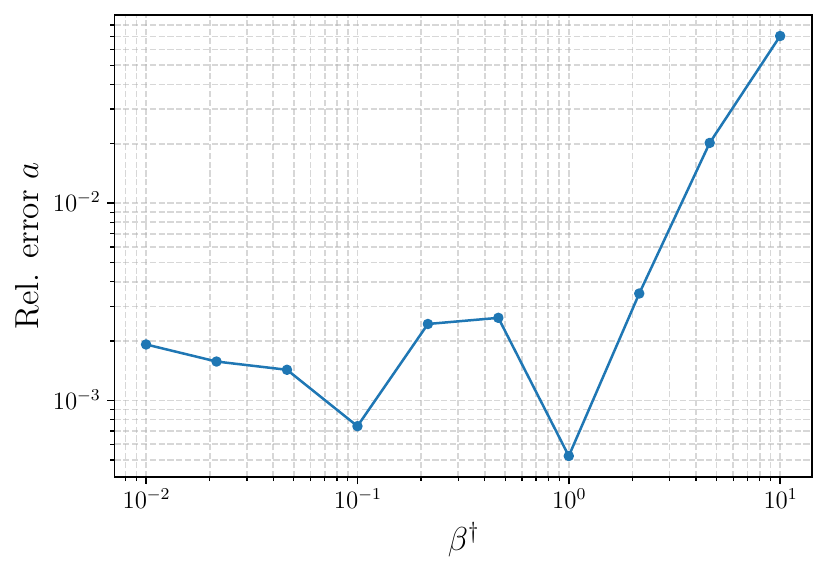}
        \caption{Relative Error $a^\star$}
        \label{fig:alpha}
    \end{subfigure}
    \begin{subfigure}[b]{0.42\linewidth}
        \centering
        \includegraphics[width=\linewidth]{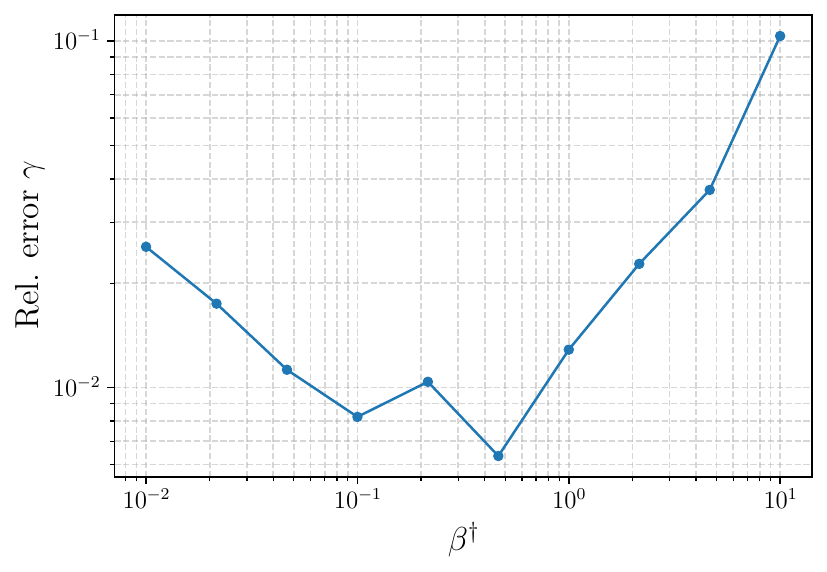}
        \caption{Relative Error $\gamma^\star$}
        \label{fig:alpha}
    \end{subfigure}
        \begin{subfigure}[b]{0.42\linewidth}
        \centering
        \includegraphics[width=\linewidth]{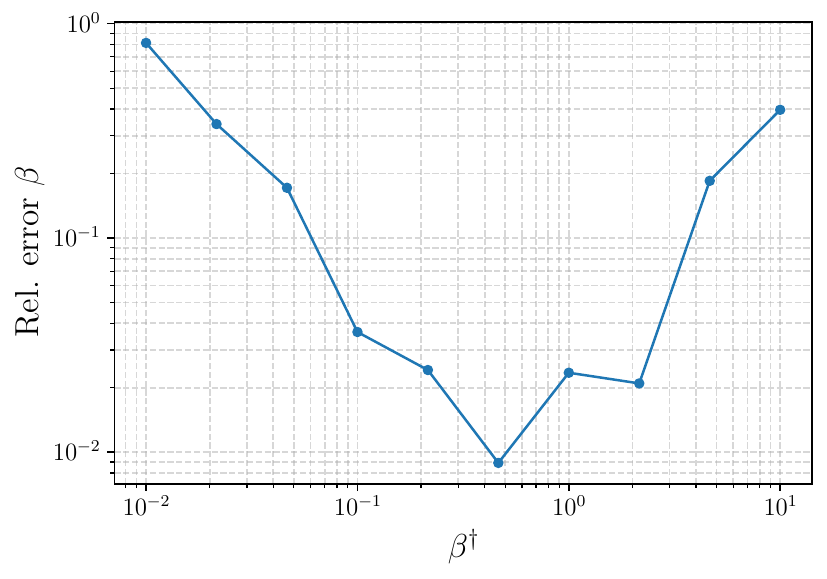}
        \caption{ Relative Error $\beta^\star$}
        \label{fig:hist}
    \end{subfigure}
    \begin{subfigure}[b]{0.42\linewidth}
        \centering
        \includegraphics[width=\linewidth]{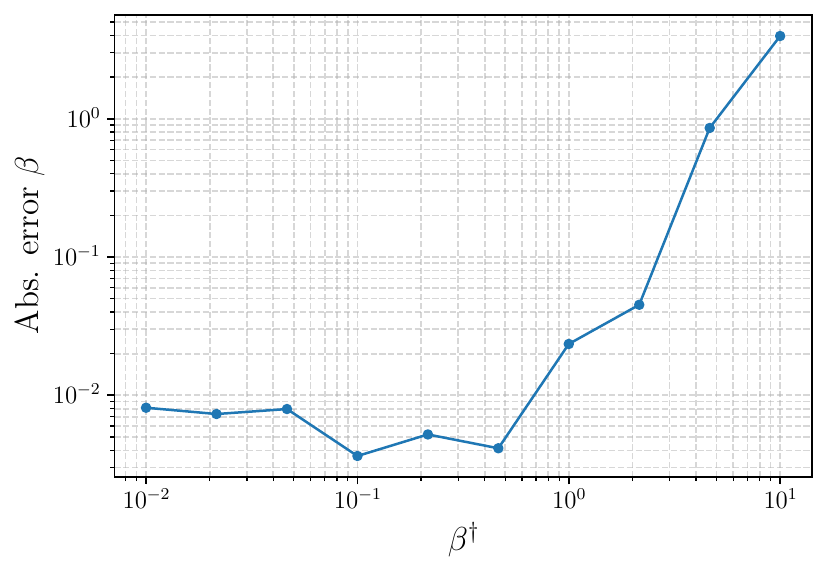}
        \caption{Absolute Error $\beta^\star$}
        \label{fig:loss}
    \end{subfigure}
    \caption{Errors in estimating $\beta, \gamma, a,$ as compared to ground truth, from distributional inversion, 1D periodic.}
    \label{fig:errorsDiri1D_varyBeta}
\end{figure}
The construction of a probability measure on $\bsDelta_{M-1}$ used in the preceding subsection, and in
\ref{sec:proof_claim_ident}, is straightforward as it uses only uniform distributions; but it is
somewhat unnatural. The canonical distribution on the simplex $\bsDelta_{M-1}$ is the Dirichlet distribution, which we study in this subsection. We describe
a setting in which the parameters of this distribution and the microstructural material properties
are jointly identifiable. The following subsection provides numerical evidence beyond the particular
parameter limit in which we prove the result.

The Dirichlet distribution is a a multivariate generalization of the beta distribution, and
may also be derived as a normalized  multivariate version of the Gamma distribution. The Gamma distribution ensures positivity and the
normalization enforces that draws are in the simplex. The density of the Dirichlet distribution, over variable $\rho \in \bsDelta_{M-1}$ and parameterized by Euclidean
vector $\alpha=(\alpha_1, \cdots, \alpha_M) \in (0,\infty)^M$, is given by
\begin{align}\label{eq:Dirich_density}
    \sfDir(\rho;\alpha) = \frac{1}{B(\alpha)}\prod_{m=1}^M \rho_m^{\alpha_m-1},\quad\quad
    B(\alpha) = \frac{\prod_{m=1}^M\Gamma(\alpha_m)}{\Gamma\left(\sum_{m=1}^M\alpha_m\right)}.
\end{align}
Here $\Gamma$ is the Gamma function.
We use $\sDir(\rho;\alpha)$ as a distribution on material volume fractions; combined with material constants $a$, and because 1D homogenization is invariant to material layout,~\eqref{eq:to}, it is a generative model for 1D material microstructure with parameters $(\alpha, a).$ 
Note that since the $\alpha_m$ are positive and finite in number we may write
$\alpha:=\beta\gamma$ for a scalar parameter $\beta$ and vector $\gamma\in\bsDelta_{M-1}.$ We introduce the notation $\sDir(\rho;\gamma, \beta):=\sDir(\rho;\gamma\beta)$ (and often drop argument $\rho$). Our second claim about identifiability of $\gamma$ concerns the limit $\beta\rightarrow0$.
It is proved in~\ref{sec:1DDirichlet_identifiability}. \ref{sec:1DDirichlet_identifiability} also contains Lemma \ref{lem:limits_dirichlet}
which characterizes the behaviour of $\sDir(\gamma, \beta)$ in the limits 
$\beta\rightarrow0$ and $\beta\rightarrow \infty.$

\vspace{1em}
\begin{theorem}\label{claim:1DDirichlet_identifiability}
    Consider the distributional inverse homogenization problem in the restricted case where $\beta\rightarrow0$. Then this problem is jointly identifiable for the parameter of the Dirichlet distribution, $\gamma \in \bsDelta_{M-1}$ 
    and the material properties $a \in (0,\infty)^M$ for $a_j\neq a_i$, $i\neq j,$ and $ 1\leq i,j\leq M.$
\end{theorem}

\subsection{Numerical Experiments} \label{ssec:labelNE}
We now present 1D numerical experiments demonstrating identifiability in the Dirichlet setting.
These suggest identifiability beyond the narrow setting of Theorem~\ref{claim:1DDirichlet_identifiability}. We first demonstrate convergence plots for the distributional inversion problem, then we explore the parameter recovery as a function of the magnitude of $\beta>0$.

\textbf{Data and Inference:} 
Figure~\ref{fig:Diri1D_demo:hist} shows the dataset histogram of the of homogenized coefficients for Dirichlet distributed volume fractions. Here, the number of materials considered, $M=3$, is fixed, and the number of data points in the empirical data distribution is $N=10^4$. The values for $a^\dagger=(0.5, 1.75, 3.)$ and $\alpha^\dagger=(0.46, 0.30, 0.24)$.

\textbf{Results:}
Figs~\eqref{fig:Diri1D_demo}(b-d)  show the decreasing~\eqref{eq:objective_per}, and the convergence of Dirichlet parameters and material properties.

These results provide further evidence for the well-posedness of distributional inversion problems of this form,
beyond the setting of the somewhat restrictive claims. In the next two sections we develop our numerical studies further to explore 2D material fields described by Voronoi diagrams.\\

\begin{remark} We briefly discuss the role of optimization in parameter identification as the problems
considered are typically nonconvex.
    In Figure~\ref{fig:errorsDiri1D_varyBeta} we test the accuracy of recovery of the Dirichlet parameters, parametrized with $\beta\gamma=\alpha$, and the material coefficients. In test where the ground truth value of $\beta$ is small we observe occasional convergence to local minima; to mitigate this effect we use ten random restarts and select the run with least final value of~\eqref{eq:objective_per} for reporting.
We observe the relative error in recovering $\beta$ increases when its magnitude decreases. The choice
of metric or divergence explored to define our objection will also have bearing on the optimization; however
we do not explore alternatives to sliced-Wasserstein in this paper.
\end{remark}

\section{2D Periodic Homogenization}\label{sec:2D_per_hom}
In this section, we apply distributional inverse homogenization to problems in periodic homogenization. We first look
at a class of i.i.d. random microstructures generated via Voronoi tesselation in
periodic homogenization. We then explore a practically relevant locally periodic Voronoi construction based on statistical copulas; we infer distributional information from spatially decorrelated measurements of bulk material properties. This section, together 
with Section~\ref{sec:stoch_hom}, addresses Contributions~\ref{con:C2} and~\ref{con:C4}. In what follows we use the terms ``Voronoi seed'' and ``nucleation sites'' interchangeably.
\vspace{1em}
\begin{remark}
    In conducting the experiments that follow, in this section~\ref{sec:2D_per_hom} and in Section~\ref{sec:stoch_hom},  nonconvexity is evidenced by the presence of 
    occasional local minima. Studying the properties and reasons for the existence of such local minima is a problem of interest in its own right, but not one we focus on in this paper. To address
    nonconvexity, we perform five random restarts and return the outcome giving the lowest objective value (best data-fit) for the last $100$ iterations.
\end{remark}\vspace{1em}

\begin{remark}
    The off-diagonal element of the constant homogenized material field is typically on a different scale to the diagonal elements. Hence, when computing the data-fit objective we scale each component by the standard deviation of that component in the data set to equally weigh contributions from the three elements. This is done both for \eqref{eq:objective_per} and \eqref{eq:objective_stoch}.
\end{remark}

\subsection{Periodic i.i.d. Microstructures}\label{sec:2D_per_hom_i.i.d.}
\begin{figure}[t] 
    \centering
    \begin{subfigure}[b]{0.32\linewidth}
        \centering
        \includegraphics[width=\linewidth]{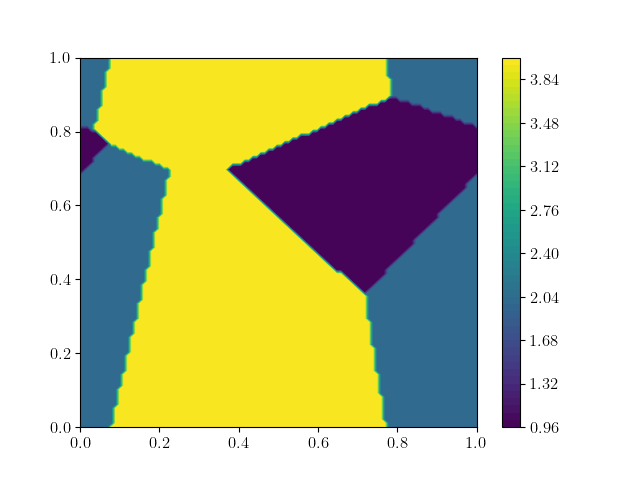}
        \caption{Random periodic cell $A^{(1)}_{11}(y)$}
        \label{fig:alpha}
    \end{subfigure}
        \begin{subfigure}[b]{0.32\linewidth}
        \centering
        \includegraphics[width=\linewidth]{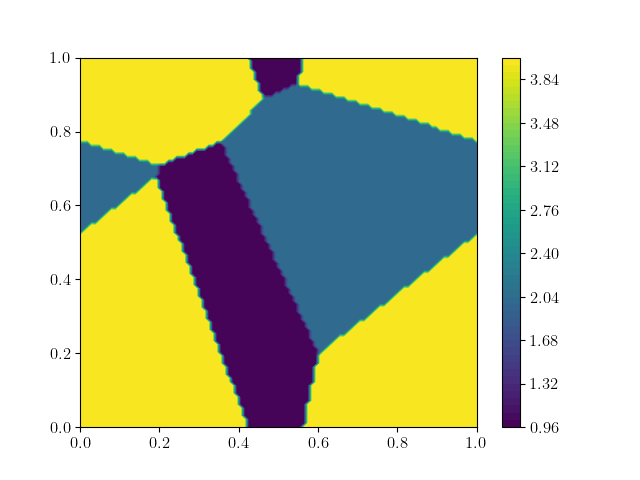}
        \caption{Random periodic cell $A^{(2)}_{11}(y)$}
        \label{fig:hist}
    \end{subfigure}
    \begin{subfigure}[b]{0.32\linewidth}
        \centering
        \includegraphics[width=\linewidth]{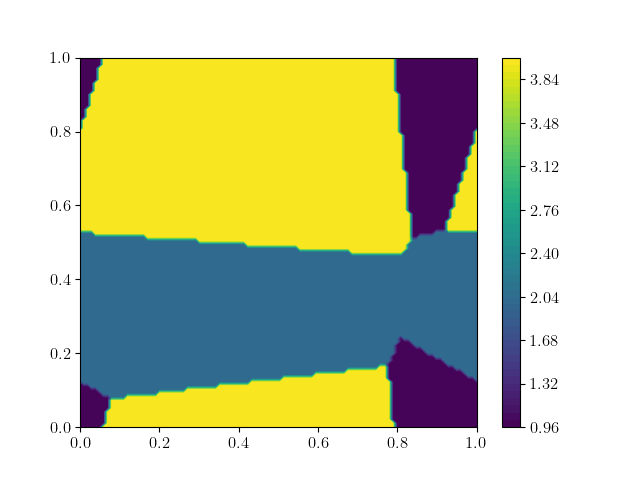}
        \caption{Random periodic cell $A^{(3)}_{11}(y)$}
        \label{fig:loss}
    \end{subfigure}
    \caption{Randomly generated periodic microstructure for $(w, a)$ randomizing over Voronoi seed locations in the cell, $\Q$.}
    \label{fig:Per_2D_rnd_A11_samples}
\end{figure}
We approach this section by first discussing the model setup, followed by the synthetic data, and finally the results.

\textbf{Setup:} In this section we consider i.i.d. data and an i.i.d. data generating model. Building on the previous section where we studied 1D materials, we seek to jointly infer information pertaining to volume fractions and material properties. As we are in 2D, we apply this methodology to material property fields described by Voronoi diagrams, as these are widely used in generative modelling for materials.

For the Voronoi description of a material field, it is not obvious how to directly specify the volume fractions in a cell. One could resort to optimizing Voronoi seed centers to tune the volume fractions shared by the $M$ material types, however we choose a different approach. We cast the problem as the task of inferring additive Voronoi seed weights, $w$, which can be used to model shifts in crystal nucleation start times.
This construction is also called a Laguerre tessellation~\cite{quey2018optimal}.
We define a partition $\bigcup_{m=1}^M\Q_m=\Q$, where
\begin{subequations}\label{eq:per_vor}
\begin{align}
    \Q_m(v; w) &= \{y\in\Q: \,\sfd_\per(y, v_m)^2-w_m\leq \,\sfd_\per(y, v_j)^2-w_j, \forall j\neq m\},\\
    \sfd_\per(y,y')&= \sqrt{\sum_{i=1}^d(y_i-y_i')^2\wedge(1-|y_i-y'_i|)^2}, \label{eq:d_per}
\end{align}
\end{subequations}
where $\wedge$ is a binary function returning the least of the two elements being compared. 
Thus $\sfd_\per$ is a  metric on the torus -- it measures distance accounting for periodicity.
We sample a candidate material field from our microstructure generative model by assigning scalar material constant $a_m$ to the $m^\text{th}$ region of a random Voronoi tesselation, $\Q_m$.
Specifically,
\begin{subequations}\label{eq:model_i.i.d._per}
\begin{align}
    v^{(m)}&\sim\Unif(\Q),\label{eq:model_i.i.d._per_Unif}\\
    A(y; a, w)&=\I_2\sum_{m=1}^M a_m\bb1_{\Q_m(v^{(1:M)};w)}(y),\label{eq:per_vor_Afield_multWeight}
\end{align}
\end{subequations}
% where $\I_2$ is the $2\times 2$ identity matrix.
% Thus we assign material property $a_m$ to Voronoi
% cell $\Q_m.$ 
where $v^{(m)}$ are uniformly distributed random nucleation sites. We refer to generative model~\eqref{eq:model_i.i.d._per} as the periodic i.i.d. model.
Note that the size of $\Q_m$, and so the volume fraction associated to material $m$, will be affected by the weight $w_m$, the relation to the weights of associated with cells defined by other nearby nucleation sites, and the location of randomly generated nucleation sites.
Here, parameter $w$ plays the role of $\rho$ in Section~\ref{sec:1D_Per_Hom}.
We seek to infer $w$ and the material properties $a$ for fixed distribution on the location of the Voronoi nucleation sites,~\eqref{eq:model_i.i.d._per_Unif}.  Thus we define
$\theta=(w, a)$ and note that we have a map $\theta \mapsto \bP_{\,\overline{A}}.$ The source of stochasticity yielding an information rich distribution on $\Abar$ comes from the randomization over the Voronoi seeds $v^{(1:M)}$ for each $\Abar$. In this setup, there is an additive lack of identifiability with regards  to the start times; the Voronoi diagram is invariant to constant shifts in the collection of lag times (physically, it does not matter if all nucleations are delayed by the same amount). And so, for inference we fix $w_1=0$ and infer the other $\{w_m\}_{m=2}^M$.
\vspace{1em}

\begin{remark}
In the Voronoi construction~\eqref{eq:model_i.i.d._per}, for any material nucleation time shifts $0\leq w_1<\hdots<w_m<\hdots<w_M$, there is positive probability that a randomly generated diagram yields a volume fraction, for any material $m<M$, of zero. 
    In practice we notice that if the difference in shift times is small enough with respect to the cell size, so that zero volume fractions are rare occurrence, we maintain identifiability of $(w,a)$.
\end{remark}
\vspace{1em}

{
\begin{remark}\label{rem:HS_Bounds_Per}
    In homogenization, Hashin-Shtrikman bounds are of general interest. These specify the upper and lower extremal values of isotropic homogenized coefficients that can be obtained with any geometric configuration of material phases for given material volume fractions and constituent material property. Under the weighted Voronoi periodic construction outlined in this section, and if $a_M=\max(a^\dagger)$ with $w_M$ the largest Voronoi weight, an upper bound on the distribution of homogenized coefficients can be obtained (it is achieved when all nucleation sites are sufficiently close to $v^{(M)}$). Then, the periodic cell will contain only material $a_M$ and hence the homogenized coefficient with be $\Abar=\I_2a_M$. Figures~\ref{fig:Per_2D_AddW_Hists_scatter} (d)-(f) show this upper bound is approximately achieved by a random sample, noting the data generating $a^\dagger=(1,2,4)$ which implies a homogenized coefficient upper bound $a^+=4$. Obtaining general upper and lower bounds on the distribution of homogenized coefficients is less straightforward and we leave this line of investigation to future work.
\end{remark}
}

We will generate from the i.i.d. model and we will attempt to estimate the crystal growth rates and material properties $(w, a)$ by minimizing~\eqref{eq:objective_per}. To deploy gradient based minimization on this objective function requires computation of gradients of $\sJ_\mathsf{p}$, and hence of $\Abar$, with respect to $w\in\theta$. As it is, the generative model~\eqref{eq:model_i.i.d._per}  is not differentiable with respect to $w$. Thus we consider a continuous relaxation of~\eqref{eq:per_vor_Afield_multWeight}. 
To do this, specify $\tau$, a relaxation parameter
and define the softmax function acting on an $M$-dimensional vector $\sfsm:\bR^M\times \bR_+\rightarrow\bR^M_+$ is
\begin{align}\label{eq:softmax}
    \sfsm(z; \tau)_i=\frac{\exp{(z_i/\tau)}}{\sum_{j=1}^m \exp{(z_j/\tau)}}.
\end{align}
We call the relaxed random field $\widetilde{A}$, it is sampled via
\begin{subequations}\label{eq:tildeA}
\begin{align}
    v^{(m)}&\sim\Unif(\Q),\\
    \widetilde{A}(y; w, a, \tau) &= \I_2\left\langle a,\, \sfsm\left(\sfd_\per(y, v^{(1:M)})^2-w;\tau)\right)\right\rangle.\label{eq:cont_relax_A_crystalLag}
\end{align}
\end{subequations}
The softmax function ensures that for every sub-Voronoi-cell some amount of information from each material parameter, $(w, a)$, is present; this allows automatic differentiation through this model to estimate $(w,a)$.
As $\tau\rightarrow0$, $\widetilde{A}(y;w,a,\tau)\rightarrow A(y;w,a)$ and $\nabla_{(w,a)} A(y;w,a,\tau)$ is well behaved for $\tau>0.$

\begin{figure}[t] 
    \centering
    \hspace{1.1em}
    \begin{subfigure}[b]{0.31\linewidth}
        \centering
        \includegraphics[width=\linewidth]{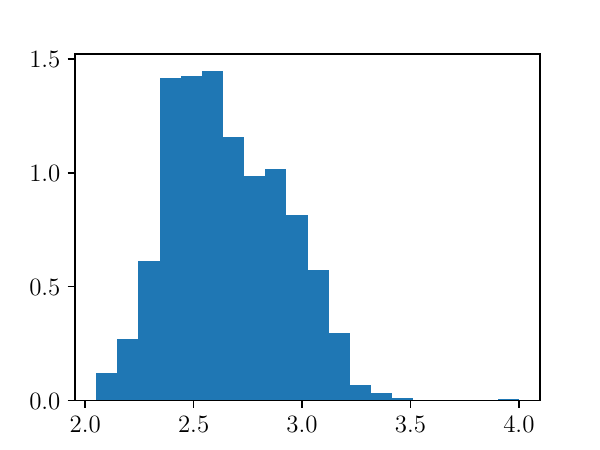}
        \caption{$\bP_{\Abar_{11}}$}
        \label{fig:alpha}
    \end{subfigure}
        \begin{subfigure}[b]{0.31\linewidth}
        \centering
        \includegraphics[width=\linewidth]{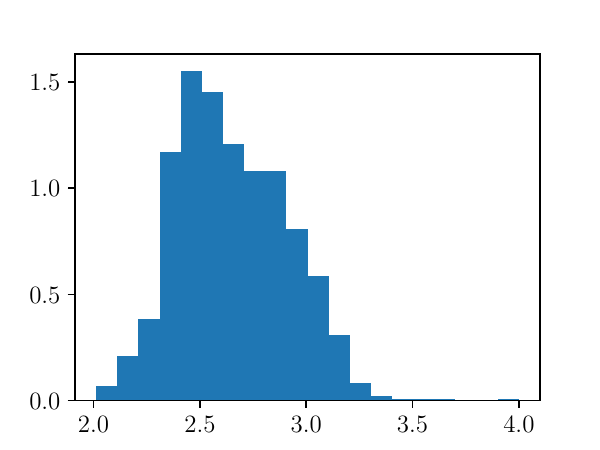}
        \caption{$\bP_{\Abar_{22}}$}
        \label{fig:hist}
    \end{subfigure}
    \begin{subfigure}[b]{0.31\linewidth}
        \centering
        \includegraphics[width=\linewidth]{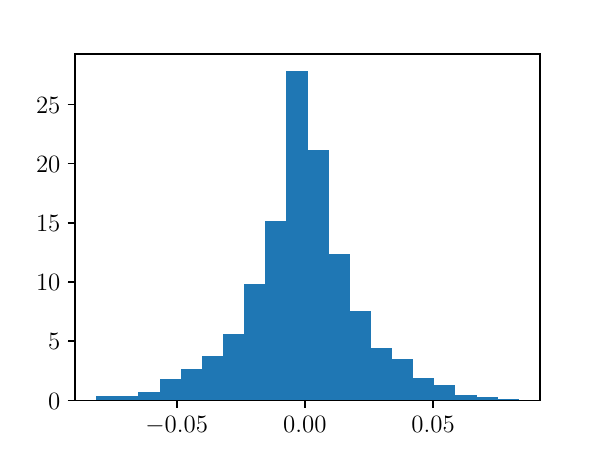}
        \caption{$\bP_{\Abar_{12}}$}
        \label{fig:loss}
    \end{subfigure}
    \begin{subfigure}[b]{0.31\linewidth}
        \centering
        \includegraphics[width=\linewidth]{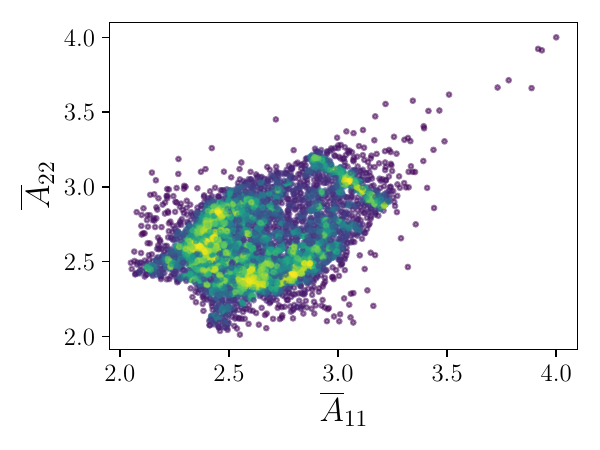}
        \caption{$\bP_{\Abar_{11}, \Abar_{22}}$}
        \label{fig:loss}
    \end{subfigure}
        \begin{subfigure}[b]{0.31\linewidth}
        \centering
        \includegraphics[width=\linewidth]{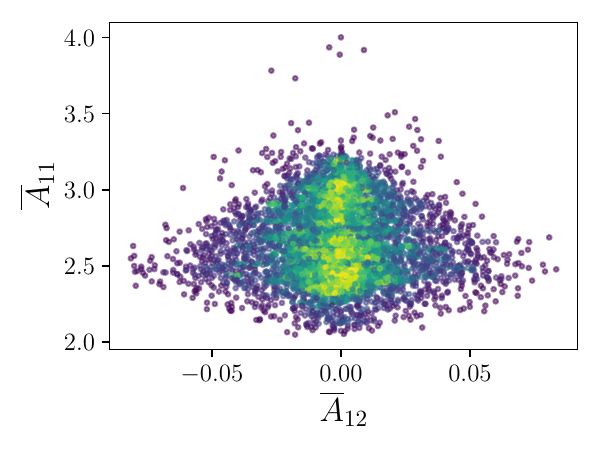}
        \caption{$\bP_{\Abar_{11}, \Abar_{12}}$}
        \label{fig:alpha}
    \end{subfigure}
        \begin{subfigure}[b]{0.31\linewidth}
        \centering
        \includegraphics[width=\linewidth]{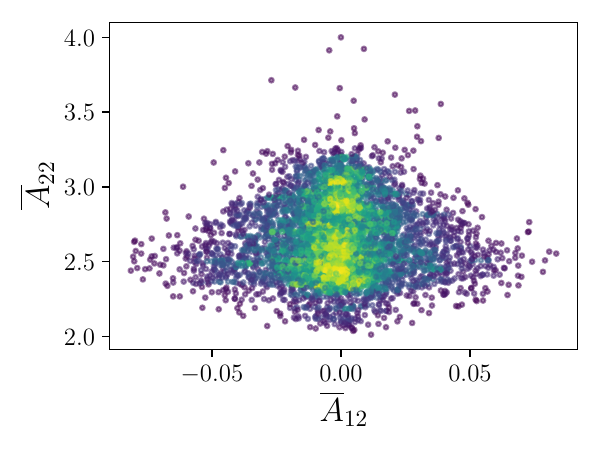}
        \caption{$\bP_{\Abar_{22}, \Abar_{12}}$}
        \label{fig:2D_Per_hists}
    \end{subfigure}
    \caption{Histograms and scatter plots of the observed data-distribution. In (d)-(f) the color is representative the local density of particles computed via the mean 10-neighbor distance to nearby particles.}
    \label{fig:Per_2D_AddW_Hists_scatter}
\end{figure}

\begin{figure}[t] 
    \centering
    % \vspace{1em}
    \begin{subfigure}[b]{0.32\linewidth}
        \centering
        \includegraphics[width=\linewidth]{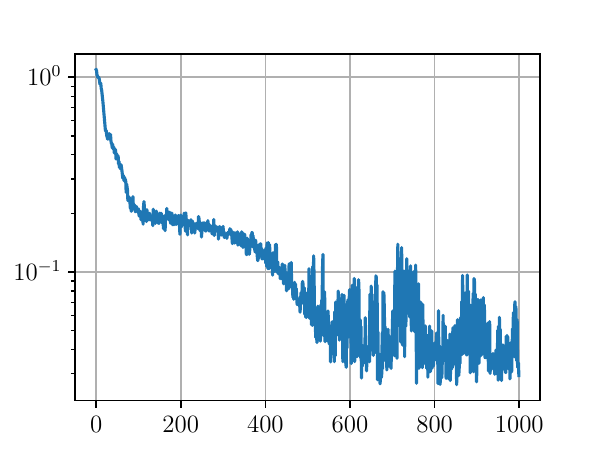}
        \caption{Objective}
        \label{fig:alpha}
    \end{subfigure}
        \begin{subfigure}[b]{0.32\linewidth}
        \centering
        \includegraphics[width=\linewidth]{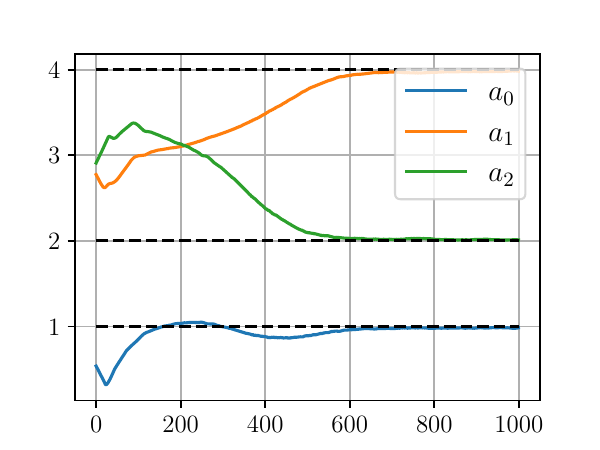}
        \caption{$a$}
        \label{fig:hist}
    \end{subfigure}
    % \hfill
    % \hspace{-1em}
    \begin{subfigure}[b]{0.32\linewidth}
        \centering
        \includegraphics[width=\linewidth]{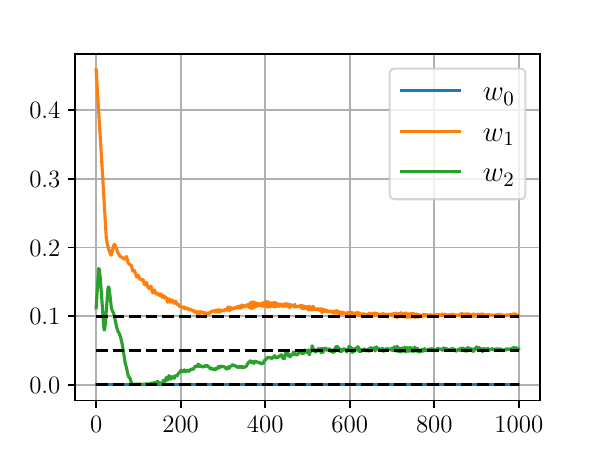}
        \caption{$w$}
        \label{fig:loss}
    \end{subfigure}
    \caption{Convergence of~\eqref{eq:objective_per} and $(w,a)$ for the i.i.d. periodic microstructure model}
    \label{fig:Per_2D_AddW_conv}
\end{figure}

\textbf{Data and Inference:}
We use $N=5\cdot10^3$ data samples. We set $a^\dagger=(1, 2, 4),$ $w^\dagger=(0,0.05, 0.1).$ To generate data we set $\tau$ to $10^{-20}$ to have sharp numerically discontinuous Voronoi cells. We do not include noise in this experiment.
Figure~\ref{fig:Per_2D_rnd_A11_samples} shows three randomly generated periodic Voronoi diagrams for the parameter settings $(w^\dagger, a^\dagger)$. Figure~\ref{fig:Per_2D_AddW_Hists_scatter} shows the histograms of the data distribution $\bP^{^N}_{\Abar}$ as scatter plots of combinations of the three data components. We fix the computational mesh to be $100\times100$ when representing $A(y)$ and in solving for the corrector function $\chi$ with finite elements in $\sF^\dagger$.

\textbf{Results:}
For inference we set $\tau$ to $10^{-3}$ to allow gradients to flow through the Voronoi construction and estimate $w$. 
Figure~\ref{fig:Per_2D_AddW_conv} shows the convergence of the objective, and $(w, a)$ against $(w^\dagger, a^\dagger).$ We run $5$ random restarts with $(w,a)$ independently drawn from $\Unif(0.01, 0.5)$ and $\Unif(0.5, 3)$ respectively, and independently componentwise.
The relative error on the recovered $a$ is $0.40\%$ and for $w$ it is $2.88\%$.

{
\vspace{1em}
% \subsubsection{Measurement }
\begin{remark}
   In practice, it may not be feasible to measure all three homogenized coefficient values $\Abar_{11}, \Abar_{22}, \Abar_{12}$, but only some subset of these.
    % We now report the parameter estimation errors when considering observations of only subsets of the three mentioned coefficient values. 
    With the same setup as described in this subsection, Table~\ref{tab:relative_errors} reports the relative errors in the recovery of $a$ and $w$ for various combinations of parameters observations, listed on the left. From these results, it can be inferred that in the prescribed setting, not every component marginal contains the same amount of information about $a,w$ and in practice, one may wish to prioritize observing certain components to maximize data efficiency. Nevertheless, observing more components yields better recovery of the ground truth.
    % %
    % $\Abar_{11}, \Abar_{22}$, the relative error in estimating $w$ is $3.15\%$, and for $a$ it is $0.836\%$;
    % %
    % when observing $A_{11}, A_{12}$ the relative error in estimating $w$ is $3.96\%$, and for $a$ it is $0.996\%$;
    % %
    % when observing $A_{11}$ the relative error in estimating $w$ is $9.07\%$, and for $a$ it is $3.16\%$;
    % %
    % when observing $A_{12}$ the relative error in estimating $w$ is $91.07\%$, and for $a$ it is $21.41\%$;
\begin{table}[t]
\centering
\caption{{Relative estimation errors for $w$ and $a$ based on observed homogenization components}}
\label{tab:relative_errors}
% \color{darkblue}
\begin{tabular}{lll}
% \hline
\textbf{Observed Component} & \textbf{Rel. Error in $w$} & \textbf{Rel. Error in $a$} \\ \hline
\vspace{-0.8em}\\
$\Abar_{11}, \Abar_{22}, \Abar_{12}$ & 2.88\%                                & 0.40\%                               \\
$\Abar_{11}, \Abar_{22}$ & 3.15\%                                 & 0.836\%                                \\
$\Abar_{11}, \Abar_{12}$             & 3.96\%                                 & 0.996\%                                \\
$\Abar_{11}$                     & 9.07\%                                 & 3.16\%                                 \\
$\Abar_{12}$                     & 91.07\%                                & 21.41\%                                \\ \hline
\end{tabular}
\end{table}
\end{remark}
}

\subsection{Locally Periodic Microstructures}\label{sec:loc_per}

In this section we explore a more physically realistic data-acquisition setup. We imagine a single large specimen with a spatially varying locally periodic microstructure. That is, $A^\varepsilon(x/\varepsilon, x)$ is assumed to vary both on a small scale, of $\mathcal{O}(\varepsilon)$, and the macroscopic scale of $\D$, of $\mathcal{O}(1)$. In this setting, we show how distributional inversion can leverage natural spatial variations in microstructure for statistical characterization. To employ distributional inverse homogenization simply requires that the measurements of bulk material properties are at points
distant enough in the specimen to be approximately i.i.d.\,.
The generative model used, and in particular its dependence on the material properties $a$ and the additive Voronoi weights $w$, remain unchanged from Section~\ref{sec:2D_per_hom_i.i.d.}. 
\begin{figure}[t] 
    \centering
    % \vspace{1em}
    \hspace{-1em}
        \begin{subfigure}[b]{0.32\linewidth}
        \centering
        \includegraphics[width=\linewidth]{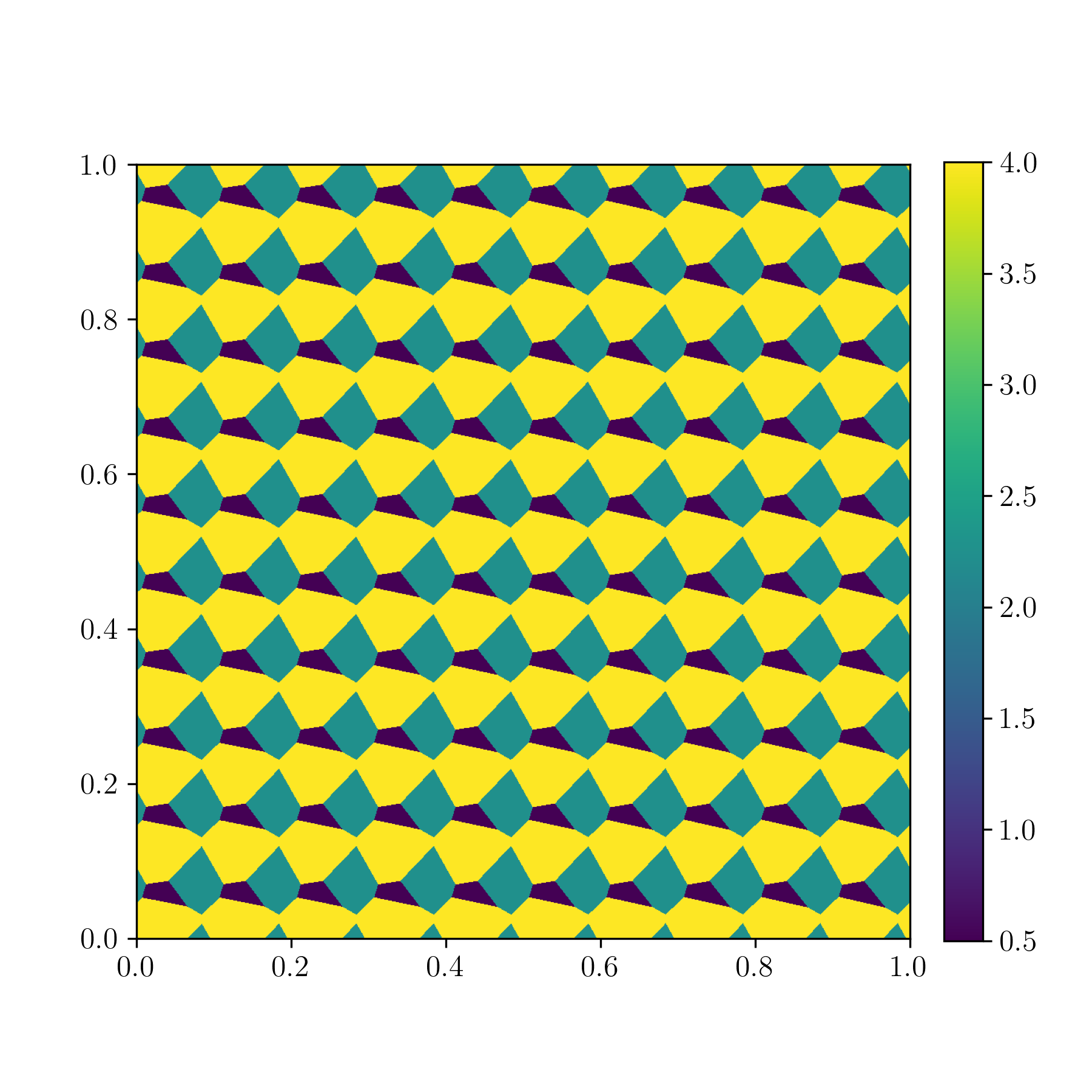}
        \caption{$10\times 10$ Voronoi, location 1.}
        \label{fig:hist}
    \end{subfigure}
    % \hfill
    % \hspace{-1em}
    \begin{subfigure}[b]{0.32\linewidth}
        \centering
        \includegraphics[width=\linewidth]{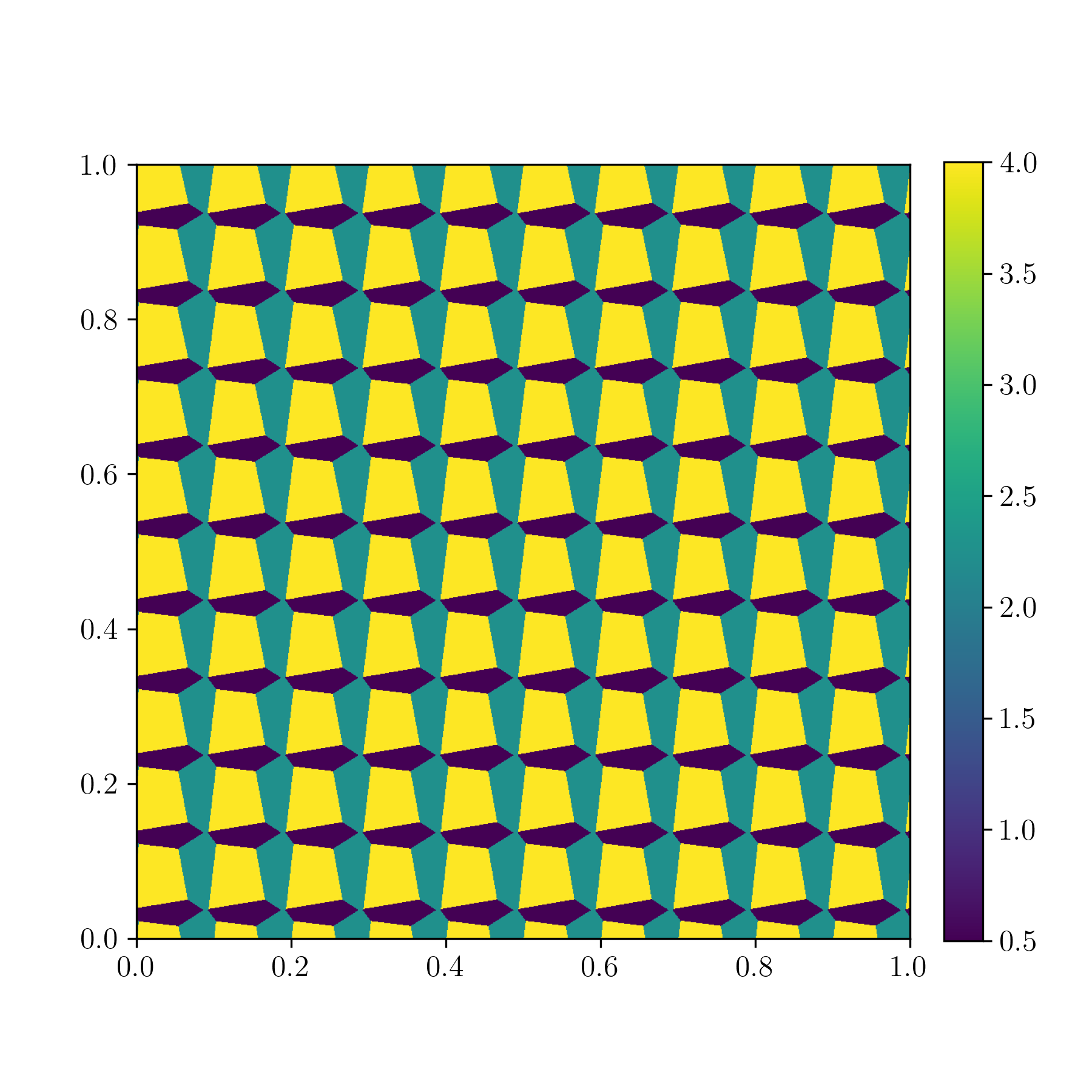}
        \caption{$10\times 10$ Voronoi, location 2.}
        \label{fig:loss}
    \end{subfigure}
    \begin{subfigure}[b]{0.32\linewidth}
        \centering
        \includegraphics[width=\linewidth]{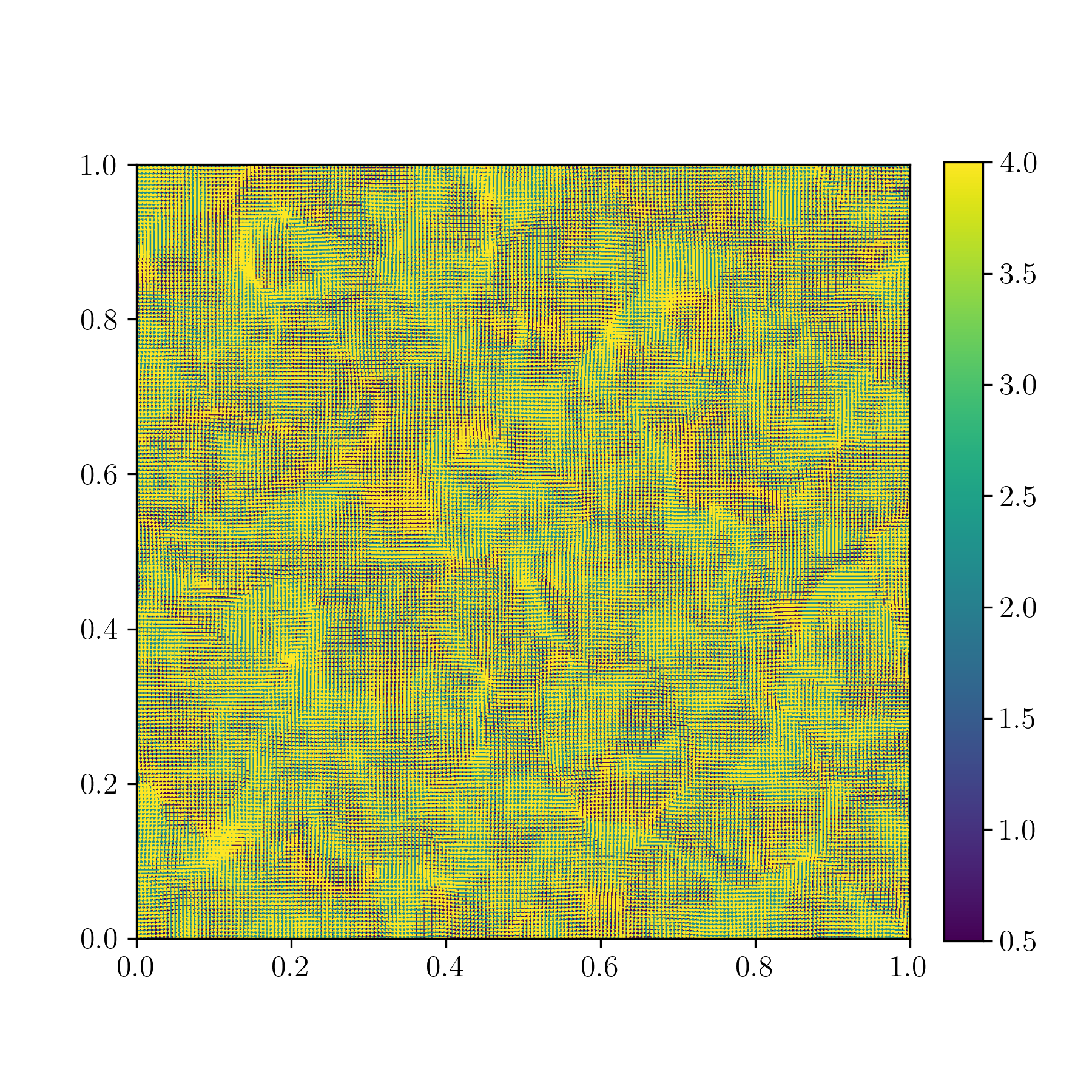}
        \caption{$8000\times8000$ Voronoi}
        \label{fig:loc_per_2D_vor_samples:c}
    \end{subfigure}
    \caption{Locally periodic Voronoi construction (a-b) $10\times 10$ Voronoi cells, (c) $8000\times8000$ locally periodic Voronoi cells (plot is 40 fold cell-level sub-sampled and 5-fold in-cell sub-sampled for display -- making spatial correlation appear to be on a shortened scale). The total specimen is $2\cdot10^4\times2\cdot10^3$ cells.}
    \label{fig:loc_per_2D_vor_samples}
\end{figure}

Here, in the way we generate the observed dataset, $\smash{\{\Abar^{(n)}\}_{n=1}^N}$, the Voronoi seeds are not randomly sampled on each cell, but rather, we generate the seeds in such a way such that contiguous cells have Voronoi centers in close but not identical layout. This results in a locally periodic microstructure. We will attempt to recover microstructural information based on homogenized measurements at locations in the large specimen which are far apart enough so that the data becomes close to i.i.d.\,.

\begin{figure}[t] 
    \centering
    \hspace{-1em}
        \begin{subfigure}[b]{0.32\linewidth}
        \centering
        \includegraphics[width=\linewidth]{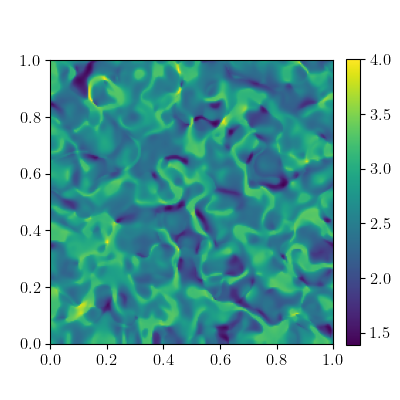}
        \caption{$\Abar_{11}(x)$}
        \label{fig:hist}
    \end{subfigure}
    \begin{subfigure}[b]{0.32\linewidth}
        \centering
        \includegraphics[width=\linewidth]{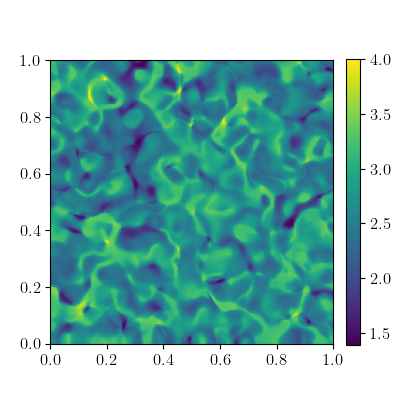}
        \caption{$\Abar_{22}(x)$}
        \label{fig:loss}
    \end{subfigure}
    \begin{subfigure}[b]{0.32\linewidth}
        \centering
        \includegraphics[width=\linewidth]{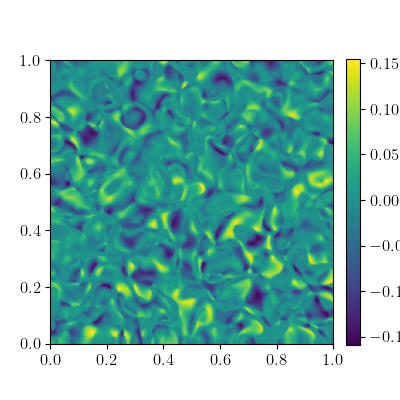}
        \caption{$\Abar_{12}(x)$}
        \label{fig:alpha}
    \end{subfigure}
    \caption{Homogenized locally periodic Voronoi construction, $\overline{A}:\S\rightarrow\sfM_2$ corresponding to applying periodic homogenization to every cell of Figure~\ref{fig:loc_per_2D_vor_samples:c}.}
    \label{fig:Homogenized_loc_per_2D_vor_samples}
\end{figure}

The construction of the locally periodic microstructure follows a \textit{copula} construction on the Voronoi seeds, where one can independently specify one- and two -point statistics: spatial correlation and point-wise marginals. It is straightforward to construct Gaussian random fields with specified one and two-point statistics~\cite{sklar1959fonctions, nelsen2006introduction} and such a field may be used to define a spatially correlated distribution for the Voronoi seeds.

We start by describing the Gaussian random field construction.
Let $\S =[0,\L]^2 =\bigcup_{j\in\mathcal{I}}\S_{j}$ a finite partition with each $\S_{j}$ a square of width $\mathsf{w}$ with  local origin at $p_{j}$. Draw $2M$ i.i.d. random functions $s\in C(\S;\bR)$ via 
\begin{align}\label{eq:gauss_meas_field}
    s^{(\Bigcdot)}\sim \cN(0, (\ell^{-2}\I-\Delta)^{-\nu}), \quad k={1,2},\quad m=1, \hdots, M.
\end{align}
Notation $s^{(m,k)}$ denotes a specific element from this collection of random fields.
Here $\cN$ denotes the normal distribution and $\Delta$ is the Laplacian operator defined on an extended domain $\S'\supset\S$ equipped with homogeneous Neumann boundary conditions; it has eigenfunctions $\varphi_{i,o}(x) = \cos(i\pi x_1)\cos(o\pi x_2)$ for $i,o$ positive integers.
This function-space definition of a Gaussian measure can be approximately sampled via the truncated Karhunen–Lo\`eve expansion \cite{dashti2011uncertainty}
\begin{align}
    s^{(\Bigcdot)}(x) = \sum_{i,o=1}^{J_\mathrm{KL}}\left((i^2+o^2)\pi^2+\ell^{-2}\right)^{-\nu/2}\xi^{(i,o)}\varphi_{i,o}(x);\quad \xi^{(i,o)}\sim\cN(0,1)\,\text{i.i.d.}
\end{align}
This Gaussian random field has a Gaussian distribution at every point $x \in \S.$ We now invert to find a spatially varying  copula field which is uniformly distributed
at every point $x \in \S$. The coupla is realized via $s^{(\Bigcdot)}$ as
\begin{align}
    \zeta^{(\Bigcdot)}(x)= F_{\mathsf{N}}( s^{(\Bigcdot)}(x))
\end{align}
where $F_{\mathsf{N}}$ is the scale-normalized CDF of the one-dimensional Gaussian, where the scaling is determined by the variance of the one-dimensional Gaussian.
Thus, $ \zeta^{(\Bigcdot)}$ will have, ignoring nonstationary boundary effects of~\eqref{eq:gauss_meas_field}, point-wise distribution $\Unif(0, 1)$, and will be spatially correlated with two-point statistics controlled by $\ell, \nu$.

We now use the copula field to construct a slowly varying Voronoi tesselation. To this end
we write $\zeta^{(\Bigcdot)}\sim\Xi(\ell,\nu; \S)$, where $\S$ is included as an argument to emphasize the restriction of the copula to $\S\subset\S'$. Now specify the Voronoi seeds, for $m=1\hdots, M$ and $k\in\{1,2\}$, 

\begin{subequations}
\begin{align}
    \zeta^{(m,k)}(x) &\sim\Xi(\ell,\nu; \S)\\
    v^{(j,m,k)}&=p_{j}+\sfw\zeta^{(m,k)}(p_j). 
\end{align}
\end{subequations}
Recall that each $\S_{j}$ a square of width $\mathsf{w}$ with  local origin at $p_{j}$ so that the construction centers and scales the seeds locations with respect to the $\S_j.$  Index $k \in \{1,2\}$ allows construction of a point in $\S_{j} \subset \bR^2$, $j$ indexes the different cells and $m$ indexes different random Voronoi seeds. The special case where $s^{(m,k)}(x)$ is constant with respect to $x\in\S$, rather than a Gaussian random field, yields a periodic Voronoi diagram on $\S$.

\begin{figure}[t] 
    \centering
    \hspace{1.1em}
    \begin{subfigure}[b]{0.31\linewidth}
        \centering
        \includegraphics[width=\linewidth]{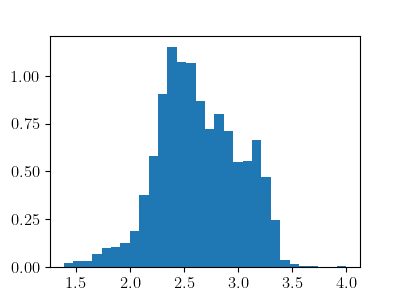}
        \caption{$\bP_{\Abar_{11}}$}
        \label{fig:alpha}
    \end{subfigure}
        \begin{subfigure}[b]{0.31\linewidth}
        \centering
        \includegraphics[width=\linewidth]{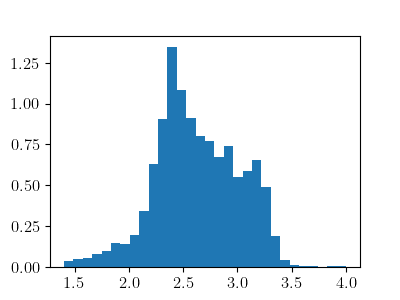}
        \caption{$\bP_{\Abar_{22}}$}
        \label{fig:hist}
    \end{subfigure}
    \begin{subfigure}[b]{0.31\linewidth}
        \centering
        \includegraphics[width=\linewidth]{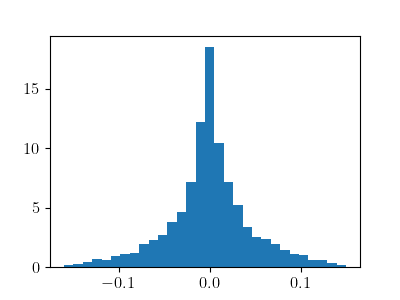}
        \caption{$\bP_{\Abar_{12}}$}
        \label{fig:loss}
    \end{subfigure}
    \begin{subfigure}[b]{0.31\linewidth}
        \centering
        \includegraphics[width=\linewidth]{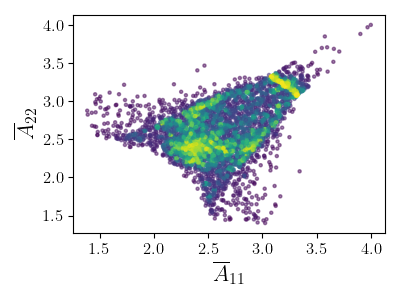}
        \caption{$\bP_{\Abar_{22},\Abar_{11}}$}
        \label{fig:loss}
    \end{subfigure}
        \begin{subfigure}[b]{0.31\linewidth}
        \centering
        \includegraphics[width=\linewidth]{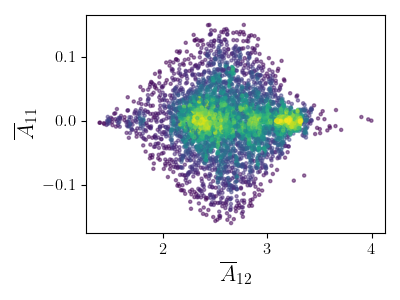}
        \caption{$\bP_{\Abar_{11}, \Abar_{12}}$}
        \label{fig:alpha}
    \end{subfigure}
        \begin{subfigure}[b]{0.31\linewidth}
        \centering
        \includegraphics[width=\linewidth]{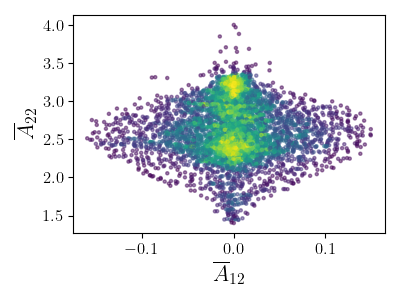}
        \caption{$\bP_{\Abar_{22}, \Abar_{12}}$}
        \label{fig:2D_LOCPer_hists}
    \end{subfigure}
    \caption{Histograms and scatter plots of the observed data-distribution (before adding noise for inference). In (d)-(f) the color is representative the local density of particles computed via the mean 10-neighbor distance to nearby particles.}
    \label{fig:LocPer_2D_Hists_scatter}
\end{figure}

We have just used the copula construction to specify Voronoi seed locations. Given the seeds, we use the Laguerre-Voronoi diagrams of Subsection~\ref{sec:2D_per_hom_i.i.d.}, and hence with $v^{(n,1:M,1:2)}$, $n$ in the index set of partition cells $\mathcal{I}$, we can compute $A^{(n)}$ and $\widetilde{A}^{(n)}$ the same way as in~\eqref{eq:per_vor_Afield_multWeight} and~\eqref{eq:cont_relax_A_crystalLag}. Figure~\ref{fig:loc_per_2D_vor_samples} shows the locally periodic construction: the first two plots is the zoomed in Voronoi diagram at two different locations showing the local periodicity, the third plot is a larger window of the Voronoi construction where the spatial variations become visible. Figure~\ref{fig:Homogenized_loc_per_2D_vor_samples} shows the three components of the periodically homogenized coefficient in each cell in Figure~\ref{fig:loc_per_2D_vor_samples:c}, exposing the spatial variability of the copula construction.

To mirror a physically testable setup, we assume the local $\Abar^{(n)}$ are noisily measured at sparse distant locations with the noise distribution $\bP_\xi:=\cN(0,\sigma_{\xi}^2\I_3)$. We attempt to estimate $(w,a)$ from sparsely observed homogenized coefficients using~\eqref{eq:objective_per}.

\textbf{Data and Inference:}
The correlation lengthscale is set to $\ell=1/100$ and $\nu=4$ (this information is not inferred from data or needed in any way for inference, the only assumption is that measurement locations are distant enough to be approximately i.i.d.) and we use $100$ KL expansion terms to sample from $\Xi(\ell, \nu;\S)$, the cell size $\Q$ is set to $5\cdot 10^{-5}$, meaning the full Voronoi diagram contains $4\cdot10^{8}$ cells. The noise $\sigma_\xi=10^{-2}$.
 We set $a^\dagger=(0.5, 2.25, 4),$ $w^\dagger=(0, 1.38\cdot 10^{-10}, 2.78\cdot 10^{-9}).$ To generate data we set $\tau=10^{-20}$ to have sharp numerically discontinuous Voronoi cells. We use $N=5\cdot10^3$ coefficient observations. We use $N=5\cdot10^3$ coefficient observations. Figure~\ref{fig:LocPer_2D_Hists_scatter} shows the histograms of the data distribution $\bP^{^N}_{\Abar}$ as scatter plots of combinations of the three data components.

\begin{figure}[t] 
    \centering
    \begin{subfigure}[b]{0.32\linewidth}
        \centering
        \includegraphics[width=\linewidth]{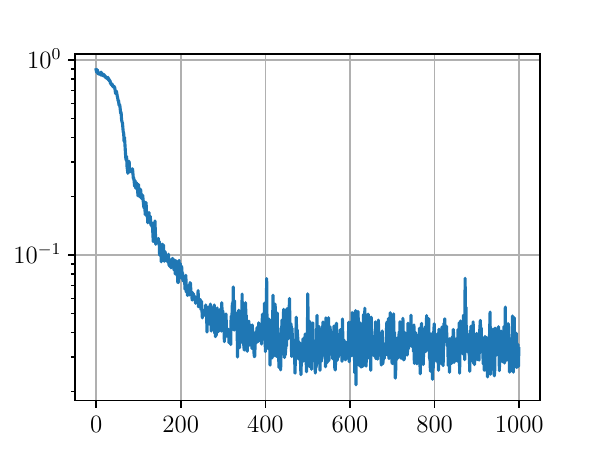}
        \caption{Objective}
        \label{fig:alpha}
    \end{subfigure}
        \begin{subfigure}[b]{0.32\linewidth}
        \centering
        \includegraphics[width=\linewidth]{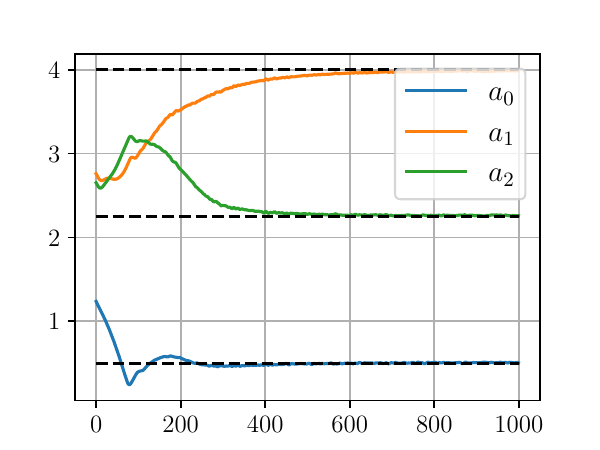}
        \caption{$a$}
        \label{fig:hist}
    \end{subfigure}
    \begin{subfigure}[b]{0.32\linewidth}
        \centering
        \includegraphics[width=\linewidth]{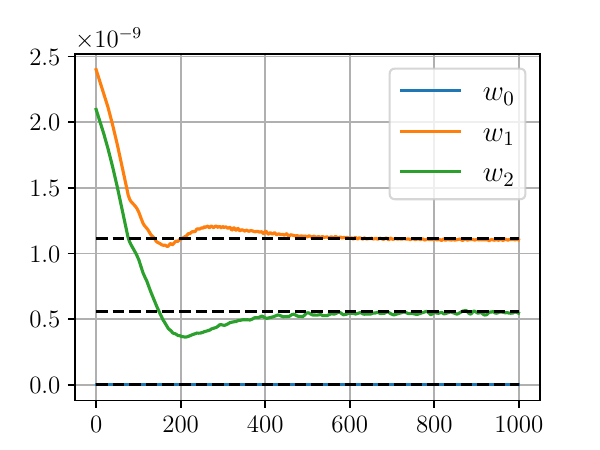}
        \caption{$w$}
        \label{fig:loss}
    \end{subfigure}
    \caption{Convergence of \eqref{eq:objective_per} and $(w, a)$  for the locally periodic microstructure.}
    \label{fig:Loc_Per_2D_conv}
\end{figure}

\textbf{Results:} For inference we set $\tau$ to $10^{-3}$ to allow gradients to flow through the Voronoi construction and estimate $w$. 
Figure~\ref{fig:Loc_Per_2D_conv} shows the convergence of the objective, and $(w, a)$ against $(a^\dagger, w^\dagger).$
The inference ran in $5\times9.33$ hours. The relative error on the recovered $w$ is 1.76\%, and for $a$ is $0.35\%$.
The relative error on $a$ is 0.383\% and the relative error on $w$ is 0.97\%.

\section{2D Stochastic Homogenization}\label{sec:stoch_hom}
In this section we explore applications of distributional inverse homogenization in 2D stochastic homogenization. As with the previous section, we focus on the Voronoi microstructure model. We first consider an i.i.d. data generating mode, followed by a locally stationary-ergodic Voronoi construction, again based on statistical copulas, supporting Contributions~\ref{con:C2} and \ref{con:C4}. To efficiently implement our inference scheme we propose an adaptive surrogate learning scheme, answering Contribution~\ref{con:C3}.

\subsection{Stationary-Ergodic i.i.d. Microstructures}

We first introduce the setup for the i.i.d. microstructure model, followed by the exposition of the surrogate modelling method, and we finish with details on the dataset and results.

\textbf{Setup:}
In this subsection we revisit Dirichlet-distributed volume fractions, $\rho$;
we make the identification $\varpi=(\rho, a)$ in~\eqref{eq:objective_stoch}
and attempt to invert the map $(\alpha, a)\mapsto\bP_{\,\Abar}$. Prescribing the exact volume fraction of a material type in a Voronoi construction requires direct manipulation of Voronoi seeds,\footnote{If keeping other parameter such as Voronoi-weights fixed.} we do not attempt this. Instead, we approximately prescribe volume fractions by randomly assigning cell material-types by sampling from a categorical distribution with category weights sampled from $\sDir(\alpha)$. This use of the Dirichlet distribution to specify the volume fractions mirrors Section~\ref{sec:1Ddirichlet}. In the following, for $\rho\in\Delta_{M-1}$, we have $r\sim\sCat(\rho)$ where $r=(0, \hdots, 1, \hdots, 0)$ is a $M$-dimensional vector of zeros with a $1$ in position $k$ with probability $\rho_k$. We define i.i.d. generative model for the material microstructure in the $\sfn$-cell $\Q$ via 
\begin{subequations}\label{eq:stoch_i.i.d._model}
\begin{align}
    \rho&\sim\sDir(\alpha),\label{eq:i.i.d.SotchA:1}\\
    v^{(j)}&\sim\Unif(\Q),\label{eq:i.i.d.SotchA:2}\\
    r^{(j)} &\sim \sCat(\rho),\label{eq:i.i.d.SotchA:3}\\
    A(x; a, \rho, v^{(1:\Jvor)}) &= \I_2\sum_{j=1}^{\Jvor} \langle a, r^{(j)}\rangle \bb1_{\Q_j(v^{1:\Jvor})}(x),\label{eq:i.i.d.SotchA:4}
\end{align}
\end{subequations}
where we define
\begin{align}
    \Q_j(v) &= \{x\in Q: \|x- v^{(j)}\|\leq \|x- v^{(k)}\|; k,j\in\{1, \hdots,\Jvor\};  k\neq j\}.\label{eq:i.i.d.SotchA:5}
\end{align}
This procedure can be explained informally as follows: \eqref{eq:i.i.d.SotchA:1} sample a volume fraction vector for a Dirichlet distribution; \eqref{eq:i.i.d.SotchA:2} sample $\Jvor$ Voronoi seeds for the cell $\Q$; \eqref{eq:i.i.d.SotchA:3} draw a categorical vector of all zeros expect for one 1 in position $m$ with probabilities given by $\rho$; \eqref{eq:i.i.d.SotchA:4} assemble the microstructure field by constructing the Voronoi diagram by~\eqref{eq:i.i.d.SotchA:5} and assigning material value to each cell via the inner product $\langle a, r^{(j)}\rangle$ (notice $r^{(j)}=e^m$ simply picks out the value $a_m$ from $a$).
Figure~\ref{fig:Stochi.i.d._2D_rnd_A11_samples} shows component $A_{11}(y)$ for two independent random coefficient fields as sampled per~\eqref{eq:stoch_i.i.d._model} with $\Jvor=3\cdot10^4$,

\begin{figure}[t] 
    \centering
    % \vspace{1em}
    \begin{subfigure}[b]{0.48\linewidth}
        \centering
        \includegraphics[width=\linewidth]{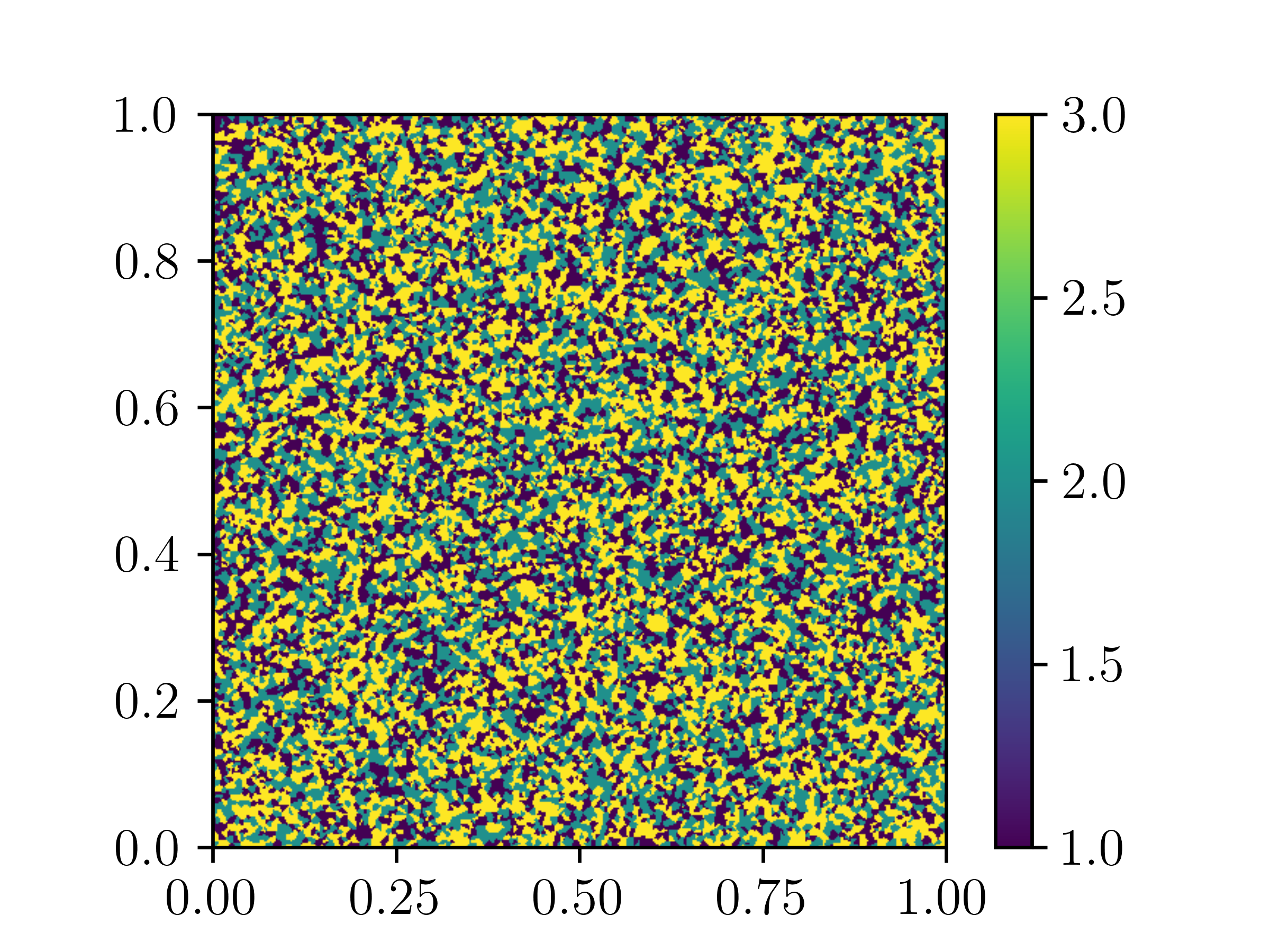}
        \caption{Random S.E. cell $A^{(1)}_{11}(y)$}
        \label{fig:alpha}
    \end{subfigure}
        \begin{subfigure}[b]{0.48\linewidth}
        \centering
        \includegraphics[width=\linewidth]{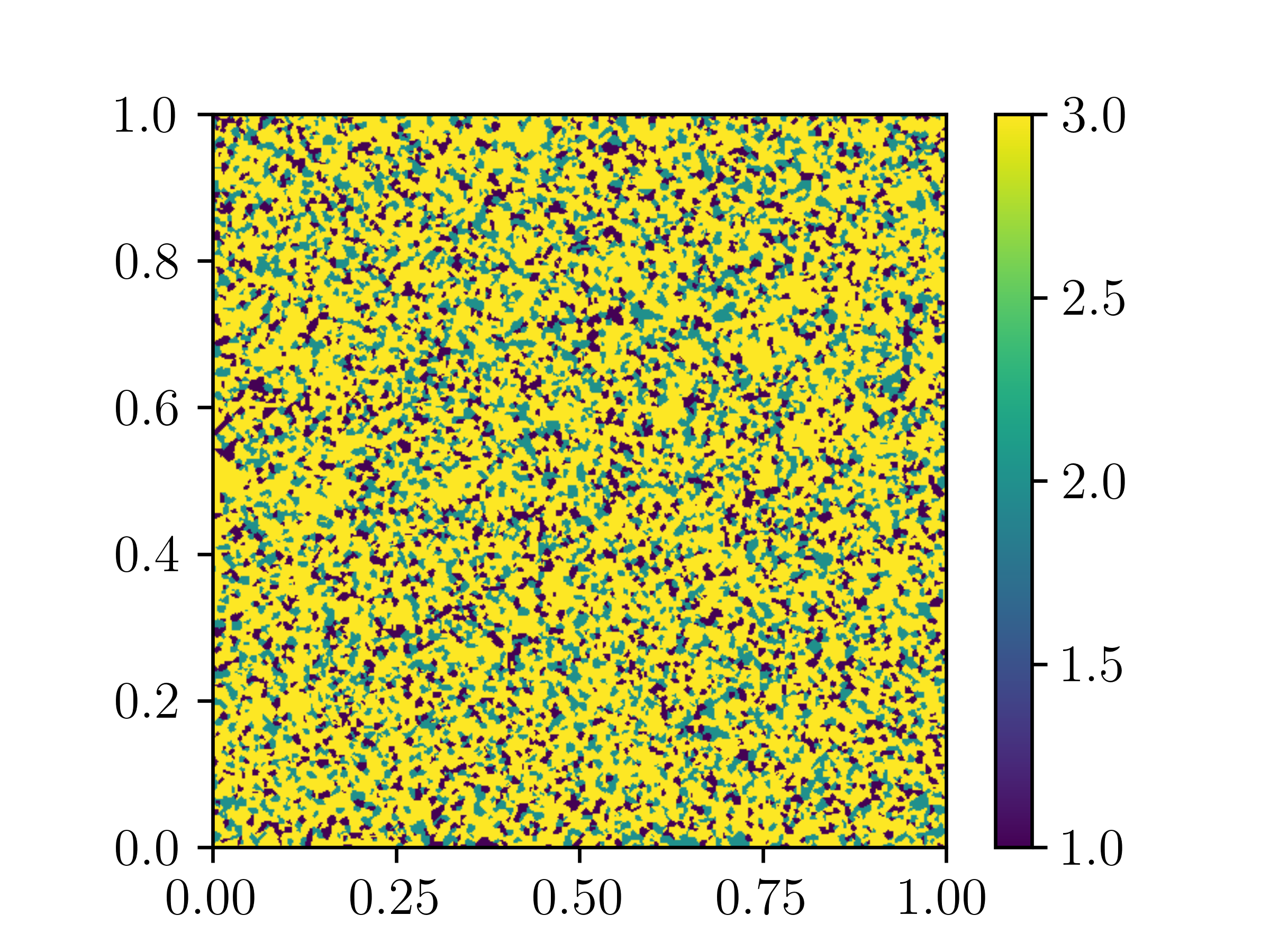}
        \caption{Random S.E. cell $A^{(2)}_{11}(y)$}
        \label{fig:hist}
    \end{subfigure}
    \caption{Randomly generated stationary-ergodic (S.E.) microstructure for $(a,\rho)$ randomizing over Voronoi seed locations in $\Q$ and $\rho\sim\sDir(\alpha)$. These are realizations of the i.i.d. stationary-ergodic material model.}
    \label{fig:Stochi.i.d._2D_rnd_A11_samples}
\end{figure}
\begin{figure}[t] 
    \centering
    \hspace{1.1em}
    \begin{subfigure}[b]{0.31\linewidth}
        \centering
        \includegraphics[width=\linewidth]{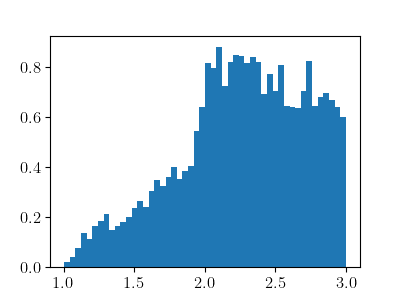}
        \caption{$\bP_{\Abar_{11}}$}
        \label{fig:alpha}
    \end{subfigure}
        \begin{subfigure}[b]{0.31\linewidth}
        \centering
        \includegraphics[width=\linewidth]{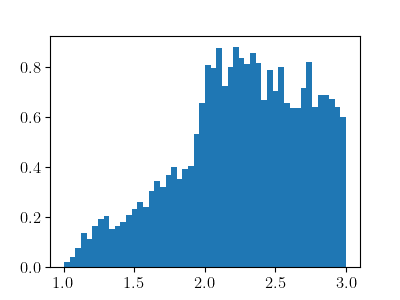}
        \caption{$\bP_{\Abar_{22}}$}
        \label{fig:hist}
    \end{subfigure}
    \begin{subfigure}[b]{0.31\linewidth}
        \centering
        \includegraphics[width=\linewidth]{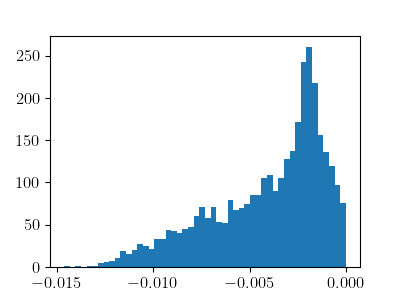}
        \caption{$\bP_{\Abar_{12}}$}
        \label{fig:loss}
    \end{subfigure}
    \begin{subfigure}[b]{0.31\linewidth}
        \centering
        \includegraphics[width=\linewidth]{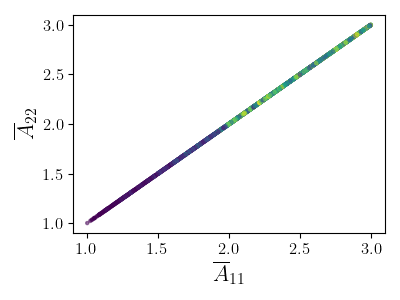}
        \caption{$\bP_{\Abar_{22},\Abar_{11}}$}
        \label{fig:loss}
    \end{subfigure}
        \begin{subfigure}[b]{0.31\linewidth}
        \centering
        \includegraphics[width=\linewidth]{figs/Stoch_Hom_2D/scatter_A11_A21_data.png}
        \caption{$\bP_{\Abar_{11}, \Abar_{12}}$}
        \label{fig:alpha}
    \end{subfigure}
        \begin{subfigure}[b]{0.31\linewidth}
        \centering
        \includegraphics[width=\linewidth]{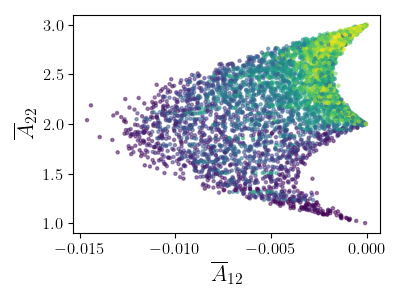}
        \caption{$\bP_{\Abar_{22}, \Abar_{12}}$}
        \label{fig:2D_Stoch_hists}
    \end{subfigure}
    \caption{Histograms and scatter plots of the observed data-distribution for the i.i.d. stationary-ergodic material model. In (d)-(f) the color is representative the local density of particles computed via the mean 10-neighbor distance to nearby particles.}
    \label{fig:Per_2D_Stoch_Hists_scatter}
\end{figure}

\textbf{Surrogate Training:}
The objective is to invert the map $(\alpha, a)\mapsto \bP_{\,\Abar}$.
We are now presented with two challenges for numerical implementation: i) the forward map is very computationally expensive as it requires a Monte-Carlo expectation over Voronoi seeds; ii) we require gradients w.r.t to $\alpha$ for a computation which passes via the discontinuous categorical distribution. Challenge ii) could be resolved by smoothing out discontinuities, similarly to Section~\ref{sec:2D_per_hom}, using~\cite{jang2017categorical,figurnov2018implicit},
but then we are still faced with i). The solution we explore in this section is to concurrently learn a neural network directly approximating $\cG^\dagger$ of~\eqref{sec:background:stoch_hom} with $\cG^{\phi^\star}:(\rho, a)\mapsto \Abar$ with adaptive surrogate training as outlined in Section~\ref{sec:background:surrogate}.
In estimating $\theta$ we run a gradient-flow on a pseudo-time $\st$, representing the steps in the optimization routine. At a given step $\st$ we update our surrogate model $\cG^{\phi_t^\star}$ with
\begin{subequations}\label{eq:cummul_emp_dataset_SH}
\begin{align}
     \phi^\star_\st &= \underset{\phi}{\argmin}\;\bE_{(\rho, a), \Abar\sim\bP^{\st}}\left[\|\cG^\phi(\rho, a) - \Abar\|^2 \right],\label{eq:cummul_emp_dataset_SH:1}\\
     \bP^\st &= \frac{1}{\st + N_\mathrm{pre.} }\sum_{\sfi=1-N_\mathrm{pre.}}^{\st} \delta_{((\rho, a)^{(\sfi)}, \cG^\dagger((\rho, a)^{(\sfi)}))},
     \;\begin{cases}
         (\rho^{(\sfi)}, a^0);\,\rho^{(\sfi)}\sim\sDir(\alpha^0),\quad \mathrm{for}\; \sfi\leq0,\\
         (\rho^{(\sfi)},a^\sfi);\,\rho^{(\sfi)}\sim\sDir(\alpha^\sfi),\quad\; \mathrm{for}\; 0<\sfi\leq \st.\end{cases}\label{eq:cummul_emp_dataset_SH:2}
\end{align}
\end{subequations}
In~\eqref{eq:cummul_emp_dataset_SH:1} we select $\phi^\star$ at iteration $\st$ to minimize the input-output error for pairs of training data drawn from the adaptive dataset which is defined via~\eqref{eq:cummul_emp_dataset_SH:2}.
Notice, we do not use the exact $\phi^\star(\theta)$ in~\eqref{eq:surrogate_active}  as we wish to avoid differentiating $\cG^\dagger$, hence the pseudo-time $\st$ formulation. This particular time stepping scheme is proposed, in a related context, in~\cite{vadeboncoeur2025efficient}. In implementing this scheme we withhold updating $\bP^\st$ until $100$ new samples $(\rho^{(\Bigcdot)}, a^{\Bigcdot})$ are ready for evaluation under $\cG^\dagger$, to make the most of hardware parallelization. The surrogate model, $\cG^\phi$, is a 4-layer, 32-neuron, feedforward network with ``swish'' activation functions~\cite{ramachandran2017searching} and a layer-wise Lipschitz number constrained to be under 50~\cite{gouk2021regularisation}.
\vspace{1em}

{
\begin{remark}
    In a similar spirit to Remark~\ref{rem:HS_Bounds_Per}, we comment on bounds on the distribution of homogenized coefficients. In this section, volume fractions are random vectors drawn from a Dirichlet distribution; for any parameter setting of the Dirichlet distribution, it is always possible to draw a vector of volume fractions which represents a ``pure material'' (volume fraction of the form $(1,0,0)$, $(0,1,0)$, or $(0,0,1)$). Hence, the lower and upper bounds over the distribution of bulk material properties in this setting are simply $a^-=\min(a^\dagger)$ and $a^+=\max(a^\dagger)$. This is readily visible from Figure~\ref{fig:Per_2D_Stoch_Hists_scatter} (d)-(f) by noting that the data-generating $a^\dagger=(1,2,3)$.
\end{remark}
}

\textbf{Data and Inference:}
We generate $5\cdot10^3$ data samples to be used for inference. We do not include noise is this experiment. We set $\alpha^\dagger=(0.66, 1.32, 2)$ and $a^\dagger=(1,2,3)$. The Voronoi construction is fully discontinuous. Figure~\ref{fig:Per_2D_Stoch_Hists_scatter} shows the data distribution in terms of histograms and scatter plots. We set $\Jvor=3\cdot10^4$ seeds per cell $\Q_\sfn$. The data is generated using single-sample Monte-Carlo estimate of $\Abar^{(n)}$, in contrast to this, the i.i.d. model used for inference uses $\Javg=10$ Monte-Carlo samples.

\begin{figure}[t] 
    \centering
    % \vspace{1em}
    \hspace{-1em}
        \begin{subfigure}[b]{0.42\linewidth}
        \centering
        \includegraphics[width=\linewidth]{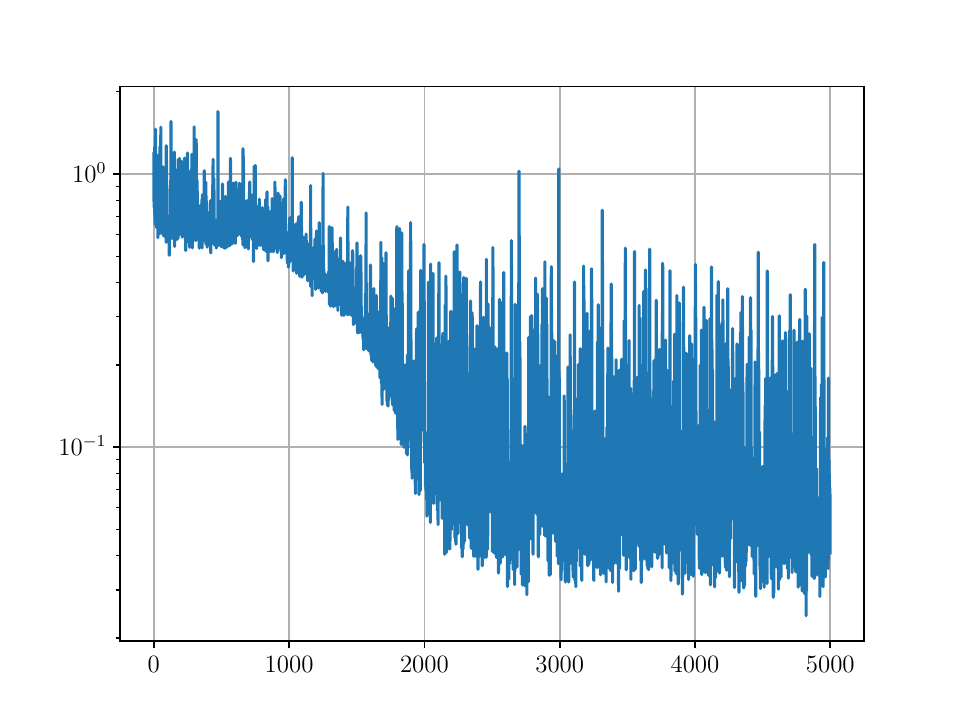}
        \caption{\eqref{eq:objective_stoch}}
        \label{fig:hist}
    \end{subfigure}
    % \hfill
    % \hspace{-1em}
    \begin{subfigure}[b]{0.42\linewidth}
        \centering
        \includegraphics[width=\linewidth]{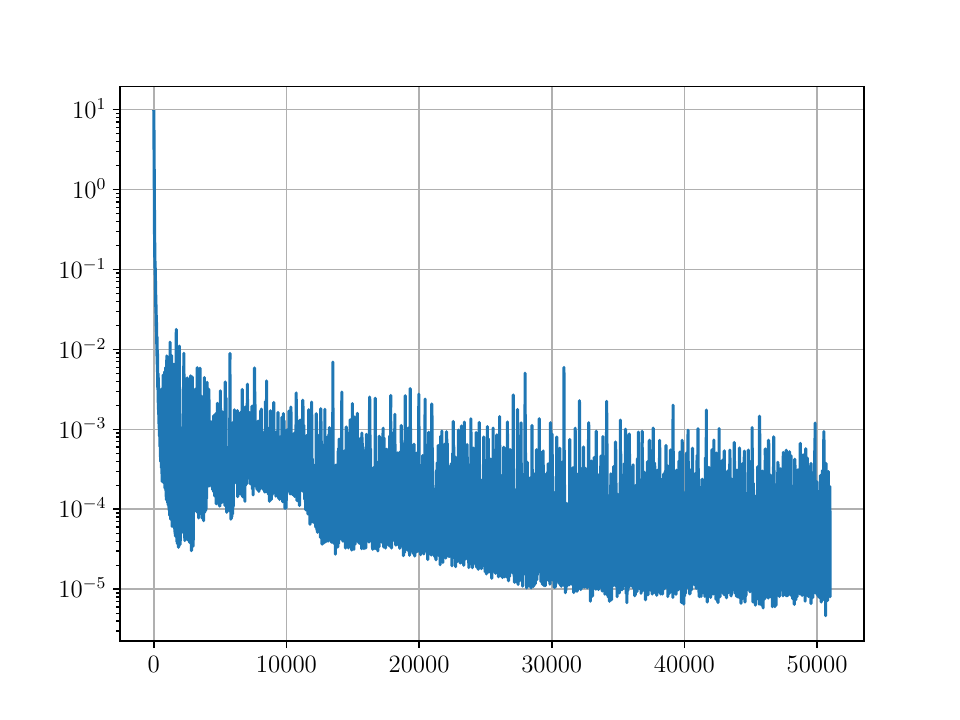}
        \caption{Surrogate objective~\eqref{eq:cummul_emp_dataset_SH:1}}
        \label{fig:loss}
    \end{subfigure}
    \begin{subfigure}[b]{0.42\linewidth}
        \centering
        \includegraphics[width=\linewidth]{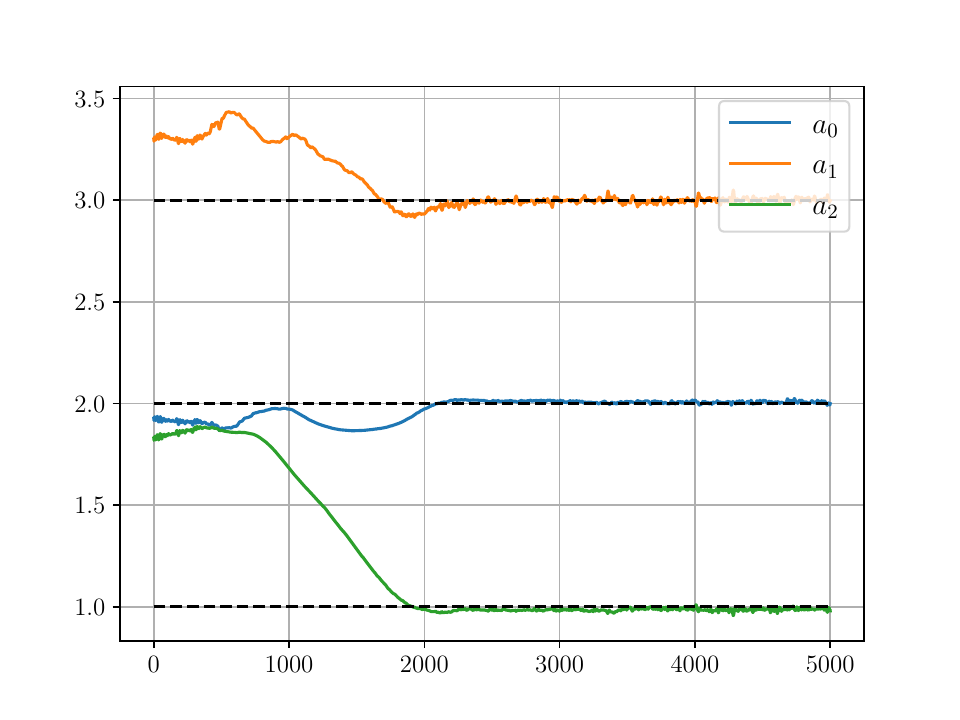}
        \caption{$a$}
        \label{fig:alpha}
    \end{subfigure}
    \begin{subfigure}[b]{0.42\linewidth}
        \centering
        \includegraphics[width=\linewidth]{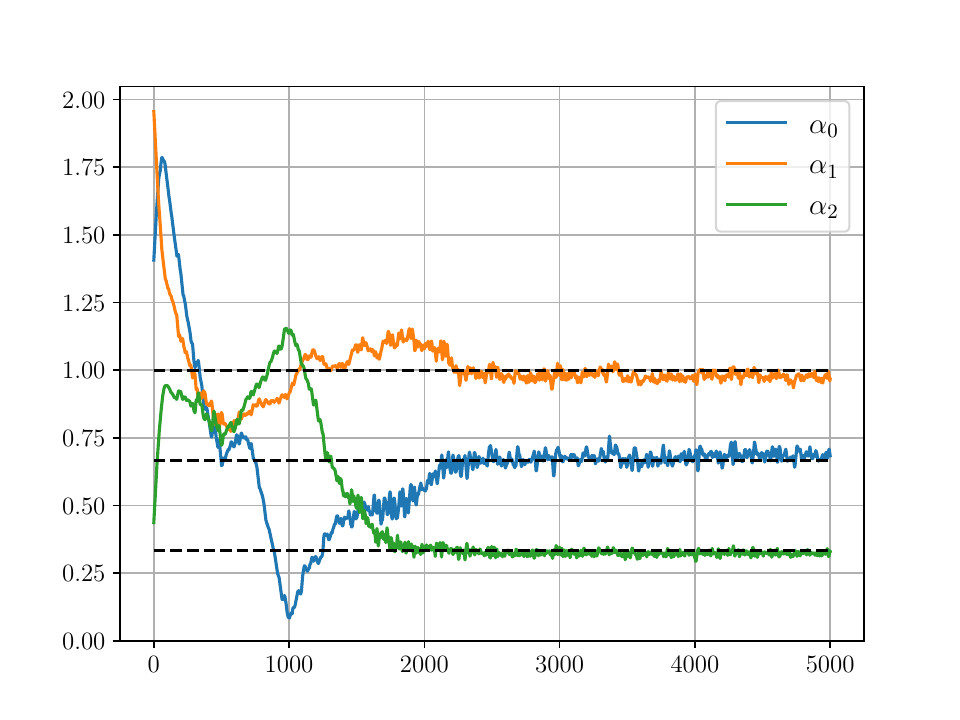}
        \caption{$\alpha$}
        \label{fig:alpha}
    \end{subfigure}
    \caption{Convergence of Objectives and $(\alpha, a)$ for the i.i.d. stationary-ergodic material model.}
    \label{fig:stoch_2D_vor_samples_conv}
\end{figure}

\textbf{Results:} The surrogate modelling dataset is limited to a size 2500 evaluations of $\cG^\dagger$. 
Every random restart includes acquiring the training dataset  of evaluations of $\cG^\dagger$ for accumulation in $\bP^\st$ along with training the surrogate $\cG^\phi$.
For initialization, $(w,a)$ are drawn independently from $\Unif(0.2, 2)$ and $\Unif(0.5, 4)$ respectively, and i.i.d. with respect to each component of $w$ and of $a$. This is on the same order of compute time as periodic homogenization case -- which did not use a surrogate model -- despite:
i) the computational mesh used in $\sF^\dagger$ being $500\times500$ for stochastic homogenization instead of $100\times100$ in the periodic case;
ii) each compute of $\Abar$ requiring 10 evaluations of $\sF^\dagger$ in the Monte-Carlo stochastic estimation. Figure~\ref{fig:stoch_2D_vor_samples_conv} shows the convergence of the algorithm on $(\alpha, a)$.
By using the surrogate construction with an adaptively acquired training set we only compute the expensive $\cG^\dagger$ map 2500 times and instead, in \eqref{eq:objective_stoch} evaluate the inexpensive $\cG^\phi$ map $2.5\cdot10^7$ times.
The relative error on $a$ is 2.98\% and the relative error on $w$ is $0.533\%$.

\subsection{Locally Stationary-Ergodic Microstructures}\label{sec:loc_se_stoch}

In this, our final example, we consider a material field with two approximately stationary-ergodic levels: a microscopic stationary-ergodic property which allows us to apply stochastic homogenization in a small cell locally, and a macroscopic stationary-ergodic property which allows us to estimate relevant statistics over the entire large specimen. In this context, distributional inverse homogenization can infer statistical properties of the microstructure, without needing to estimate or assume specific spatial correlation information, only that the measured bulk properties are separated well enough so as 
to be approximately i.i.d.\,. We now introduce this construction.

Using a copula, $\Xi(\ell, \nu;\S)$, defined in Section~\ref{sec:loc_per}, we consider $\rho(x)$, a spatially correlated random vector field describing local material volume fractions with $\sDir(\alpha)$ marginals. Such a field is constructed via

\begin{subequations}
\begin{align}
    \zeta^{(m)}(x) &\sim\Xi(\ell,\nu; \S)\label{eq:dircCop:1}\\
    Y^{(m)}(x) &= F^{-1}_{\mathrm{Gamma({\alpha_m, 1})}}(\zeta^{(m)}(x)),\label{eq:dircCop:2}\\
    \rho_{(m)}(x) &= \frac{Y^{(m)}(x)}{\sum_{j=1}^M Y^{(j)}(x)}, \quad m=1\hdots,M.\label{eq:dircCop:3}
\end{align}
\end{subequations}
We abbreviate this procedure as 
\begin{align}
    \rho \sim {\Lambda}(\alpha; \nu, \ell, \S).
\end{align}
We can now summarize this procedure informally as follows: \eqref{eq:dircCop:1} sample $M$ independent uniform copula fields; \eqref{eq:dircCop:2} for copula field value at location $x$, evaluate the inverse CDF of the Gamma distribution, this yields $M$ Gamma distributed values at $x$; \eqref{eq:dircCop:3} normalize each $M$ Gamma distributed values by their sum, the resulting values are Dirichlet distributed and, for varying $x$, spatially correlated.
Following this construction, 
$\rho(x)\sim\sDir(\alpha)\; \forall x\in\S$, and $\mathrm{Cov}(\rho(x), \rho(x'))$ is proportional to $(\ell, \nu)$. We note, the inverse CDF of the Gamma distribution at a value $y$, $F^{-1}_{\mathrm{Gamma({\alpha_m, 1})}}(y)$, is the root of $y-F_{\mathrm{Gamma({\alpha_m, 1})}}(x)$, where $F_{\mathrm{Gamma({\alpha_m, 1})}}$ is the CDF; we can approximately compute the inverse CDF via bisection iterations.
Now to construct an approximately locally and globally stationary-ergodic field on a specimen $\S$ and a subdomain $\S_{p}$, sample
\begin{subequations}
\begin{align}
    \rho&\sim\Lambda(\alpha, \ell, \nu;\S)\label{eq:locES:1}\\
    v^{(j)}&\sim\Unif(\S_{p}),\label{eq:locES:2}\\
    r^{(j)} &\sim \sCat(\rho(v^{(j)})),\label{eq:locES:3}\\
    A(x; a, \rho, v^{(1:\Jvor)}) &= \I_2\sum_{j=1}^{\Jvor} \langle a, r^{(j)}\rangle \bb1_{\Q_j}(x).\label{eq:locES:4}
\end{align}
\end{subequations}
We can now summarize this procedure informally as follows: \eqref{eq:locES:1} sample a random vector valued Dirichlet Copula field; \eqref{eq:locES:1} for a select sub-specimen $\S_p$ sample $\Jvor$ seeds; \eqref{eq:locES:3} at each seed location evaluate the Dirichlet copula field, this return a vector of entries which sum to 1, this the local volume fraction, draw a random categorical vector for this volume fraction vector; \eqref{eq:locES:3} assemble the microstructure field using the seeds to specify the Voronoi diagram and the categorical variable to assign material property to each Voronoi cell. In Figure~\ref{fig:DiriCop_loc_per_2D_vor_samples} we a show the three components of $\rho\sim\Lambda(\alpha, \ell,\nu)$ on a $1/20\times1/20$ subset of $\S$, notice the sum of the three fields is equal 1 at all spatial locations. Collections of evaluations of $\rho(x)$ at sparse locations will be Dirichlet distributed. Figure~\ref{fig:Homogenized_field_loc_per_2D_samples} shows the locally stochastic homogenization of random Voronoi diagrams where the volume fractions are specified by the field in Figure~\ref{fig:DiriCop_loc_per_2D_vor_samples}.
\begin{figure}[t] 
    \centering
    % \vspace{1em}
    \hspace{-1em}
        \begin{subfigure}[b]{0.32\linewidth}
        \centering
        \includegraphics[width=\linewidth]{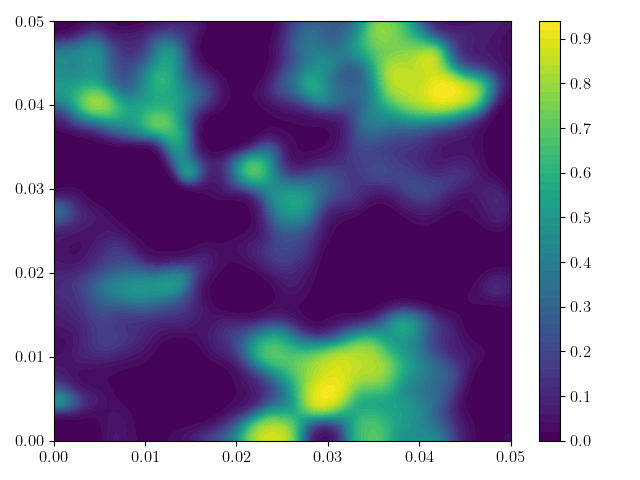}
        \caption{$\rho_1(x)$}
        \label{fig:hist}
    \end{subfigure}
    % \hfill
    % \hspace{-1em}
    \begin{subfigure}[b]{0.32\linewidth}
        \centering
        \includegraphics[width=\linewidth]{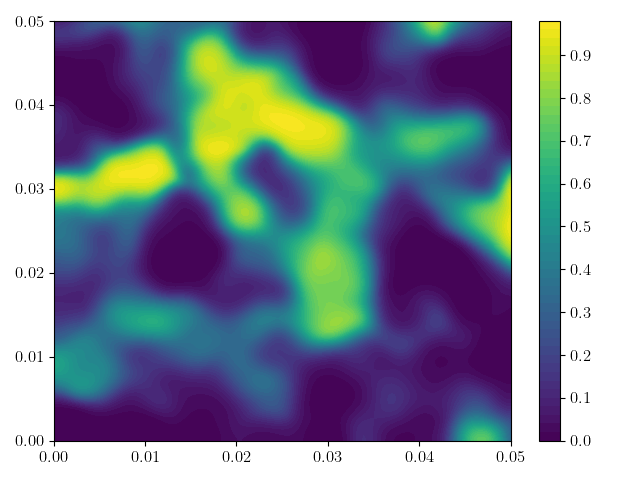}
        \caption{$\rho_2(x)$}
        \label{fig:loss}
    \end{subfigure}
    \begin{subfigure}[b]{0.32\linewidth}
        \centering
        \includegraphics[width=\linewidth]{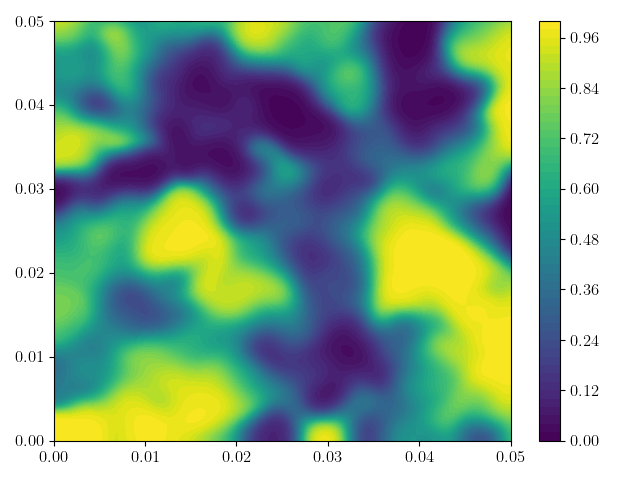}
        \caption{$\rho_3(x)$}
        \label{fig:alpha}
    \end{subfigure}
    \caption{Dirichlet marginal copula, $\Lambda(\alpha, \ell, \nu; \S)$,  components on a small $1/20\times 1/20$ subset of $\S$. Notice, $\sum_{m=1}^M\rho_m(x)=1$. This sub-specimen contains $100\times 100$ sub-specimen cells, $\Q_\sfn$, which appear in Figure~\ref{fig:Stochi.i.d._2D_rnd_A11_samples}.}
    \label{fig:DiriCop_loc_per_2D_vor_samples}
\end{figure}
\begin{figure}[t] 
    \centering
    \vspace{1em}
    \hspace{-1em}
        \begin{subfigure}[b]{0.34\linewidth}
        \centering
        \includegraphics[width=\linewidth]{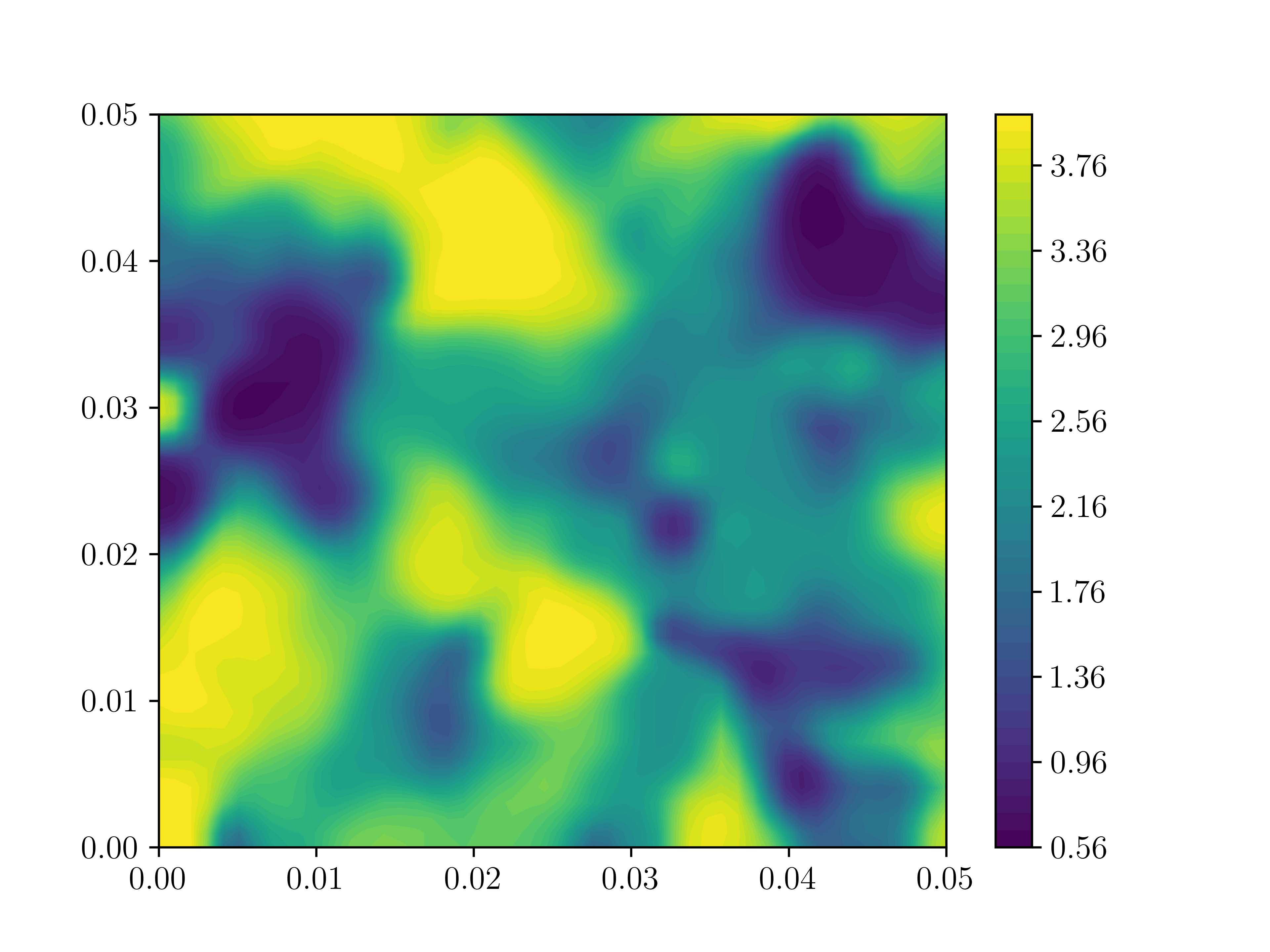}
        \caption{$\Abar_{11}(x)$}
        \label{fig:hist}
    \end{subfigure}
    % \hfill
    \hspace{-1em}
    \begin{subfigure}[b]{0.34\linewidth}
        \centering
        \includegraphics[width=\linewidth]{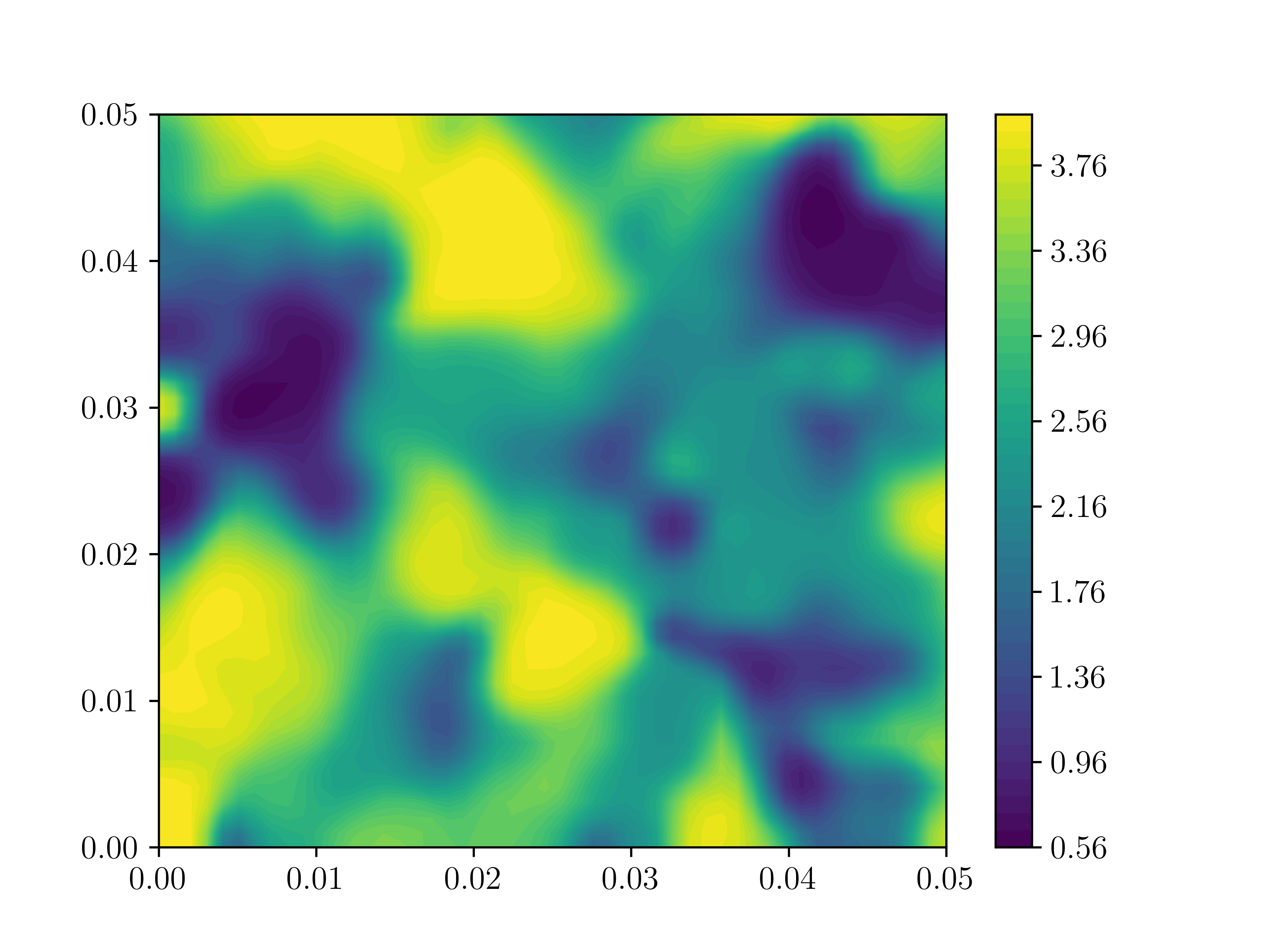}
        \caption{$\Abar_{22}(x)$}
        \label{fig:loss}
    \end{subfigure}
    \hspace{-1em}
    \begin{subfigure}[b]{0.34\linewidth}
        \centering
        \includegraphics[width=\linewidth]{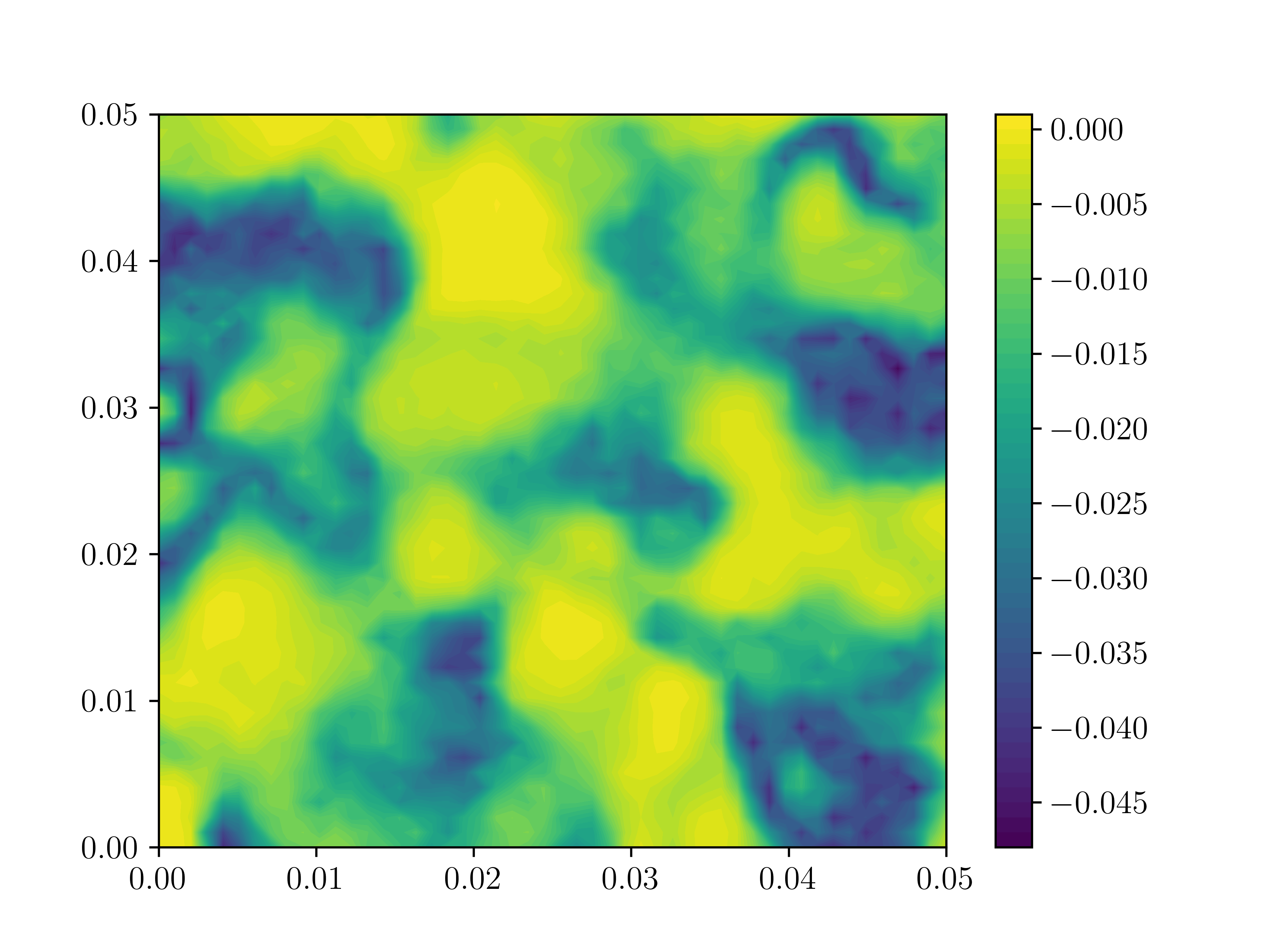}
        \caption{$\Abar_{12}(x)$}
        \label{fig:alpha}
    \end{subfigure}
    \caption{The independent components of the homogenized coefficient field on a subset of a $1/20\times 1/20$ subset of $\S$ where the volume fractions are specified by the random field in Figure~\ref{fig:DiriCop_loc_per_2D_vor_samples}. This sub-specimen contains $100\times 100$ computational cells $\Q_\sfn$.}
    \label{fig:Homogenized_field_loc_per_2D_samples}
\end{figure}
\begin{figure}[t] 
    \centering
    % \vspace{1em}
    \hspace{-1em}
        \begin{subfigure}[b]{0.42\linewidth}
        \centering
        \includegraphics[width=\linewidth]{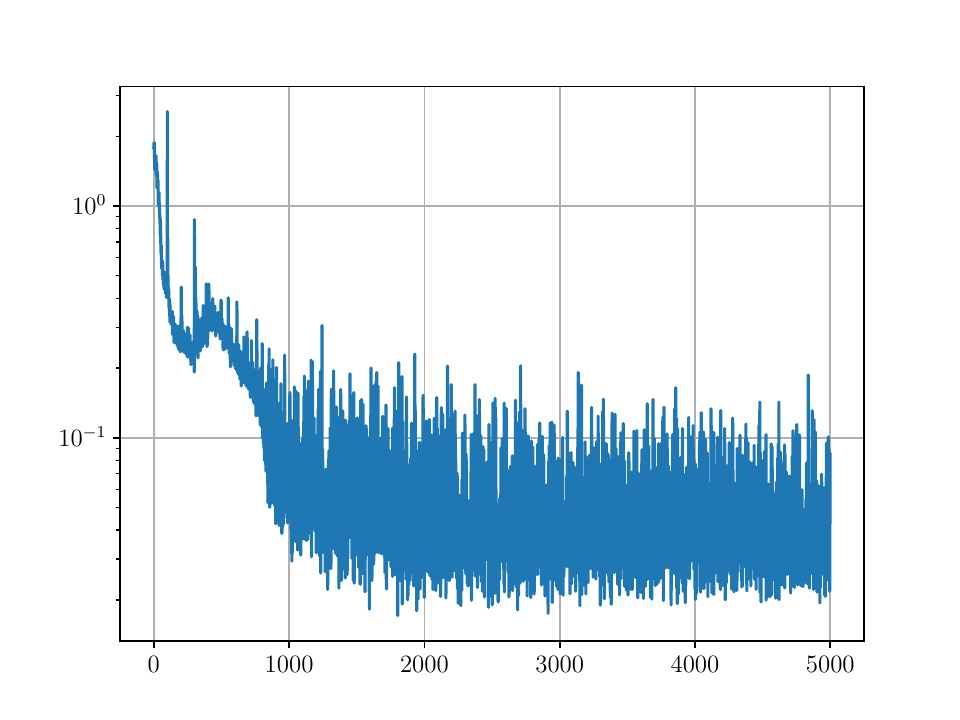}
        \caption{Objective~\eqref{eq:objective_stoch}}
        \label{fig:hist}
    \end{subfigure}
    % \hfill
    % \hspace{-1em}
    \begin{subfigure}[b]{0.42\linewidth}
        \centering
        \includegraphics[width=\linewidth]{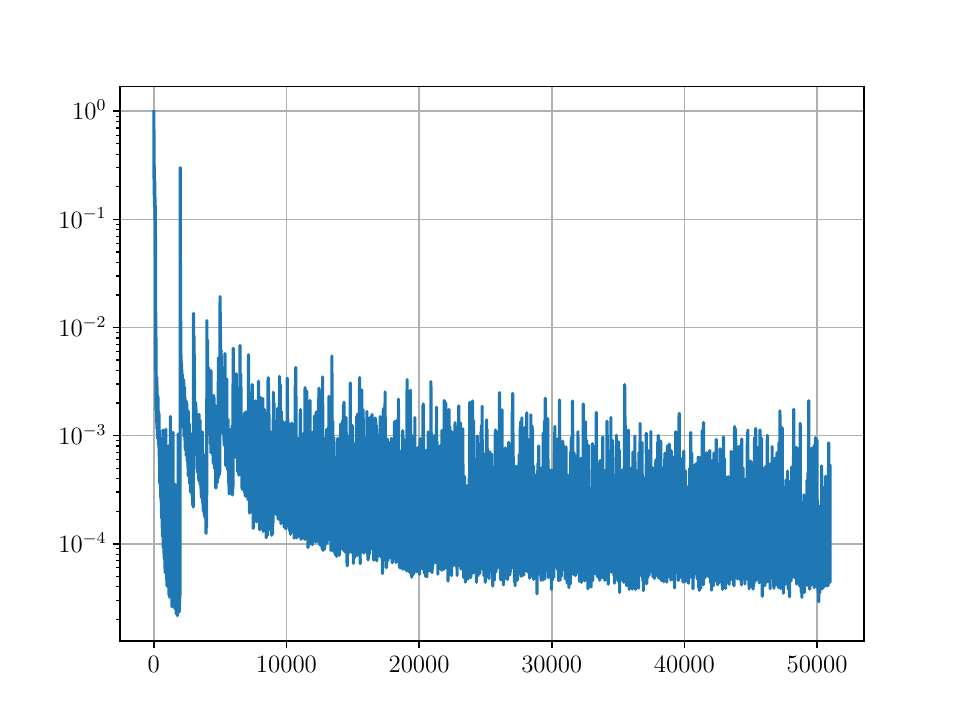}
        \caption{Surrogate Objective~\eqref{eq:cummul_emp_dataset_SH:1}}
        \label{fig:loss}
    \end{subfigure}
    \begin{subfigure}[b]{0.42\linewidth}
        \centering
        \includegraphics[width=\linewidth]{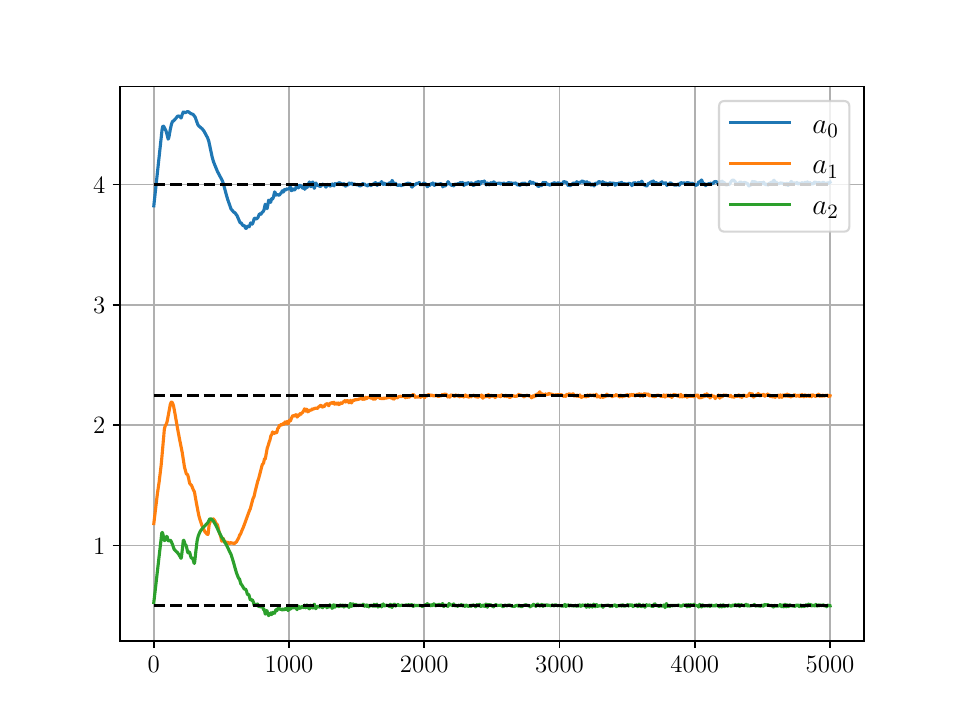}
        \caption{$a$}
        \label{fig:alpha}
    \end{subfigure}
    \begin{subfigure}[b]{0.42\linewidth}
        \centering
        \includegraphics[width=\linewidth]{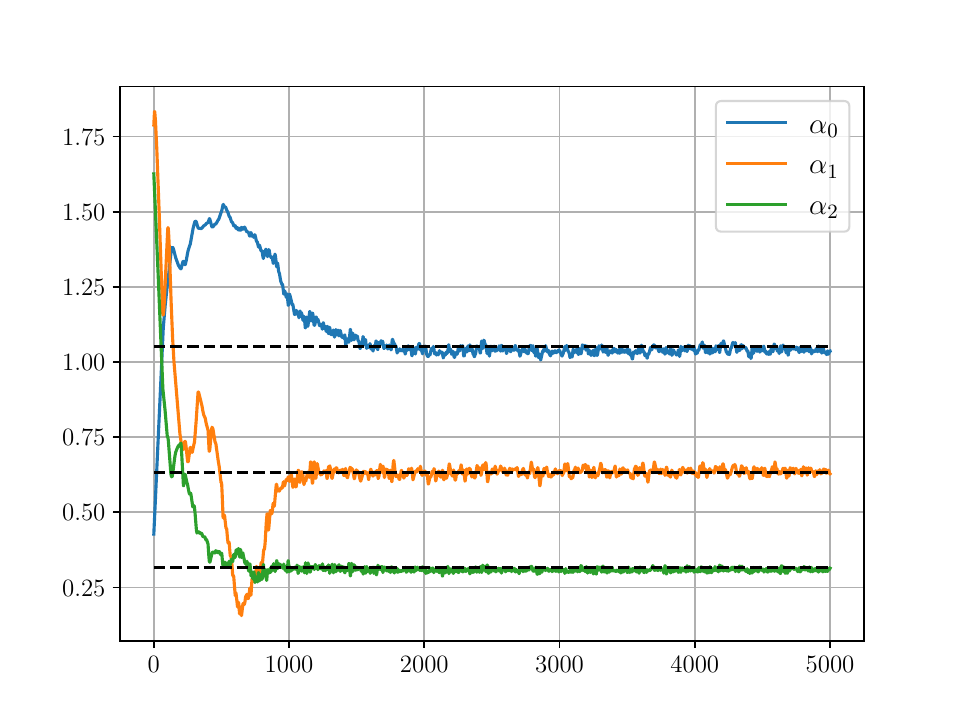}
        \caption{$\alpha$}
        \label{fig:alpha}
    \end{subfigure}
    \caption{Convergence of objectives and $(\alpha, a)$ for the locally stationary-ergodic microstructure model.}
    \label{fig:loc_stat_2D_vor_samples_conv}
\end{figure}

\textbf{Data and Inference:}
The number of Voronoi seeds per sub-specimen $\S_p$ is $3\cdot10^4$. The inter-seed distance in the Voronoi construction with respect the Specimen size is roughly 1 to $6\cdot10^{5}$. We set the number of observed homogenized coefficients to $N=5\cdot10^3$, the noise standard deviation $\sigma_\xi=5\cdot10^{-3}$, $(\ell,\nu)=(1/500, 4)$,
$\Jvor=3\cdot10^{4}$ for $\S_p$ in $\S$.
$N_\mathrm{srgt.}$ Each observed data $\smash{\Abar^{(n)}}$ is generated from a single $\Abar_\sfn$ defined on $\S_p$, hence it is a single-sample Monte-Carlo estimate of $\smash{\Abar^{(n)}}$; in contrast to this,  the i.i.d. model used for inference uses $\Javg=10$ Monte-Carlo samples.

\textbf{Results:} The surrogate modelling dataset is limited to a size 2500 evaluations of $\cG^\dagger$. We use 5 random restarts for different initializations of $(\alpha,a)$, these are drawn independently from $\Unif(0.2, 2)$ and $\Unif(0.5, 4)$ respectively. We report the run achieving the lowest value of~\eqref{eq:objective_stoch} over the last 100 iterations.
Here, again, the expensive map $\cG^\dagger$ is computed 2500 times and $\cG^\phi$ is evaluated $2.5\cdot10^7$ times.
The relative error on $a$ is 0.46\% and the relative error on $w$ is $1.32\%$.

{
\section{Elasticity}\label{sec:elast}
In this section we demonstrate that the ideas developed in the paper can be
applied to the complex, and computationally demanding, setting of elasticity. We apply distributional inversion to the case of elastic homogenization on an i.i.d. periodic microstructure similar to that of Section~\ref{sec:2D_per_hom}. We begin with a brief introduction to homogenization for elasticity, followed by a description of the inference task, a discussion of the results and a summary.

The elasticity problem is an elliptic vector-valued counterpart of the scalar elliptic \eqref{eq:darcy_mscale}. Given Dirichlet boundary conditions, it has the form
\begin{subequations}\label{eq:elastic_mscale}
\begin{alignat}{2}
    -\nabla\cdot({\sfC^\varepsilon}\,\nabla^\mathrm{sym} u^\varepsilon) &= f,\; &&\mathrm{for}\; x\in D\subset\bR^2,\\
    u^\varepsilon(x) &= 0,\; &&\mathrm{for}\; x\in\partial D.
\end{alignat}
\end{subequations}
where $u=(u_1, u_2)^\top$ is two-dimensional displacement,  and the symmetric gradient $\nabla^\mathrm{sym} u = \frac{1}{2}(\nabla u + (\nabla u)^\top)$.
Here $\sfC^\varepsilon: D \to \mathcal{V}^4$ is a fourth order tensor field.
Analogously to the scalar elliptic case, $\sfC^\varepsilon(x)=\sfC(x/\varepsilon)$ for small $\varepsilon$. Also analogously to the scalar elliptic case, the homogenization task is to identify a constant 4-tensor $\overline{\sfC}$ for which the solution to \eqref{eq:elastic_mscale}, with $\overline{{\sfC}}$ in place of $\sfC^\varepsilon$, approximates \eqref{eq:elastic_mscale}. In this work we will focus on the plane strain 2D variant of elasticity~\cite{cioranescu1999introduction}. For a given periodic 4-tensor material field, $\sfC$, one first solves the vector-valued corrector problems for $\chi^{kl}\in H^1_\per (\bT^2; \bR^2), kl\in\{1,2\}$
which weakly satisfy
\begin{subequations}\label{eq:elast_corrector}
\begin{align}
    % \frac{\partial}{\partial y_j}\left(\sfC_{ijmn}(y)\left(\left(\nabla^{\mathrm{sym}}\chi^{kl}\right)_{mn} + \frac{1}{2}(\bdelta_{mk}\bdelta_{nl} + \bdelta_{ml}\bdelta_{nk})\right)\right)&=0,\\
    -\frac{\partial}{\partial y_j}\left(\sfC_{ijop}(y)\frac{\partial\chi^{kl}_o}{\partial y_p}\right)&=\frac{\partial\sfC_{ijkl}}{\partial y_j},\\
    \int_{\bT^2}\chi^{kl}(y)\md y &= 0.
\end{align}
\end{subequations}
% where $\bdelta$ is the Kronecker delta.
The homogenized coefficient is then computed from the formula
\begin{align}
    \overline{\sfC}_{ijkl} = \int_{\bT^2} \left(\sfC_{ijkl}(y) + \sfC_{ijop}(y)\frac{\partial\chi^{kl}_o}{\partial y_p}\right)\md y.
\end{align}
The symmetries of $\sfCbar$ imply 6 independent tensor entries.

We specify the distributional inverse homogenization problem as follows. The constituent material properties $(a_1, a_2, a_3)$ and Voronoi weights $(w_1, w_2, w_3)$ parameterise a spatially varying Young's modulus 
\begin{align}\label{eq:elast_Young}
    E(y; a, w) = \sum_{m=1}^3 a_m\bb1_{\Q_m(v^{(1:M)};w)}(y),
\end{align} where the Voronoi cells, $\Q$, are defined as in~\eqref{eq:per_vor},~\eqref{eq:model_i.i.d._per}. The periodic field $E(y)$ and constant Poisson ratio (assumed known and fixed at 0.3 for data generation and inference) parametrise the periodic 4-tensor field $\sfC(y)$. To use automatic differentiation in minimizing~\eqref{eq:objective_per}, we apply the continuous relaxation of~\eqref{eq:cont_relax_A_crystalLag} to~\eqref{eq:elast_Young}. The inference task is to identify $(a,w)$ from collections of homogenized coefficients $\overline{\sfC}$ generated via randomisation of the Voronoi seeds.

\textbf{Data and Inference:}
We use $N=5\cdot10^3$ data samples. We set $a^\dagger=(1, 2, 4),$ $w^\dagger=(0,0.05, 0.1).$ To generate data we set $\tau$ to $10^{-20}$ to have sharp numerically discontinuous Voronoi cells. We do not include noise in this experiment.
We fix the computational mesh to be $100\times100$ when representing $\sfC(y)$ and in solving for the vector-valued corrector functions with finite elements in~\eqref{eq:elast_corrector}.

\textbf{Results:}
For inference we set $\tau$ to $10^{-3}$ to allow gradients to flow through the Voronoi construction and estimate $w$. 
Figure~\ref{fig:Elast_conv} shows the convergence of the objective, and $(w, a)$ against $(w^\dagger, a^\dagger).$ 
We run $5$ random restarts with $(w,a)$ independently drawn from $\Unif(0.01, 0.5)$ and $\Unif(0.5, 3)$ respectively, and independently component-wise, and extend the optimization of the initialization achieving lowest Objective value by an additional 1000 iterations.
The relative error on the recovered $a$ is $2.57\%$ and for $w$  it is $6.23\%$.

\begin{figure}[t] 
    \centering
    \begin{subfigure}[b]{0.32\linewidth}
        \centering
        \includegraphics[width=\linewidth]{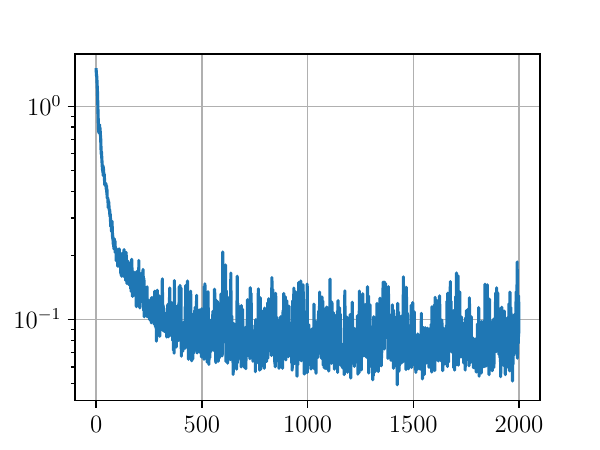}
        \caption{Objective}
        \label{fig:alpha}
    \end{subfigure}
        \begin{subfigure}[b]{0.32\linewidth}
        \centering
        \includegraphics[width=\linewidth]{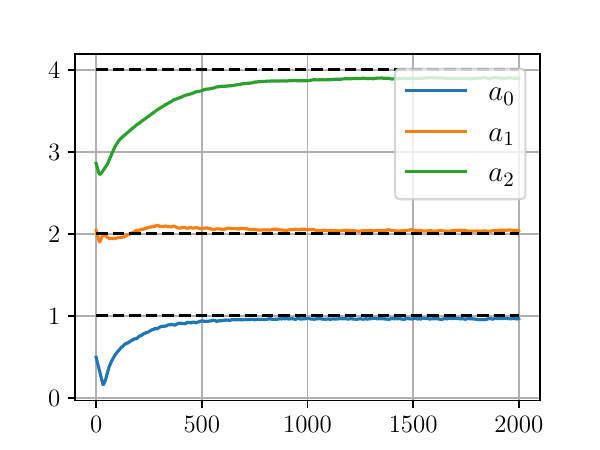}
        \caption{$a$}
        \label{fig:hist}
    \end{subfigure}
    \begin{subfigure}[b]{0.32\linewidth}
        \centering
        \includegraphics[width=\linewidth]{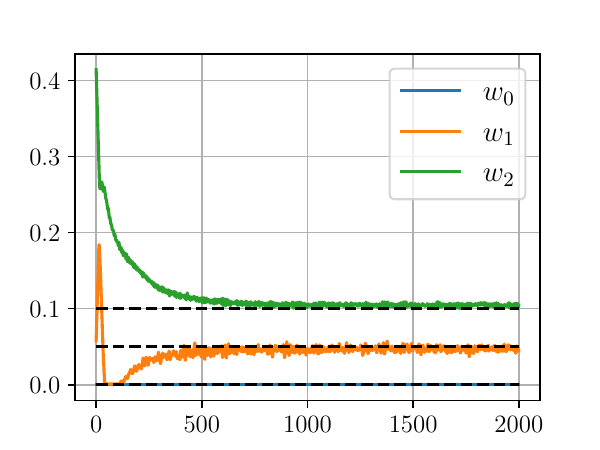}
        \caption{$w$}
        \label{fig:loss}
    \end{subfigure}
    \caption{Convergence of \eqref{eq:objective_per} and $(w, a)$  for the distributional inverse elastic homogenizaton.}
    \label{fig:Elast_conv}
\end{figure}

\textbf{Summary}
The work in this section shows that the ideas developed for scalar
thermo-mechanical problems earlier in the paper do extend
to the mechanical problems of elasticity. The concept that it is
possible to invert for the statistics of microstructure, and that
it is feasible to do so by means of distributional matching, have been shown to apply. The increased computational burden for the vector-valued
elasticity problems suggests that use of surrogate modelling, to evaluate the
microstructure to macroscopic response map, has great potential in this area;
the paper \cite{akyildiz2025efficient} provides a methodology that is
well-adapted to this task; and the paper \cite{Bhattacharyaetal2024} provides
theory concerning the ability to approximate the microstructure to macroscopic response map using neural operators.
}

\section{Conclusions and Future Work}\label{sec:conclusion}

In this paper we present a new approach for investigating the microscopic properties of materials by way of measurement of large collections of bulk material properties. Using this approach it is possible to leverage indirect, homogenized, observed quantities to characterize microstructure at a distributional level, both in the case of periodic and stochastic homogenization. In practice one may think of collecting the
bulk material property data by making measures at different, decorrelated, points in the same
specimen. Algorithmically we model this by using spatial copula. Interesting directions for future work include generalization to two and three dimensional elasticity; a starting point would be to first consider similar scalar-valued microstructures which results in 4-tensor homogenized coefficients. We anticipate that going from measured symmetric $2\times 2$ homogenized coefficient, $\Abar$, which has three independent components, to measuring up to 21 independent components for fourth order elasticity tensors, $\overline{\mathsf{C}}$, may provide the help needed to tackle this broader class of problems. The direction we describe and deploy in order to accelerate computations, using surrogate modelling, 
is also applicable to elasticity.

\addtocontents{toc}{\protect\setcounter{tocdepth}{1}}

\appendix
\section{A First Probability Distribution on the Simplex}\label{sec:proof_claim_ident}
In this section we prove Theorem~\ref{claim:exist_identif}, and also extend it to the
case $M>2$. In 1D it is more convenient to work with reciprocals of material
properties and so we define $c_m=a_m^{-1}$, $\overline{c}=\bigl(\overline{a}\bigr)^{-1}$ 
and write $c=(c_1, \hdots, c_M)$. Then \eqref{eq:to} may be written
\begin{align}
    \overline{c}=\langle c, \rho\rangle.
\end{align}
Using the fact that $\rho_M=1-\sum_{m=1}^M\rho_m$ we can write
\begin{align}\label{eq:vol_frac_1D_Mm1}
\cbar = c_M\sum_{m=1}^{M-1}(c_M^{-1}c_m-1)\rho_m+c_M.
\end{align}
Thus we have established the explicit dependence of $\abar$ (equivalently $\cbar$) on
the $M$ material parameters $c$ and the $M-1$ volume fraction parameters $(\rho_1,\cdots, \rho_{M-1})$. We now explore a volume-fraction distribution on the simplex and the map taking its parameters to the distribution of homogenized coefficients. We need a preliminary technical result. To state the result,
note that any uniform distribution on an interval may be written as
$\Unif(b-l,b+l)$ and we refer to $b$ as the \emph{shift} and $2l$ as the \emph{width}
of the probability distribution.\\

\begin{lemma}\label{lem:identif_sums_uniforms}
    The convolution of $n$ uniform distributions uniquely identifies the set of all $n$ widths and the sum of the $n$ shifts of the individual distributions.
\end{lemma}
\begin{proof}\label{proof:lem:identif_sums_uniforms}
        The Fourier transform (characteristic function) of $f:=\Unif(b-l,b+l)$ is given by~\cite{bradley2002distribution}~\footnote{Within this proof we use the symbol $\xi$ to denote wave number. This is not to be confused with its use as the noise variable in the main body of the paper.}
    \begin{align*}
        \widehat{f}(\xi) &= \int_{b-l}^{b+l}\frac{1}{2l}e^{-i2\pi\xi t}\md t=e^{-i2\pi b\xi}\mathrm{sinc}(2\pi l\xi)\nonumber,
    \end{align*}
    Now let $f_j := \Unif(b_j-l_j,b_j+l_j)$. Then
    \begin{align*}
        (f_1*\cdots * f_n\widehat{)}(\xi) = \prod_{j=1}^n e^{-i2\pi b_j\xi} \sinc(2\pi l_j\xi).
    \end{align*}
    On $L^2({\Rb})$, the Fourier transform is bijective, and as $e^{-i2\pi\xi}$ has no complex zeros,
    \begin{subequations}
        \label{eq:labelit}
    \begin{align}
        \prod_{j=1}^n e^{-i2\pi b_j\xi} \sinc(2\pi l_j\xi) &= \prod_{j=1}^n e^{-i2\pi b_j'\xi} \sinc(2\pi l_j'\xi)\\
        \implies \left\{\xi\in\Rb:\prod_{j=1}^n \sinc(2\pi l_j\xi)=0\right\} &=  \left\{\xi\in\Rb:\prod_{j=1}^n\sinc(2\pi l_j'\xi)=0\right\}.
    \end{align}
    \end{subequations}
    The zeros of $\prod_{j=1}^n \sinc(2\pi l_j\xi)$ are $\cup_{j=1}^n \{k/2l_j\}_{k\in\bN\backslash\{0\}}$. These uniquely identify $\cup_{j=1}^n\{l_j'\}=\cup_{j=1}^n\{l_j\}$. \\
    
    Since $l_j=l_j'$, the identity (\ref{eq:labelit}a) implies that
    \begin{align*}
    \prod_{j=1}^n e^{-i2\pi b_j\xi} &=  \prod_{j=1}^n e^{-i2\pi b_j'\xi}.
    \end{align*}
    Hence, $\exp(-i2\pi\xi\sum_{j=1}^n {b_j}) = \exp(-i2\pi\xi\sum_{j=1}^n {b_j'}), \forall\xi\in\Rb$. Thus we can identify $\sum_{j=1}^n {b_j}$.
\end{proof}

Now we address Theorem~\ref{claim:exist_identif}, and generalization to $M>2.$ Define, for $m=1, \cdots, M-1$,
\begin{align}
    \bP_{\rho_m} =\Unif(b_j-l_j, b_j+l_j).\nonumber
\end{align}
We assume that the $l_j>0$ so that the uniform distribution has support with positive Lebesgue
measure. To ensure positivity of the $\rho_j$ drawn from this distribution we require 
$b_j- l_j\geq0.$ We now construct a probability measure by taking the independent product of
these measures, to obtain
$$\bP_{\rho_{_{1:M-1}}} = \prod_{j=1}^{M-1}\,\bP_{\rho_j}.$$
It remains to specify the choice of $\rho_M$ in a manner which ensure $\rho \in \bsDelta_{M-1}$,
the probability simplex. To this end we add the assumption $\sum_{m=1}^{M-1}b_j+l_j\leq1$
and define
\begin{align}
    f(\rho_{_{1:M-1}})=1-\sum_{j=1}^{M-1}\rho_j.\nonumber
\end{align}
Then we set $\bP_\rho=(\Id.(\Bigcdot), f(\Bigcdot))_\#\bP_{\rho_{_{1:M-1}}}$, 
thereby creating a distribution on $\bsDelta_{m-1}.$  Define  $\theta=(\sum_{m=1}^{M-1}b_m, \{l_m\}_{m=1}^{M-1}).$
Recall formula \eqref{eq:stnal} for the homogenized diffusion coefficient.
The distribution over these coefficients implied by $\bP_\rho$ is thus
\begin{align}\label{eq:cbar_first_dist_simplex}
    \bP_{\abar}&=  \sF(\,\cdot\,;a)_\# \bP_\rho(\theta).\nonumber
\end{align}
Although the homogenization map, $\rho\mapsto\sF(\rho;a)$, is not injective,
we have the following form of distributional injectivity:

\begin{theorem}\label{th:identif_unif_simp}
Fix $a \in (0,\infty)^m.$ Then map $\theta\mapsto(\sF(\,\cdot\,; a))_\#\bP_\rho(\theta)$ is injective.
\end{theorem}

\begin{proof}\label{proof:th:identif_unif_simp}
The map ${\abar}\mapsto\cbar$ is bijective, so it suffices to study pushforward
     of $\bP_\rho(\theta)$ onto $\cbar$. Note from~\eqref{eq:vol_frac_1D_Mm1}, $\cbar c_M^{-1}-1=\sum_{m=1}^{M-1}(c_M^{-1}c_m-1)\rho_m$. Hence, we can re-write $\bP_{\cbar}=(\langle\,\cdot\,, c\rangle)_\#\bP_\rho(\theta)$ 
     as a convolution of scaled uniforms on $\rho_{1:M-1}$.
Using Lemma~\ref{lem:identif_sums_uniforms}, the distribution of $\cbar$ identifies the widths of all uniforms and the sum of centers of the uniforms on $\rho_{1:M-1}$. 
    Thus, the distributions of $\abar$ identifies $\theta$. 
    The distribution of $\rho_M$ is given by the distribution of $\sum_{m=1}^{M-1} \rho_m-1$.
\end{proof}

\section{Theorem~\ref{claim:1DDirichlet_identifiability}: Dirichlet Limit Identifiability}\label{sec:1DDirichlet_identifiability}

We find it useful to think of the  $M$-dimensional Dirichlet parameters $\alpha$ as $\alpha:=\beta\gamma$ where $\gamma\in\bsDelta_{M-1}$ and $\beta\in(0, \infty)$, noting that the condition $\sum_{m=1}^M\gamma_m=1$ implies that we still only have $M$ degrees-of-freedom. Recall the standard notation $\sDir(\alpha)$ for a Dirichlet distribution, introduced in \eqref{eq:Dirich_density}, and the definition $\sDir(\gamma, \beta):=\sDir(\gamma\beta)$, for $\gamma \in \bsDelta_{M-1}$, 
in the text that follows
 \eqref{eq:Dirich_density}. For our purposes it is useful to  extend the definition of the Dirichlet distribution to $\beta\in[0,\infty].$
To do this, we prove the existence of the limits $\beta=0, \beta=\infty$.\\

\begin{lemma}\label{lem:limits_dirichlet} Let $\gamma\in\bsDelta_{M-1}$ and $\beta \in (0,\infty).$
    Define $\sDir_0(\gamma):=\sum_{m=1}^M\gamma_m\delta_{e^m}$ and $\sDir_\infty(\gamma):=\delta_{\gamma}$.
    Then, $\sDir(\gamma, \beta) \xrightarrow[\beta\rightarrow0]{} \sDir_0(\gamma)$ and $\sDir(\gamma, \beta)\xrightarrow[\beta\rightarrow\infty]{}\sDir_\infty(\gamma)$. 
\end{lemma}
\begin{proof}\label{proof:lem:limits_dirichlet}
   Let $\alpha=\gamma \beta.$ We start with the first limit. 
   Let $\rho \sim \sDir(\alpha).$ Then
   $\bE_{\sDir(\alpha)}[\rho_k]=\alpha_k/\sum_{m=1}^M\alpha_m.$ Thus i): $\bE_{\sDir(\beta, \gamma)}[\rho_k]=\gamma_k,$ $\forall \beta>0$.
    Also, ii): from~\eqref{eq:Dirich_density} we can see
   that for any $A\subset\bsDelta_{M-1}$ such that $e^m\notin A$ for any $m\in\{1, \hdots, M\}$, then, $$\lim_{\beta\rightarrow 0}\int_A\sDir(\rho;\gamma, \beta)\md\rho=0.$$ Thus, probability
   % the density of the Dirichlet 
   concentrates on the vertices of the simplex as $\beta\rightarrow0.$
Points i) and ii) can be used to conclude convergence in distribution: $\sDir(\beta\gamma)\rightarrow\sum_{m=1}^M\gamma_m\delta_{e^m}$ as $\beta$ goes to zero.

   We now prove the second limit.
    To sample from $\sDir(\beta\gamma), \beta>0,$ first sample $x_j\sim\sfGamma(\beta\gamma_j, 1)$ where 
   $$\sfGamma(x; \beta\gamma_j, 1)=\frac{x^{\beta\gamma_j-1}e^{-x}}{\Gamma(\beta\gamma_j)}.$$ Next set
   \begin{align*}
       \rho_j=\frac{x_j}{\sum_{k=1}^M x_k},
   \end{align*}
   thus producing $\rho\sim\sDir(\beta\gamma).$
   For $\beta\geq1$ the density $\sfGamma(x;\beta\gamma_j,1)$ has unique mode at $x^\star=(\beta\gamma_j-1)$ and approaches a Gaussian with mean and variance
   $\beta\gamma_j;$ this may be shown using the central limit theorem.
   As a consequence, as $\beta$ increases, the Dirichlet density behaves like
   \begin{align*}
       \rho_j = \frac{\beta\gamma_j+\mathcal{O}(\sqrt{\beta})}{\sum_{k=1}^M \beta\gamma_k+\mathcal{O}(\sqrt{\beta})}.
   \end{align*}
   Applying l'H\^opital reveals, for $\beta\rightarrow\infty$, $\rho_j=\gamma_j$. These calculations can be used to establish convergence in distribution: $\sfDir(\beta\gamma)\rightarrow\delta_\gamma.$
\end{proof}

We now define the 1D distributional homogenization map 
\begin{align}
    \bP_{\abar}(\gamma
    , \beta, a) = \sF(\,\cdot\,; a)_\#\sDir(\gamma, \beta).\nonumber
\end{align}
The following proves Theorem~\ref{claim:1DDirichlet_identifiability}.
\vspace{1em}

\begin{lemma}\label{lem:inj_at_beta0}
    The map $(\gamma, \beta, a)\mapsto \sF(\,\cdot\,; a)_\#\sDir(\gamma,\beta)$ is injective at $\beta=0$, and is not injective at $\beta=\infty$.
\end{lemma}
\begin{proof}\label{proof:inj_at_beta0}
    We start with $\beta=0.$ We work with $\abar^{-1}=\cbar$ and $a_m^{-1}=c_m$.
    Note that the distribution $(\langle c,\,\cdot\,\rangle)_\#\sDir_0(\gamma)=\sum_{m=1}^M \gamma_m\delta_{c_me^m}$. This distribution identifies $(\gamma,\beta, c)$ which in turn identifies $(\gamma,\beta, a)$. 

    For $\beta=\infty$, $(\langle c,\,\cdot\,\rangle)_\#\sDir_\infty(\gamma)=\delta_{\langle c, \gamma\rangle}$, which does not identify $(\gamma,c)$ and so does not identify $(\gamma, \beta, a)$.
\end{proof}

\section*{Acknowledgments}
The authors are thankful to John R Willis for conversation on the topic of homogenization which helped improve this paper.

\bibliographystyle{elsarticle-num} 
\bibliography{biblio}

\end{document}